\documentclass[9pt, shorpaper,twoside,web]{ieeecolor}
\usepackage{etoolbox}
\usepackage{threeparttable}
\makeatletter
\@ifundefined{color@begingroup}%
{\let\color@begingroup\relax
	\let\color@endgroup\relax}{}%
\def\fix@ieeecolor@hbox#1{%
	\hbox{\color@begingroup#1\color@endgroup}}
\patchcmd\@makecaption{\hbox}{\fix@ieeecolor@hbox}{}{\FAILED}
\patchcmd\@makecaption{\hbox}{\fix@ieeecolor@hbox}{}{\FAILED}
\usepackage{generic}
\usepackage{cite}
\usepackage{amsmath,amssymb,amsfonts}
\usepackage{algorithmic}
\usepackage{multirow}
\usepackage{booktabs}
\usepackage{setspace}
\usepackage{graphicx}
\usepackage{textcomp}
\usepackage{color}
\usepackage{enumerate}
\usepackage{algorithm}
\usepackage{algorithmic}
\usepackage{xcolor}
\usepackage{bm}
\usepackage{amstext}
\usepackage{stfloats}
\usepackage{subcaption}

\usepackage{amsthm}
\newtheorem{remark}{Remark}{}
\newtheorem{theorem}{Theorem}

\newtheorem{lemma}{Lemma}
\newtheorem{definition}{Definition}

\newtheorem{assumption}{Assumption}
\newtheorem{corollary}{Corollary}

\usepackage{xpatch}
\usepackage{mathrsfs}
\usepackage{diagbox}
\usepackage{cleveref}

\allowdisplaybreaks[4]
\usepackage{cases}
\usepackage{subeqnarray}
\def\BibTeX{{\rm B\kern-.05em{\sc i\kern-.025em b}\kern-.08em
		T\kern-.1667em\lower.7ex\hbox{E}\kern-.125emX}}

\newcounter{example1}
\newenvironment{example1}[1][]{\refstepcounter{example}\par
	\textbf{Case~\theexample #1} \rmfamily}

\newcommand{\caref}[1]{\textbf{Case~\ref{#1}}}

\crefrangelabelformat{example}%
{\textbf{#3#1#4}{--}\textbf{#5\crefstripprefix{#1}{#2}#6}}
\creflabelformat{example}%
{\textbf{#2#1#3}}

\crefrangelabelformat{figure}%
{{#3#1#4}{--}{#5\crefstripprefix{#1}{#2}#6}}
\creflabelformat{figure}%
{\textbf{#2#1#3}}

\crefrangelabelformat{theorem}%
{{#3#1#4}{--}{#5\crefstripprefix{#1}{#2}#6}}
\creflabelformat{theorem}%
{\textbf{#2#1#3}}

\crefname{example}{\textbf{Case}\rmfamily}{{\textbf{Cases}}\rmfamily}

\newcommand{\colorbibs}[2][blue]%
{%
	\DeclareBibliographyCategory{ColoredBiblist#1}%
	\addtocategory{ColoredBiblist#1}{#2}%
	\AtEveryBibitem{\ifcategory{ColoredBiblist#1}{\color{#1}}{}}
}

\makeatletter
\ExplSyntaxOn
\cs_new:Npn \bibColoredItems #1#2
{
	\clist_map_inline:nn {#2} { \cs_new:cpn {bib@colored@##1} {#1} } 
}
\ExplSyntaxOff

\newcommand\bib@setcolor[1]{%
	\ifcsname bib@colored@#1\endcsname
	\expandafter\color\expandafter{\csname bib@colored@#1\endcsname}
	\else
	\normalcolor
	\fi
}

\xpatchcmd\@bibitem
{\item}
{\bib@setcolor{#1}\item}
{}{\PatchFailed}

\xpatchcmd\@lbibitem
{\item}
{\bib@setcolor{#2}\item}
{}{\PatchFailed}
\makeatother
\begin{document}
\title{Resilient Control for Networked Switched Systems With/Without ACK: An Active Quantized Framework }

\author{Rui Zhao,  \IEEEmembership{Member, IEEE},
	Zhiqiang Zuo, \IEEEmembership{Senior Member, IEEE},
	Yijing Wang, \\  Wentao Zhang, \IEEEmembership{Member, IEEE}, 	and Yang Shi, \IEEEmembership{Fellow, IEEE}
	\thanks{This work was supported by the National Natural Science Foundation of China (grants 62173243).}
\thanks{R.~Zhao is with the Department of Electrical Engineering, City University of Hong Kong, Hong Kong SAR, China. (e-mail: ruizhao@tju.edu.cn, ruzhao@cityu.edu.hk)}
	\thanks{Z.~Zuo, and Y.~Wang are with the Tianjin Key Laboratory of Intelligent Unmanned Swarm Technology and System, School of Electrical and Information Engineering, Tianjin University, Tianjin 300072, China. 	(e-mail:  zqzuo@tju.edu.cn; yjwang@tju.edu.cn)}
		\thanks{W.~Zhang is with the Continental-NTU Corporate Lab, Nanyang Technological University, 639798, Singapore, and is also with the School of Electrical and Electronic Engineering, Nanyang Technological University, Singapore 639798. (e-mail: wentao.zhang@ntu.edu.sg) }
	\thanks{Y.  Shi is with the Department of Mechanical Engineering, University of Victoria, Victoria, BC V8W 2Y2, Canada. (e-mail: yshi@uvic.ca)}
}

\maketitle

\begin{abstract}
	This paper deals with the quantized control problem for switched systems under denial-of-service (DoS) attack. 
	Considering the system's defensive capability and the computational resources of quantizers and controllers, four control strategies are proposed. These strategies incorporate different combinations of controllers (active and passive), quantizers (centered on the origin or custom-designed), and network configurations (with or without ACK signals). For each strategy, specific update laws for the encoder and decoder are designed to avoid quantization saturation. 
	Furthermore, the uniformity of encoder and decoder operations is maintained by transmitting additional information to the decoder.  To achieve asymptotic stability, sufficient conditions concerning the switching signal and DoS attack constraints are derived by taking into account the asynchronous behaviors. The proposed active quantization strategy with the ACK signal leverages the system model information to compute the control signal in real-time, allowing for possible convergence of the system state despite DoS attack. Additionally, a well-designed switching signal is suggested to further mitigate the impact of DoS attack. A passive quantization strategy with ACK signal is also developed as a simplified version of the active quantized control strategy, providing the foundation for a strategy without ACK signal. Inspired by time-triggered and event-triggered mechanisms, the passive quantization strategy without ACK signal is investigated, with two feasible update laws for the quantizer. Finally, two simulations are conducted to validate the effectiveness of the proposed strategies.
\end{abstract}
\section{Introduction}
Switched systems are composed of a series of subsystems and a logical law governing how these subsystems evolve \cite{DT}. Owing to the superiority in describing complex control systems, they have been widely used in modeling various practical systems \cite{App1,App2}. Moreover, several systems with switching properties can also be  modeled as switched systems, for example, aero-engine \cite{APP4}. 
In the past few decades, with the development of computer technology, the research of networked switched systems has received increasing attention, encompassing both theoretical investigation and applications \cite{NSS}.

On the other hand,  an effective method to overcome the limited bandwidth is quantifying the state information before being transmitted via network. \color{black} The quantization approaches for switched systems can be categorized into static and dynamic ones. 
Ref. \cite{Q_S_rui} studied the event-triggered control for networked switched systems with dynamic quantization. The state is bounded within the range of the quantizer and the quantization parameters are adjusted with the changes in time and switching instants.   Liberzon presented a state quantization scheme for switched systems, where the quantization center is the prediction state \cite{quan_2014}. The key idea is to design the center and radius of quantization regions.  Furthermore, \cite{quan_2018} extended the above approach for switched systems with unknown disturbance to achieve input-to-state stability with exponential decay.

For networked control systems, the security is a challenging yet crucial issue\textcolor{blue}{\cite{R2-1,R2-2}}. It is known that DoS attack deteriorates system performance by blocking the communication channel. Some important results about resilient control are proposed \cite{etDoS-Ass,9903320}.   Moreover, several studies focused on  quantized control for non-switched under DoS attack, including state quantization \cite{Shi2022},  output quantization \cite{output2}, and output and control signal quantization \cite{Liu2022}.
Feng et al. \cite{Quan_DoS} studied the tradeoff between bandwidth and resilience against DoS attack. For a system under quantized control, it has been identified that the imposed hypotheses in the associated DoS attack are stricter because dealing with the quantization error requires an extra effort. 

Recently, the resilient control problem for switched systems under DoS attack has attracted much attention. \cite{Rui2}  revealed the relationship between DoS attack and switching signal. Then a new switching law was proposed to mitigate the negative impacts of periodical DoS attack \cite{Fu2022}. Moreover, the stability problem	for switched systems with unstable mode was investigated in \cite{Wang2022}. In addition, an active control strategy was devised for switched systems under asynchronous DoS attack \cite{Rui1,Rui3}, which improves resilience with the help of predictor and buffer. Such a strategy provides an appropriate control signal during DoS attack intervals. And this makes it possible for the system state to converge even if an attack occurs. 

To the best of our knowledge, little research has been conducted on switched systems under denial-of-service (DoS) attacks incorporating quantized control. Refs. \cite{NAHS} and \cite{quant2} explored quantized control using a dynamic quantizer for switched systems subject to DoS attack. The control signal either becomes zero or retains the last valid signal during DoS attack intervals. 
More recently, \cite{IJRNC} proposed a quantizer centered on the predicted state rather than the origin. Unfortunately, this approach introduces a decision-making mechanism to address asynchronous behavior, which may pose practical implementation challenges.
In the context of switched systems utilizing quantized control, a significant challenge arises from ensuring uniformity of the encoder and decoder due to attack-induced asynchronicity. The interdependent error among DoS attack , asynchronous behavior, and quantization must be effectively managed, as these coupling effect deteriorates system performance. This can manifest as conservative switching signals, stringent DoS attack constraints, and elevated quantization levels.
Therefore, an interesting yet important question arises: Can we design an active quantization control strategy to enhance system performance in the presence of DoS attack? Furthermore, existing approaches typically rely on the status of DoS attack to update the quantizer. As we know, detection mechanisms often rely on acknowledgment (ACK) signals. However, some networks, such as those employing the user datagram protocol (UDP), do not facilitate ACK signal transmission. Consequently, another motivation for this research is to develop an update law for the quantizer under conditions where the status of DoS attack remains unknown.

In this paper, the resilient control problem for switched systems under DoS attack using quantized control is studied. The main contributions are listed as follows.

\textcolor{blue}{1) An active control strategy (\textbf{\textit{Strategy 1}}) is proposed to enhance the defense capacity of switched systems, which consists of an active controller and a quantizer with non-origin center (\Cref{thm1}). 
	It stabilizes switched systems under higher attack frequencies and longer attack durations while requiring lower quantization levels.   The designed center of the quantizer accelerates the zooming-in phase, whereas the active controller mitigates divergence during the zooming-out stage, jointly relaxing the DoS attack frequency and duration bounds. 
	Consequently, this treatment significantly improves defense capacity. A well-structured switching signal is also proposed that effectively counters DoS attacks (\Cref{coro1}).}


2) The strategy of a quantizer centered at the origin is developed to address the quantization problem in scenarios with limited computational resources (\Cref{thm1_ori}). This approach reflects a good compromise between computational burden and overall performance.

3) Two novel quantization schemes  tailored for networks lacking ACK signals are designed based on the passive control strategy with ACK signal (\Cref{thm1_zero}), inspired by time-triggered and event-triggered mechanisms (\Cref{thm_time,thm1_ET}). This research makes the first attempt to investigate the uniformity of quantizers for switched systems under DoS attack in the absence of ACK signals. Notably, this also fills the gap in research even for non-switched systems.

4) Uniformity between the encoder and decoder is maintained during attack-free intervals while accounting for DoS attack and asynchronous behavior. Since the duration of attack-induced asynchronous behavior exceeds one sampling period and switching signals cannot be transmitted during DoS attack intervals, the mode of the encoder is different from the mode of decoder. In this study, the switching instant is transmitted alongside the switching signal. Furthermore, the \textit{zooming-out} instant is also transmitted with the sensor signal when there is no ACK signal. These two enhancements do not significantly increase network load but effectively ensure the uniformity of the quantizer.	 


\begin{table}[h]
	\caption{Notations}	\label{tab1}
	\renewcommand\arraystretch{1.1}
	\setlength{\tabcolsep}{2pt}
	{\color{blue}
		\begin{tabular}{p{40pt} p{205pt}}
			\hline
			Symbol& Definition\\
			\hline
			$\tau_s$& sampling period\\
			$t_k,~t_s,~s_k$ & sampling, switching and successful transmission instants\\
			$x(t)$, $u(t)$& state and control input \\
			$x(t_k)$& state at sampling instants\\
			$\hat x(t)$& estimated state  by predictor \\
			$c_k$, $c_{\overline{k}}$ &   state after being decoded without/with attack\\
			${q}_k$, ${q}_{\overline{k}}$ &  state after being encoded without/with attack  \\ 
			$N_{\max}$ & the maximal number of asynchronous interval\\
			$N$ &   quantization level \\
			$\sigma(k),~\hat{ \sigma}(k)$ & switching signal for subsystems and controller\\
			$E_k^e$, $E_k^d$ & quantization parameters for encoder and decoder\\
			$x_k^{e*}$, $x_k^{d*}$ & quantization centers for encoder and decoder\\
			$\|\cdot\|$ & infinity norm \\
			$\llcorner t_s \lrcorner$ & the latest sampling instant just before $t_s$\\
			$\lceil t_s \rceil $ & the first successful transmission instant after $t_s$ \\
			$	\ulcorner t_s \urcorner $ & the next sampling instant after $t_s$\\
			$n(t,t_0)$& the number of DoS attack off/on switching over $[t_0,t)$\\
			$|\Xi(t,t_0)|$& the length of DoS attack duration during $[t_0,t)$\\
			$|\Theta(t,t_0)|$ & the length of DoS attack-free duration during $[t_0,t)$\\
			$ACK_{t_k}$ & ACK signal at time instant $t_k$\\
			$SY_{t_k} $ & the synchronous stage flag during $[t_{k-1}, t_k)$\\
			$\tau_d$ & dwell time constraint \\
			$e(t)$ & the error between system state and predicted state \\
			\hline
		\end{tabular}
	}
\end{table}
\section{Problem formulation}\label{sec_pf}

Consider a class of networked switched systems
\begin{equation}\label{equ_state}
	\dot{x} (t)  = A_{\sigma(t)} x(t) + B_{\sigma(t)} u(t)
\end{equation}
where $x(t) \in \mathbb{R}^{n_x}$ and $u(t) \in \mathbb{R}^{n_u}$ are the system state and control input. $\sigma(t) \in \mathcal{M} = \{1,2,\cdots,m\}$ is the switching signal with $m$ being the number of subsystems. $A_{\sigma(t)},B_{\sigma(t)}$  are constant matrices with appropriate dimensions. 	Without loss of generality, each pair $(A_p,B_p)$, $p \in \mathcal{M}$ is stabilizable \cite{quan_2014}.

\begin{figure}[b]
	\centering
	\includegraphics[width=0.7\linewidth]{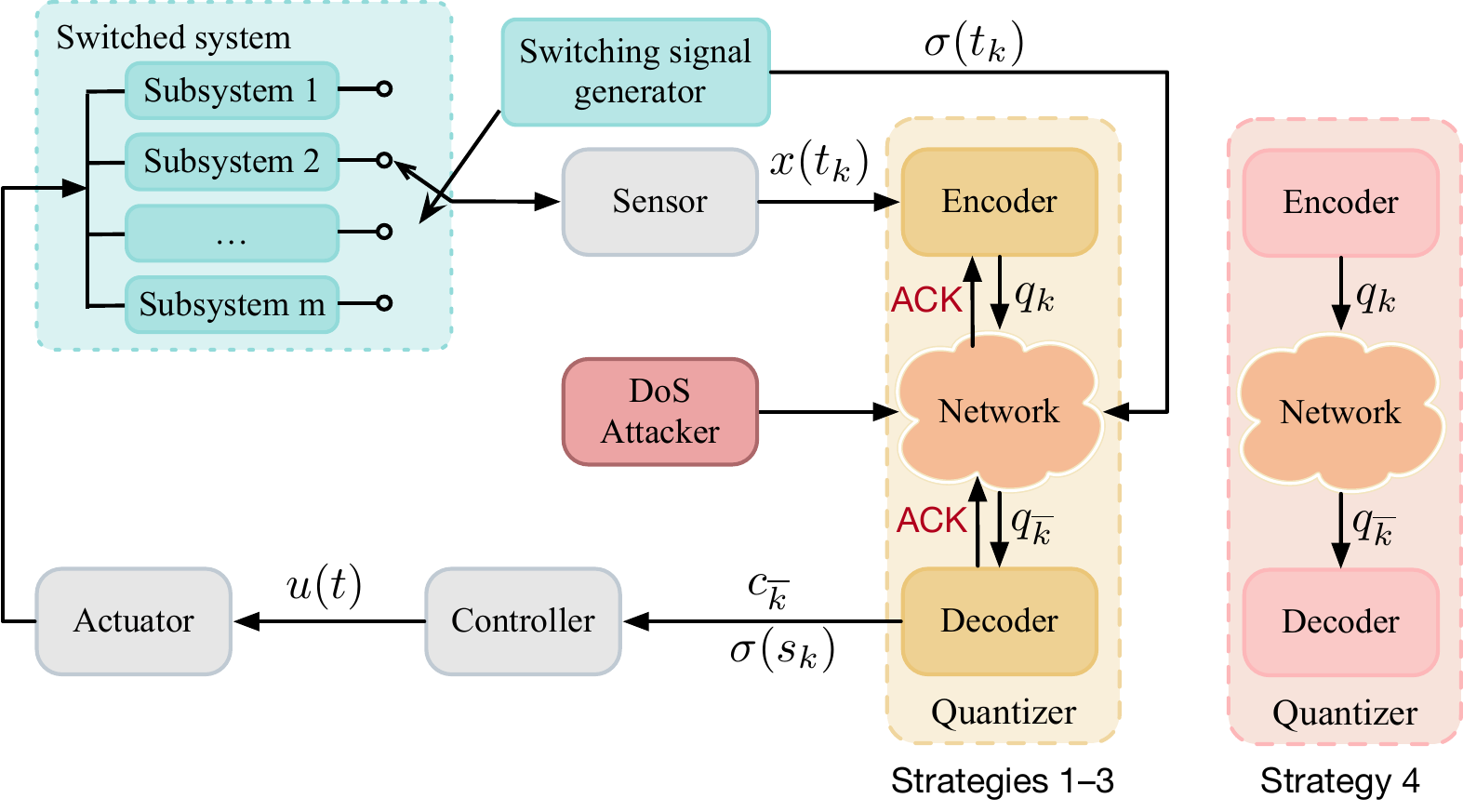}
	\caption{\textcolor{blue}{The block diagram of networked switched systems with quantizer under DoS attack. The left quantizer and the right one are designed for \textbf{\textit{Strategies 1--3}} with ACK signal and  \textbf{\textit{Strategy 4}} without ACK signal, respectively.}}\label{fig:frame}
\end{figure}

Fig. \ref{fig:frame} illustrates the block diagram of network switched systems with quantizer under DoS attack. 
The closed-loop system comprises a quantizer, a network suffering from DoS attack, and a controller. 
The system under consideration is a sampled-data control one, where the sensor samples the state signal at discrete time instants with a fixed sampling period $\tau_s$. The sampling instants are defined as $t_k = k \tau_s$, where $k = 0, 1, 2, \cdots$. At each sampling instant, the state signal $x(t_k)$ will be encoded into an integer $\textcolor{blue}{q_k} \in \{1, 2, \dots, N^{n_x}\}$, where $N$ is an odd positive integer and $n_x$ represents the dimension of the state vector. More specifically, the integer $N$ represents the number of quantization levels.
After the encoding process, the encoder transmits the encoded information \textcolor{blue}{$q_k$} to the decoder. Upon receiving  \textcolor{blue}{$q_k$}, the decoder reconstructs the corresponding state information $c_k$ which is the center of the box labeled \textcolor{blue}{$q_k$}. Then the decoded information $c_k$ is subsequently used by the controller.
In addition to the transmission of the encoded state information, a switching signal $\sigma(t_k)$ is transmitted simultaneously at each sampling instant. 

In what follows, we introduce the components of the system, including quantizer, network subject to DoS attack, and controller.  Next, the control strategies are explained, along with the cases that arise in the design of the quantizer. Finally, the main objectives and structure of this paper are presented.
\subsection{Quantizers}
The quantization error depends on the quantization center and the quantization range, see \cite{quan_2014} for details. In existing literature, the quantization parameters for both encoder and decoder are uniform when the system is a non-switched one with/without DoS attack or the system is switched system without DoS attack. However, the consistency of the quantizer needs to be guaranteed by a well-designed update law for switched systems. Asynchronous behavior occurs between the encoder and decoder sides of switched systems under DoS attack. Consequently, we design the quantization parameters for the encoder and decoder separately. Moreover, the uniformity of quantizer will be 
be addressed in this paper. The encoder and decoder satisfy the following conditions.

\textbf{Encoding}: The encoder admits
\begin{equation}\label{equ_bound}
	\| x(t_{k})-x_k^{e*} \| \leq E_k^e, ~\forall{k \in \mathbb{Z}^+}
\end{equation} 
where $x_k^{e*}$ is the center of the quantizer and $E_k^e$ represents the quantization range on the encoder side.  These parameters will be determined in the sequel.
The hypercube $ \mathcal{H} \triangleq \{x\in \mathbb{R}^{n_x}: \|x(t_{k}) - x_k^{e*}\| \leq E_k^e\}$ is divided into $N^{n_x}$ identical hypercubic boxes, where $N$ is the number of quantization levels. More specifically, the quantization area in each dimension is divided into $N$ segments. These boxes are then labeled from $1$ to $N^{n_x}$. The box number $q_k$ containing the true value of $x(t_k)$  will be transmitted via network.

\textbf{Decoding}:  Once the decoder receives the encoded signal $q_k$, the center of box $c_k$ will be calculated with the help of quantized parameters $E_k^d$ and $x_k^{d*} $.
Then it follows that
\begin{equation}\label{equ_error_bound}
	\| x(t_k)-c_k\| \leq \frac{E_k^d}{N}
\end{equation}
and 
	$	\| c_k-x_k^{d*}\| \leq \frac{(N-1)E_k^d}{N}.$
To ensure stability for systems under quantized control,  it is required that $E_k^d= E_k^e$ and $x_k^{d*} = x_k^{e*} $.

According to the quantizer center, two kinds of quantizers are given:
\begin{itemize}
	\item \textit{Quantizer i:} The quantization center	$x_k^{*}$ is not the origin. Both $x_k^*$ and $E_k^e$ need to be designed.
	\item \textit{Quantizer ii:} The quantization center is the origin. The update law of parameter $E_{k}^e$ will be devised.
\end{itemize}
\subsection{DoS Attack via a Network}
Since the quantized signal is transmitted over a network, it is more vulnerable to attackers.  
In such cases, as shown in Fig. \ref{fig:frame}, the quantized signal $q_k$ cannot always be successfully sent to the decoder at every sampling instant. The received quantization signal then becomes $q_{\overline{k}}$. Accordingly, the decoded signal is thus denoted as $c_{\overline{k}}$.
It is noted that the switching signal may also be attacked. At the sampling instants, the switching signal $\sigma(t_k)$ can only be transmitted at the attack-free instant $s_k$. Due to the occurrence of DoS attack, asynchronous behavior may occur.

Let the starting instant of the $n$-th DoS attack be $h_n$ and the duration be $\tau_n$, then the $n$-th intervals for attack and attack-free scenarios are 
$	\mathcal{H}_n \triangleq{h_n} \cup [h_n,h_n+\tau_n),
\mathcal{D}_n \triangleq [h_n+\tau_n,h_{n+1})$.
Therefore,  the total DoS attack and attack-free duration can be expressed as 
$			\Xi(t,t_0) \triangleq \cup _{n\in\mathbb{N}_{\geq 0}} \mathcal{H}_n \cap [t_0,t),
\Theta(t,t_0) \triangleq  [t_0,t) \setminus \Xi(t,t_0).$
In this paper, the network under DoS attack satisfies the following constraints.
\begin{assumption}[DoS attack frequency\cite{etDoS-Ass}]\label{DoSF}
	The number of DoS attack off/on switching over the interval $[t_0,t)$ satisfies 
	\begin{equation}\label{equ_dosF}
		n(t,t_0) \leq n_0+ \frac{t-t_0}{\tau_D}, ~\forall ~t>t_0\geq 0.
	\end{equation}
	where scalars $n_0 \in \mathbb{R}_{\geq 0 }$ and $\tau_D \in \mathbb{R}_{>0}$.
\end{assumption}
\begin{assumption}[DoS attack duration\cite{etDoS-Ass}]\label{DoSD}
	The length of DoS attack duration obeys 
	\begin{equation}\label{equ_dosD}
		|\Xi(t,t_0) | \leq \kappa+\frac{t-t_0}{T},~\forall ~t>t_0\geq 0.
	\end{equation}
	where  scalars $\kappa\in \mathbb{R}_{\geq 0 }$ and $T \in \mathbb{R}_{>1}$.
\end{assumption}

Note that the quantized signal is sampled with a fixed period $\tau_s$, it cannot be sent to the controller until the first sampling instant after one DoS attack interval. Similar to  \cite{Rui2}, the equivalent DoS attack duration admits 
$	|\overline{\Xi} (t,t_0)| \leq  |\Xi(t,t_0)| + \tau_s \cdot (n(t,t_0)+1)
\leq \overline{\kappa}+\frac{t-t_0}{\overline{T}}$
where $\overline{\kappa} = \kappa + \tau_s(n_0+1)$ and $\frac{1}{\overline{T}} =\frac{1}{{T}}+\frac{\tau_s}{\tau_D}  $.
Thus the DoS attack-free duration turns to be 
$	|\overline{\Theta}(t,t_0)| = t-t_0- |\overline{\Xi}(t,t_0)|.$

As we know, 	DoS attacks disrupt network availability through UDP flooding, challenge collapsar, SYN flooding, etc.  The status of DoS attack can be detected with the help of ACK signal.
ACK signal is commonly used as a response to confirm the successful receipt of signal; interested readers may refer to TCP/IP protocol for more details.
For example, if the network protocol used in this paper is TCP/IP, an ACK signal will be sent from the decoder to the encoder for attack-free case, i.e., $ACK =1$. Otherwise, if there exists an attack, $ACK = 0$. However, some protocols, such as the UDP protocol, do not provide an ACK signal. In such a case, the status of DoS attack cannot be directly determined. In this paper, two kinds of networks are investigated, i.e., 
\begin{itemize}
	\item \textit{Network i}: The network with ACK signal. The DoS attack is available to the encoder. 
	\item \textit{Network ii}: The network without ACK signal.The DoS attack cannot be known by the encoder.  
\end{itemize}		
\subsection{Controller }
There are mainly two approaches to calculate the control signal. The first one employs the predicted state to generate a continuous control signal, as in \cite{quan_2014}. The other updates the control signal at sampling instants and keeps its value between two consecutive sampling instants, i.e., the zero-order hold (ZOH) mechanism; see \cite{output1} for more details. 
In this paper, we put forward two strategies to generate the control signal for different quantizers. We call the two control approaches as the \textit{active controller} and the \textit{passive controller}.

\textit{Active controller: }The control signal is generated in terms of the predictor. And the predictor has the form
\begin{equation}\label{equ_pre_state}
	\begin{aligned}
		\dot{\hat{x}}(t) = A_{\hat{\sigma}(t)} \hat{x}(t) + B_{\hat{\sigma}(t)} u(t), \text{ for } t\in [s_k,s_{k+1}),
	\end{aligned}
\end{equation}
where $\hat x (s_k^+) = c_{\overline{k}}$, $s_k$ is the $k$-th instant when the signal has been successfully transmitted, $c_{\overline{k}}$ is the state signal after decoding and ${\hat{\sigma}(t)} = \sigma(s_k)$ represents the controller mode.
For DoS attack-free case, the predictor updates the state at sampling instants. However, when a DoS attack occurs, the predictor would not update the state at each sampling instant since the decoder cannot receive the latest state. Consequently, the state will be updated at the successful transmission instants.
With the predicted state,  the control signal is designed as
$u(t) = K_{\hat{\sigma}(t)} \hat{x}(t)$
where $K_p$ is the control gain for mode $p\in\mathcal{M}$.

\textit{Passive controller: }The control signal keeps unchanged using ZOH mechanism during two sampling instants. The control signal for DoS attack-free case is 
$	u(t) = K_{\hat{\sigma}(t_k)} \hat{x}(t_k), \forall t\in [t_k,t_{k+1}).$
When DoS attack occurs, $\hat{x}(t_k)=\textbf{0}$. Hence, 
$	u(t) = \textbf{0}$.
\color{blue}
In this sense, passive control is similar to intermittent pinning control \cite{R1-1}.
\color{black}		
\subsection{Control Strategies}
Four quantized control strategies are designed employing different controllers and quantizers, while considering whether DoS attack is detectable:
\begin{itemize}
	\item \textbf{\textit{Strategy 1:}} \textit{Active controller} \& \textit{Quantizer i }\& \textit{Network i}
	\item \textbf{\textit{Strategy 2:}} \textit{Active controller} \& \textit{Quantizer ii }\& \textit{Network i}
	\item \textbf{\textit{Strategy 3:}} \textit{Passive controller} \& \textit{Quantizer ii }\&\textit{ Network i}
	\item \textbf{\textit{Strategy 4:}} \textit{Passive controller} \& \textit{Quantizer ii }\& \textit{Network ii}
\end{itemize}

The decoder can detect the status of DoS attacks under each strategy, but this is not the case for the encoder. Notably, \textit{\textbf{Strategies 1-3}} employ ACK technique, allowing the encoder to recognize the occurrence of DoS attack and respond accordingly. On the contrary, in \textit{\textbf{Strategy 4}}, the encoder fails to detect DoS attack. Therefore, the quantizer's update law in \textbf{\textit{Strategy 4}} has no relationship with the DoS attack status, as the quantization parameters should remain uniform. 

The core principle of dynamic quantizer for systems without attacks is the convergence of the quantization parameters, which meets condition \eqref{equ_bound}. This process is referred to as \textit{zooming-in}. However, if a DoS attack occurs, the new state cannot be transmitted, and the quantization error may diverge, violating condition \eqref{equ_bound}; see \cite{output1} and \cite{output2}. To capture the state divergence caused by DoS attack, the \textit{zooming-out} scheme is employed.
In addition, asynchronous behavior brings additional challenges. In \cite{quan_2014} and \cite{quan_2018}, asynchronous behavior emerges since a switch may occur between two consecutive sampling instants. The \textit{zooming-out} scheme is used to capture the system state when the controller mode is inconsistent with the subsystem mode. Note that the maximum asynchronous interval is less than one sampling period, i.e., $\tau_s$, and the encoder and decoder parameters remain unchanged at all times.
In contrast, the asynchronous behavior induced by a DoS attack may far exceed one sampling period. Furthermore, the new mode is unknown to the decoder if there exists a switching during the DoS attack interval. In other words, the encoder’s mode (which matches the subsystem’s mode) and the decoder’s mode (which matches the controller’s mode) become different, potentially destabilizing the closed-loop system due to incorrect decoding. Therefore, the status of the DoS attack during the asynchronous interval should be further investigated.

\begin{table}
	\centering
	\caption{Different cases of updating $E_k^e$ and $x^{e*}_k$ for \textbf{\textit{Strategies 1-3}}}
	\label{Table:case}
	\renewcommand\arraystretch{1}
		\begin{tabular}{ccccc}
			\hline
			$SY _{t_{k}}$ &$SY _{t_{k+1}}$ & 	$ACK_{t_k}$ & $ACK_{t_{k+1}}$ &  Cases \\
			\hline
			1/0& 1& 1&1/0& \begin{example1}\label{case_1}\end{example1} \\ 
			1/0& 1& 	0&1/0&    \begin{example1}\label{case_2}\end{example1}\\ 
			1& 0&1&1&   \begin{example1}\label{case_3}\end{example1}\\
			1& 0&1&0&   \begin{example1}\label{case_4}\end{example1}\\
			1& 0&0&0&  \begin{example1} \label{case_5}\end{example1}\\
			1& 0&0&1&   \begin{example1}\label{case_6}\end{example1}\\
			0& 0&0&0&  \begin{example1} \label{case_7}\end{example1}\\
			0& 0&0&1&   \begin{example1} \label{case_8}\end{example1}\\
			\hline
		\end{tabular}
	\end{table}
	For clarity, let $ACK_{t_k}$ represent the ACK signal at time instant $t_k$. Define the synchronous stage flag as $SY_{t_k} = 1$ during the interval $[t_{k-1}, t_k)$. When $SY_{t_k} = 0$, the system is operating in the asynchronous stage.
	Depending on the values of $ACK_{t_k}$, $ACK_{t_{k+1}}$, $SY_{t_k}$, and $SY_{t_{k+1}}$, the update laws for $E_k^e$ and $x^{e*}_k$ can be categorized into eight cases. See Table \ref{Table:case} for details.
	The analysis presented above pertains to \textbf{\textit{Strategies 1-3}}, with the help of the ACK technique. And the classification for \textbf{\textit{Strategy 4}} will be discussed in the subsequent section.
	\subsection{Main Objective}
	Here, we first present a definition and an assumption for this paper. 
	
	\begin{definition}[\hspace{-0.015cm}\cite{DT}]
		For any two consecutive switching instants $t_{s}$ and $t_{s+1}$ with $s \in \mathbb{N}^{+}$, if $t_{s+1} - t_s \geq \tau_d>0$, then $\tau_d$ is called the dwell time.
	\end{definition}
	
	\begin{assumption}[Initial state bound]\label{initial}
		The initial state $x(0)$ is bounded by a constant $E_0> 0$, i.e., $\|x(0)\| \leq {E}_0$.
	\end{assumption}
	
	Assumption \ref{initial} has been widely adopted in the literature \cite{Liu2022,output1,output2}. Furthermore, the value of $E_0$ can be determined using the \textit{zooming-out} strategy \cite{quan_2014,quan_2018}. 
	
	The aim of this paper is to design appropriate update laws for the quantizers and the corresponding switching law in each strategy, ensuring the stability of switched systems under DoS attack. Fig. \ref{fig:structure} shows the structure of this paper.

	\begin{figure}
		\centering
		\includegraphics[width=0.85\linewidth]{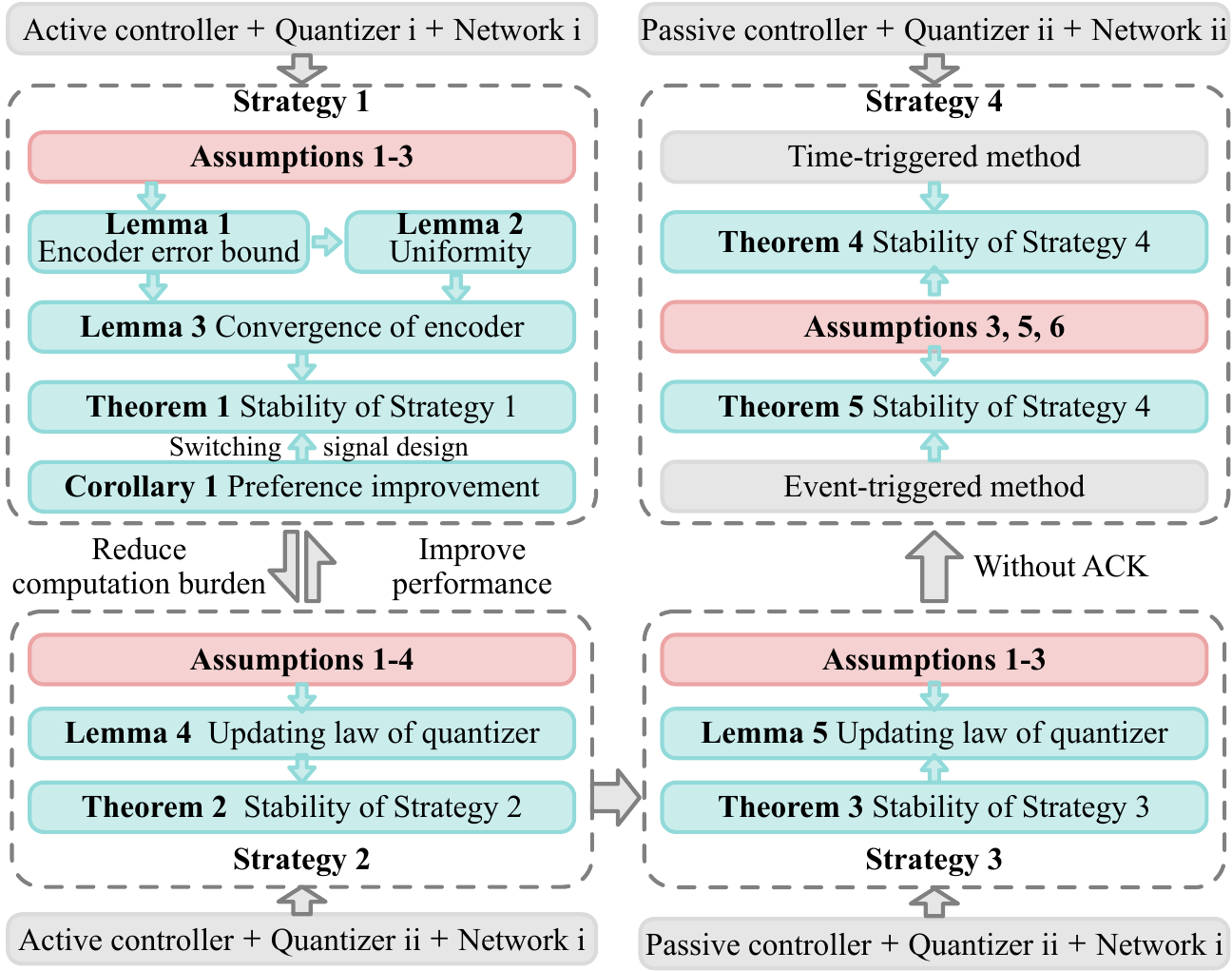}
		\caption{\textcolor{blue}{The structure of this paper}}
		\label{fig:structure}
	\end{figure}
	
\color{black}
\section{Active Quantized Control Strategy} \label{sec_active}

It is well known that DoS attack may degrade system performance and even destroy stability by disrupting the control signal. However, an appropriate control signal, provided by an active controller, can stabilize the system. As stated in \cite{Rui2}, the system state still converges during synchronous intervals when a DoS attack occurs. Moreover, this holds for systems subject to bounded disturbances \cite{Rui3}.
Inspired by active control strategy in \cite{Rui2,Rui3}, we investigate \textbf{\textit{Strategies 1-2}}, which incorporate the \textit{active controller}. The update laws for the quantization parameters are derived to ensure that the system state remains within the quantization range. Subsequently, a sufficient condition concerning DoS attack constraints and the switching signal is proposed to guarantee closed-loop stability.
\subsection{Closed-loop System Using Active Controller}		
In the sequel, we give the dynamics of the closed-loop system using \textit{active controller}.
First of all,  define the error between system state and predicted state 
\begin{equation}\label{equ_error}
	e(t) = x(t)- \hat{x}(t),
\end{equation}
then the dynamics of the closed-loop system admits
\begin{equation}\label{equ_x_dot}
	\dot x(t)  = \overline{A}_{\sigma(t)\hat \sigma(t)}x (t) + B_{\sigma(t) }K_{\hat \sigma(t) } e(t)
\end{equation}
where $\overline{A}_{\sigma(t)\hat \sigma(t)} = A_{\sigma(t)}+B_{\sigma(t)}K_{\hat \sigma(t)}$. 
The closed-loop system is divided into synchronous stage and asynchronous stage with/without DoS attack. 

$\bullet$ Synchronous stage:
In this stage, the modes of controller and subsystem are the same, that is, $\sigma(t)  = \hat \sigma(t)$. 
Suppose that $\sigma(t) = p$, then the solution of \eqref{equ_x_dot} has the form 
$	x(t) = e^{\overline{A}_{pp} (t-t_k)} x(t_k) + \int_{t_k}^{t} e^{\overline{A}_{pp} (t-s)}B_pK_p e(s) ds$
for all $t\in [t_k,t_{k+1})$. 	From  \eqref{equ_state} and  \eqref{equ_pre_state},  we have 
\begin{equation}\label{equ_edot}
	\dot {e} (t) = A_p e(t), 	 \forall t\in [t_k,t_{k+1}).
\end{equation}
which implies  
\begin{equation}\label{equ_etsol}
	e(t) = e^{{A}_p (t-t_k)} e(t_k) ,\forall t  \in [t_k,t_{k+1}).
\end{equation}
Thus, one has 
\begin{equation} \label{equ_xtk1_ds}
	\begin{aligned}
		&x(t_{k+1}) 
		= A_p^d x(t_k) + B_p^d e(t_k)
	\end{aligned}
\end{equation}
where $A_p^d  = e^{\overline{A}_{pp} \tau_s}$ and $B_p^d = \int_{0}^{\tau_s} e^{\overline{A}_{pp} (\tau_s-s)}B_pK_pe^{{A}_p s}ds$.

$\bullet$Asynchronous stage containing the switching instant: 
For interval $[t_s, t_{k+1})$ with $t_s$ being the switching instant, the dynamics of  the error becomes
$	\dot	e(t) 
= ( \overline{A}_{pq} - \overline{A}_{qq})x(t) + (B_pK_q- \overline{A}_{qq}) {e}(t)
=  \Pi_{pq}^1 x(t) + \Pi_{pq}^2  e(t)$
where $p= \sigma(t),~q= \hat{\sigma}(t)$, $\Pi_{pq}^1  =  \overline{A}_{pq} - \overline{A}_{qq}$, and $\Pi_{pq}^2  =  B_pK_q- \overline{A}_{qq}$. 
Let $e_z(t) =\left[\begin{matrix}
	x^T(t) &e^T(t)
\end{matrix}\right]^T$, we have
\begin{equation} \label{equ_ez}
	\begin{aligned}
		\dot	e_z(t) =~& \left[\begin{array}{cc}
			A_p + B_p K_q	& B_p K_q  \\
			\Pi_{pq}^1 &  \Pi_{pq}^2
		\end{array}
		\right] e_z(t) 
		=~ \widetilde{\mathcal{A}}_{pq} e_z(t).
	\end{aligned}
\end{equation}

From \eqref{equ_x_dot} and  \eqref{equ_edot}, one gets 
$	\dot e_z(t) =  \left[\begin{array}{cc}
	A_p	& B_pK_p  \\
	\textbf{0}	&  A_p 
\end{array}\right]e_z(t) = \overline{\mathcal{A}}_{p} e_z(t),~ \forall t\in [t_k,t_s).$
Hence, $	e_z(t_{k+1}) =  e^{\tilde{\mathcal{A}}_{pq}\overline{t}}e^{ \overline{\mathcal{A}}_{p} (\tau_s -\overline{t})} e_z(t_k)$ where $\overline{t} = t_{k+1}-t_s$. 
Then, it can be obtained that 
\begin{equation} \label{equ_closed_asy}
	\begin{aligned}
		x(t_{k+1}) 
		=~& \left[\begin{array}{cc}
			\textbf{I} &\textbf{0}
		\end{array}\right] e^{\widetilde{\mathcal{A}}_{pq}\overline{t}}e^{ \overline{\mathcal{A}}_{p} (\tau_s -\overline{t})} e_z(t_k) \\
		= ~& A_{pq}^d (\overline{t}) 	\tilde{x}(t_k)+ B_{pq}^d (\overline{t})  	\tilde{e}(t_k)
	\end{aligned}
\end{equation}
where $	\tilde{x}(t)  =  \left[\begin{array}{c}
	{x}(t)	\\
	\textbf{0}
\end{array}\right]$, $A_{pq}^d (\overline{t})  = \left[\begin{array}{cc}
	\textbf{I} &\textbf{0}
\end{array}\right] e^{\widetilde{\mathcal{A}}_{pq}(\overline{t})}e^{ \overline{\mathcal{A}}_{p} (\tau_s -  \overline{t})} $,  $	\tilde{e}(t)  =  \left[\begin{array}{c}
	\textbf{0} \\ {e}(t) 
\end{array}\right]$ and $B_{pq}^d (\overline{t})  =  A_{pq}^d (\overline{t})  $.

$\bullet$ Asynchronous stage without switching instant:
Based on the above analysis, we have 
$	x(t_{k+1}) 
=  A_{pq}^d (\tau_s) 	\tilde{x}(t_k)+ B_{pq}^d (\tau_s)  	\tilde{e}(t_k).$

Since the system is stabilizable under control gain $K_p$ and the divergence rate of state has an upper bound, there exist scalars  $\eta_{pq} ,~\xi_{pq}$,  $\rho_p>0$ and $0<\lambda_p <1$, such that 	$ \| ({A}_p^d)^k\| \leq \rho_{p} \lambda_p^k$ and $  \left \| \left[\begin{array}{c}
	{A}_{pq}^d(\tau_s)	\\
	\textbf{0}
\end{array}\right]^k\right\| \leq \xi_{pq} \eta_{pq}^k$  for all mode $p,q\in\mathcal{M}$. 

\subsection{ Strategy 1} \label{sec_quan1}
First, some definitions of notation are given. Define a time sequence $s_{k-1} \leq t_{k} < t_s\leq t_{k+1}<\cdots <  t_{k+m} = s_k (m\geq 1)$. 
$\llcorner t_s\lrcorner$ means the latest sampling instant $t_k$, i.e., $\llcorner t_s\lrcorner = t_k$. 	$	\ulcorner t_s \urcorner $ stands for the next sampling instant after $t_s$, i.e., $	\ulcorner t_s \urcorner = t_{k+1}$.
$\lceil t_s\rceil$ represents the next successful transmission instant $s_k$, i.e., $\lceil t_s\rceil = s_k$.

\begin{lemma}[Encoder error bound]\label{lemma1}
	For the switched system \eqref{equ_state}  with \textbf{\textit{Strategy 1}}, if the update law of  the quantizer has the form
	
	{\scriptsize  \vspace{-1em}	\begin{align}
			E_{k+1}^e = &
			\begin{cases}
				\begin{array}{ll}
					\frac{\Gamma_{\sigma(t_k)}}{N} {E_k^e} &\text{\cref{case_1}} \\
					\Gamma_{\sigma(t_k)} {E_k^e}  &\text{\cref{case_2}}\\
					\Gamma_{{\sigma}(t_k) \hat\sigma(t_k)}^{1} E_k^e	
					+ \Gamma_{{\sigma}(t_k) \hat\sigma(t_k)}^{2} \| x_k^{e*}\|  &\text{\cref{case_3,case_4,case_5,case_6}}\\
					\Gamma_{{\sigma}(t_k) \hat\sigma(t_k) }^3E_{k}^e&\text{\cref{case_7,case_8}}
				\end{array}
			\end{cases}\label{equ_Eke}\\
			x_{k+1}^{e*}=  &
			\begin{cases} \hspace{-0pt}
				\begin{array}{ll}
					e^{\big(A_{ \sigma(t_k) }+B_{ \sigma(t_k)} K_{ \sigma(t_k)} \big)\tau_s} c_k   & \hspace{-4pt} \text{\cref{case_1,case_3,case_4}}  \\ \hspace{-0pt}
					e^{\big(A_{ \sigma(t_k)} +B_{ \sigma(t_k)} K_{ \sigma(t_k)} \big)\tau_s} \hat{x}(t_k^-)& \hspace{-4pt} \text{\cref{case_2,case_5,case_6}} \\ \hspace{-0pt}
					[\begin{matrix}
						\textbf{I} &\textbf{0}
					\end{matrix}
					] e^{ \mathcal{A}_{{\sigma}(t_k)\hat \sigma(t_k) } n \tau_s} \left[\begin{matrix}
						\textbf{I} \\
						\textbf{I}
					\end{matrix}
					\right]  \hat x(\ulcorner t_s\urcorner)&\hspace{-4pt} \text{\cref{case_7,case_8}} 	
				\end{array}
			\end{cases}\label{equ_xke}
	\end{align}}where $\Gamma_p  = \| e^{A_p \tau_s} \|$, $\Gamma_{pq}^{1} =\max_{\overline{t}\in [0,\tau_s)} \left \|e^{\mathcal{A}_{pq} \overline{t}}  e^{\mathcal{A}_{qq} ( \tau_s -\overline{t})} \right\| +\Gamma_{pq}^{2} \frac{N}{N-1}$, $\Gamma_{pq}^{2}  =  \max_{\overline{t}\in [0,\tau_s)} \left \|e^{\mathcal{A}_{pq} \overline{t}}  e^{\mathcal{A}_{qq} ( \tau_s -\overline{t})}-e^{\hat{\mathcal{A}}_{q} \tau_s } \right \| $,   $\Gamma_{pq}^3 = \|e^{ \mathcal{A}_{pq}\tau_s}\| $ and  $n = \frac{\lceil t_{s} \rceil - \ulcorner t_s \urcorner }{\tau_s}$
	$(p,q\in \mathcal{M})$. 	  
	Then 
	\begin{equation}\label{equ_bound1}
		\|x(t_k)-x_k^{e*}\| \leq E_k^e,~ k\geq 0.
	\end{equation}
\end{lemma}
\begin{proof}
	According to  \eqref{equ_error_bound}, \eqref{equ_pre_state} and \eqref{equ_error}, the error when the quantization signal is successfully transmitted  turns to be
	\begin{equation}\label{equ_error_bound2}
		\|	e(t_k^+)\| = \| x(t_k) - c_k^*\|  \leq \frac{E_k^e}{N}.
	\end{equation}
	On the other hand, if the transmission of quantization signal  fails, one has $\hat{x}(t_k^+) = \hat{x}(t_k^-)$
	and	
	\begin{equation}\label{equ_etk}
		\| e(t_k^+) \| = \| e(t_k^-) \| \leq E_k^e.
	\end{equation}
	
	\textbf{\cref{case_1}: $ACK_{t_k} = 1$ and $SY_{t_{k+1}}=1$.}
	
	Let the system mode be $p$ during $[t_k,t_{k+1})$. 
	Combining \eqref{equ_etsol} with  \eqref{equ_error_bound2},  the error at $t_{k+1}$ admits
	\begin{equation}\label{equ_Eke1}
		\|	e(t_{k+1}^-)\|  =\|e^{A_p \tau_s} e(t_k^+) \|  \leq \Gamma_p \frac{E_k^e}{N}  =:E^{e}_{k+1}
	\end{equation}
	where $\Gamma_p  = \| e^{A_p \tau_s} \|$.
	Let 
		$x^{e*}_{k+1}  :=	 \hat{x}(t_{k+1}^-) = e^{(A_p+B_pK_p)\tau_s} c_k$
	Based on \eqref{equ_Eke1} and the  fact
	\begin{equation}\label{equ_lem1_1}
		\|	e(t_{k+1}^-)\|  = \|x(t_{k+1}^-)-\hat x(t_{k+1}^-)\| =  \|x(t_{k+1}^-)-x^{e*}_{k+1} \|, 
	\end{equation}
	formula \eqref{equ_bound1} follows directly.

	\textbf{\cref{case_2}: $ACK_{t_k} = 0$ and $SY_{t_{k+1}}= 1$.}

	From \eqref{equ_edot} and \eqref{equ_etk}, we have 
	\begin{equation}\label{equ_Eke2}
		\|	e(t_{k+1}^-)\|  =\|e^{A_p \tau_s} e(t_k^+) \|  \leq \Gamma_p {E_k^e}  =: E^{e}_{k+1}
	\end{equation}
	and 
	$\hat{x}(t_{k+1}^-) = e^{(A_p+B_pK_p)\tau_s} \hat{x}(t_k^-) =:	x^{e*}_{k+1} $.
	Similar to \eqref{equ_lem1_1}, we can get  \eqref{equ_bound1}.
	
	\textbf{\cref{case_3,case_4}:  $ACK_{t_k} = 0$ and $SY_{t_{k+1}}= 0$.}
	
	Suppose that the system mode is $p$, i.e.,  $\sigma(t_{k+1}) = p$ and the predictor mode is $q$, i.e., $\hat{\sigma}(t_{k+1}) = q$. 
	By introducing a new variable $z(t) = \left[\begin{matrix}
		x^T(t) &
		\hat x^T(t)
	\end{matrix}
	\right]^T$, we know that   for $t\in [  t_s ,t_{k+1})$
	\begin{equation}\label{equ_asy_error_asy}
		\dot z(t) = \left[\begin{array}{cc}
			A_p & B_pK_q \\
			\textbf{0}& A_q+B_qK_q
		\end{array}
		\right]z(t) = \mathcal{A}_{pq} z(t)
	\end{equation}
	and   for $t\in [ t_k, t_s)$, 
	$	\dot z(t) = \left[\begin{array}{cc}
		A_q & B_qK_q \\
		\textbf{0}& A_q+B_qK_q
	\end{array}
	\right]z(t) = \mathcal{A}_{qq} z(t).$
	The solution of the above differential equation at instant $t_{k+1}$ is
	\begin{equation}\label{equ_equ_asy_error_asy}
		z(t_{k+1}) = e^{\mathcal{A}_{pq} \overline{t}}  e^{\mathcal{A}_{qq} ( \tau_s -\overline{t})} z(t_k^+)
	\end{equation}
	where $ \overline{t} = t_{k+1}-t_s$.
	Furthermore, define the auxiliary system as
	$	\dot{\hat{z}}(t)  =  \hat{\mathcal{A}}_{q} \hat z(t),~ \forall t\in [  t_k  ,t_{k+1})$
	with $\hat z(  t_k)=\left[\begin{array}{c}
		\hat x(  t_k^+ ) \\
		\hat x(  t_k^+ )
	\end{array}
	\right]$ and $ \hat{\mathcal{A}}_{q} =\left[\begin{array}{cc}
		\overline{A}_{qq}&\textbf{0}  \\
		\textbf{	0}& \overline{A}_{qq}
	\end{array}
	\right] $.
	Thus, one has
	\begin{align}
		&\| z(t_{k+1}) - \hat  z(t_{k+1})\|\label{equ_Eke3_asy}   \\
		\leq ~&\big  \|e^{\mathcal{A}_{pq} \overline{t}}  e^{\mathcal{A}_{qq} ( \tau_s -\overline{t})} \big\| \big\| z(t_{k}^+)-\hat z(t_{k}) \big\| \notag  \\&
		+ \big\|	e^{\mathcal{A}_{pq} \overline{t}}  e^{\mathcal{A}_{qq} ( \tau_s -\overline{t})}-e^{\hat{\mathcal{A}}_{q} \tau_s }  \big\| \big\|  \hat z (t_{k})\big \|   \notag \\
		\overset{(a)}{\leq}  ~& \hat  \Gamma_{pq}^{1}E_k^e+\Gamma_{pq}^{2}\big (\frac{N-1}{N} E_k^e +\| x_k^e\| \big)    \notag  \\
		\leq ~&\big (\hat  \Gamma_{pq}^{1}+\Gamma_{pq}^{2}\frac{N-1}{N}  \big)E_k^e + \Gamma_{pq}^{2} \| x_k^{e*}\| =:  \Gamma_{pq}^{1}E_k^e + \Gamma_{pq}^{2} \| x_k^{e*}\| \notag
	\end{align}
	where  $\hat \Gamma_{pq}^{1}  = \max_{\overline{t}\in [0,\tau_s)} \left \|e^{\mathcal{A}_{pq} \overline{t}}  e^{\mathcal{A}_{qq} ( \tau_s -\overline{t})} \right\| $,  $\Gamma_{pq}^{2}  =  \max \limits_{\overline{t}\in [0,\tau_s)} \left \|e^{\mathcal{A}_{pq} \overline{t}}  e^{\mathcal{A}_{qq} ( \tau_s -\overline{t})}-e^{\hat{\mathcal{A}}_{q} \tau_s } \right \| $ and $\Gamma_{pq}^{1} = \hat \Gamma_{pq}^{1}+\Gamma_{pq}^{2} \frac{N}{N-1}$.
	Inequality (a) can be derived in terms of $\left \| \left[\begin{array}{c}
		x(t_k) - 	\hat x(t_k^+) \\
		\hat x(t_k^+ )	-\hat x(t_k^+)
	\end{array}
	\right]   \right\| = \|x(t_k) - 	\hat x(t_k^+)\| =\|e(t_k^+)\| \leq E_k^e$ and $\| \hat{z}(t_k)\| = \| \hat{x}(t_k^+)\| \leq \frac{N-1}{N} E_k^e +\| x_k^{e*}\|$.
	The quantization center is 
	$	x_{k+1}^{e*}: = 	e^{(A_q+B_q K_q )\tau_s} c_k ,$
	then \eqref{equ_bound1} is satisfied.
	
	\textbf{\cref{case_5,case_6}:  $SY_{t_{k}}= 1$, $SY_{t_{k+1}}= 0$ and $ACK_{t_k} = 0$.}

	Due to $ACK_{t_k}=0$, it is obvious that $\hat{x} (t_k^-) = \hat{x} (t_k^+) = x^{e*}_k$. 
	Therefore, repeating the process in \eqref{equ_Eke3_asy}, one has 
	\begin{align}
		&\| z(t_{k+1}) - \hat  z(t_{k+1})\|\\
		\leq~& \big \|e^{\mathcal{A}_{pq} \overline{t}}  e^{\mathcal{A}_{qq} ( \tau_s -\overline{t})} \big\|  E_k^e+\big \|e^{\mathcal{A}_{pq} \overline{t}}  e^{\mathcal{A}_{qq} ( \tau_s -\overline{t})}-e^{\hat{\mathcal{A}}_{q} \tau_s } \big \| \| x_k^{e*}\|  \notag \\
		\leq ~&  \hat \Gamma_{pq}^{1} E_k^e + \Gamma_{pq}^{2} \| x_k^{e*}\|
		\leq \Gamma_{pq}^{1}E_k^e + \Gamma_{pq}^{2} \| x_k^{e*}\|.\notag
	\end{align}

	The main difference is that  the quantization center is calculated based on the predicted state instead of the value from decoded signal since a DoS attack occurs at instant $t_k$, that is, 
	$	x_{k+1}^{e*}: = 	e^{(A_q+B_q K_q )\tau_s} \hat{x}(t_k^+).$
	It is obvious that \eqref{equ_bound1} holds.

	\textbf{\cref{case_7,case_8}:  $SY_{t_{k}}= 0$, $SY_{t_{k+1}}= 0$ and $ACK_{t_k} = 0$.}
	
	From \eqref{equ_asy_error_asy}, the solution of the above differential equation admits
	$
	z(t_{k+1}) = e^{\mathcal{A}_{pq} n\tau_s} z( \ulcorner t_s \urcorner )
	$
	where $ n \tau_s =t_{k+1}-\ulcorner t_s \urcorner$.
	
	Define the auxiliary system as
	\begin{equation}\label{equ_pre_error2}
		\dot{\tilde{z}}(t)  = \mathcal{A}_{pq} \tilde z(t),~ \forall t\in [\ulcorner  t_s \urcorner ,t_{k+1})
	\end{equation}
	where $\tilde z(\ulcorner  t_s \urcorner) =\left[\begin{array}{c}
		\hat x( \ulcorner  t_s^+ \urcorner ) \\
		\hat x( \ulcorner  t_s^+ \urcorner )
	\end{array}
	\right]$.
	Thus, one has
	\begin{align*}
		&	\| z(t_{k+1}) - \tilde  z(t_{k+1})\|
		\leq  \|e^{ \mathcal{A}_{pq}{n\tau_s}}\| \|z(\ulcorner  t_s \urcorner) - \tilde z (\ulcorner  t_s \urcorner)\| \notag \\
		\overset{(b)}{\leq} ~&\|e^{ \mathcal{A}_{pq}\tau_s}\| ^nE_{k+1-n}^e
		= (\Gamma_{pq}^3)^nE_{k+1-n}^e
	\end{align*}
	where $\Gamma_{pq}^3 = \|e^{ \mathcal{A}_{pq}\tau_s}\| $. Inequality (b) follows from \eqref{equ_Eke3_asy} and the update of $x_{k+1}^{e*}$.
	Correspondingly, \eqref{equ_pre_error2} derives 
	$ [\begin{array}{cc}
		\textbf{I} &\textbf{0}
	\end{array} ]
	\tilde	z(t_{k+1}^-)  
	=[\begin{array}{cc}
		\textbf{I} &\textbf{0}
	\end{array}
	] e^{ \mathcal{A}_{pq}n\tau_s}  \left[\begin{array}{c}
		\textbf{I} \\
		\textbf{I}
	\end{array}
	\right] \hat x(\ulcorner t_s\urcorner) =:x^{e*}_{k+1} $,
	then \eqref{equ_bound1} is satisfied.
	
	To sum up, inequality \eqref{equ_bound1} is always true in each case.
\end{proof}
\color{black}
\begin{remark}
	In the encoder's update law presented in \Cref{lemma1}, it is obvious that \cref{case_3,case_4}, \cref{case_5,case_6}, and \cref{case_7,case_8} are identical. However, for the decoder, the update laws among these three pairs are different. The encoder will update its parameters according to the system's mode, while the decoder fails due to the presence of DoS attack. The asynchronous duration induced by DoS attack exceeds one sampling period. As a result, within the asynchronous interval, the decoder's parameters will only be changed in terms of the system's historical mode, i.e., controller mode.
	The switching signal is transmitted to the controller at the first sampling instant immediately after the DoS attack ends. At that instant, the decoder updates the quantization parameters according to the new mode. In \cite{quan_2014}, only \cref{case_3} is considered, which becomes unsuitable for this paper. Adopting the same approach would cause inconsistencies between the rules of the encoder and decoder, potentially leading to instability.
	To address this issue, the switching asynchronous interval $n$ is transmitted. Moreover, 	
	the method mentioned in \cite{quan_2014} is used to reduce the complexity for \cref{case_3,case_4,case_5,case_6}. This treatment simplifies the calculation of the quantizer center, reducing overall complexity.
	
\end{remark}

In what follows, the update law for the decoder is provided. Additionally, the quantization parameters of both encoder and decoder should be the same during attack-free intervals, as this is crucial for ensuring correct transmission.
When a DoS attack occurs, no information can be successfully transmitted. The property of uniformity is not required during DoS attack intervals. 

\begin{lemma}[Uniformity]\label{lemma2}
	For the switched system \eqref{equ_state}  with \textbf{\textit{Strategy 1}}, if the update law of the quantizer obeys	
	
	{\scriptsize \vspace{-1em}	 \begin{align}
			E_{k+1}^d = ~&
			\begin{cases}
				\frac{	\Gamma_{\hat{\sigma}(t_k)}}{N} E_k^d &\text{\cref{case_1}} \\
				\Gamma_{\hat \sigma(t_k)} {E_k^d}  &\text{\cref{case_2,case_4,case_5,case_7}}\\
				\Gamma_{{\sigma}(t_k) \hat\sigma(t_k)}^{1} E_k^d	
				+ \Gamma_{{\sigma}(t_k) \hat\sigma(t_k)}^{2} \| x_k^{d*}\| &\text{{\cref{case_3,case_6}} }\\
				\hspace{-0.23cm}
				\begin{array}{l}
					\Gamma_{{\sigma}(t_k)\hat{\sigma}(t_k)}^4(n)E_{{k-n}}^d \\
					+\Gamma_{{\sigma}(t_k)\hat{\sigma}(t_k)}^5 (n)\| x_{k-n}^{d*}\|	\end{array}  &\text{\cref{case_8}  } \\
			\end{cases}\label{equ_Ekd}\\
			x_{k+1}^{d*}=  ~&
			\begin{cases}
				e^{\big(A_{\hat{\sigma}(t_k)}+B_{\hat{\sigma}(t_k)}K_{\hat{\sigma}(t_k)}\big)\tau_s} c_k   &\text{ \cref{case_1,case_3,case_4}} 	\\
				e^{\big(A_{\hat{\sigma}(t_k)}+B_{\hat{\sigma}(t_k)}K_{\hat{\sigma}(t_k)}\big)\tau_s} \hat{x}(t_k^-)&\text{ \cref{case_2,case_5,case_7}} \\
				[\begin{matrix}
					\textbf{I} &\textbf{0}
				\end{matrix}
				] e^{ \mathcal{A}_{{\hat{\sigma}(t_k){\sigma}(t_k)}}n \tau_s}  \left[\begin{matrix}
					\textbf{I} \\
					\textbf{I}
				\end{matrix}
				\right] 	 \hat x(\ulcorner t_s\urcorner)&\text{\cref{case_6,case_8}}
			\end{cases}\label{equ_xkd}
	\end{align} }
	where $\Gamma_p  = \| e^{A_p \tau_s} \|$, $\Gamma_{pq}^{1} =\max_{\overline{t}\in [0,\tau_s)} \left \|e^{\mathcal{A}_{pq} \overline{t}}  e^{\mathcal{A}_{qq} ( \tau_s -\overline{t})} \right\| +\Gamma_{pq}^{2} \frac{N}{N-1}$, $\Gamma_{pq}^{2}  =  \max\limits_{\overline{t}\in [0,\tau_s)} \left \|e^{\mathcal{A}_{pq} \overline{t}}  e^{\mathcal{A}_{qq} ( \tau_s -\overline{t})}-e^{\hat{\mathcal{A}}_{q} \tau_s } \right \| $,   $\Gamma_{pq}^3 = \|e^{ \mathcal{A}_{pq}\tau_s}\| $,  $	\Gamma_{pq}^4(n) =( \Gamma_{pq}^3)^n \Gamma_{pq}^1 $,  $	\Gamma_{pq}^5(n) = (\Gamma_{pq}^3)^n \Gamma_{pq}^2 $
	$(p,q\in \mathcal{M})$ and $n = \frac{\lceil t_{s} \rceil - \ulcorner t_s \urcorner }{\tau_s}$. Then the quantization parameters in encoder and decoder are the same at the DoS-free instants.  
	
\end{lemma}

The proof is obvious in \cref{case_1,case_2,case_3}, \textbf{\ref{case_7}}, \textbf{\ref{case_8}} and it is omitted here.

\begin{remark}
	When an attack occurs and the switching instant falls within this DoS attack interval, i.e., \cref{case_4,case_5,case_6,case_7,case_8}, the real-time values \( E_{k+1}^e \), \( E_{k+1}^d \), \( x_{k+1}^{e*} \), and \( x_{k+1}^{d*} \) are constantly changing. Since the system is suffering from DoS attack, no signal will be decoded. In other words, these differences do not cause incorrect decoding. At the end of the DoS attack, i.e., \cref{case_6,case_8}, the quantization parameters are synchronized according to \eqref{equ_Eke}--\eqref{equ_xke} and \eqref{equ_Ekd}--\eqref{equ_xkd}. Therefore, the signal after decoding is accurate.
\end{remark}

Now we have derived the update laws for the quantizer. However, the quantization level, switching signal, and DoS attack constraints have not yet been addressed. For non-switched systems without attacks, selecting an appropriate quantization level would assure the convergence of the quantization parameters. However, when accounting for the effects of DoS attack and switching, additional conditions on the DoS attack and switching signal are required to guarantee the convergence of quantization parameters throughout the entire process. Furthermore, it is necessary to investigate whether the state will converge to the origin over time.

\begin{theorem}[Stability]\label{thm1}
	For  the switched system \eqref{equ_state}  with  \textbf{\textit{Strategy 1}}, if the quantizer updates its parameters according to \eqref{equ_Eke}, \eqref{equ_xke}, \eqref{equ_Ekd} and \eqref{equ_xkd}, the  quantization level $N> \Gamma$ is odd,  there exists scalar $\rho_{p}  \| {B}_{p}^d \| +b<1$ such that  DoS attack and switching signal satisfy
	\begin{equation}\label{equ_con1}
		\frac{\tau_s}{\tau_d}\log \frac{\overline{\Gamma}}{\Gamma^{N_{\max}}}+ \left(\frac{1}{{T}}+\frac{\tau_s}{\tau_D}-1\right ) \log N+\log {\Gamma}\leq \log b
	\end{equation}
	and 
	\begin{equation}\label{equ_con2}
		\tau_d\geq \max_{p, q\in \mathcal{M}}  \frac{(N_{\max}-1)\log \frac{\nu_p }{\mu_{pq}^3}-\log \frac{\hat{\nu} _p\mu_{pq}^2\mu_{pq}^1}{\mu_{pq}^3 } }{\log\nu_p }\tau_s
	\end{equation}
	where  $\Gamma = \max_{p\in\mathcal{M}}  \Gamma_{p}$, $\overline{\Gamma}=  \max\limits_{\substack{p,q\in\mathcal{M},p\neq q, \\m\in [0,N_{\max}-1]}}(\Gamma_{pq}^3)^{m-1}\Gamma_{pq}^{1}$, 	  $\nu_p = \max\{ \lambda_p , \rho_{p}  \| {B}_{p}^d \| +b \}$, 
	$ \hat{\nu}_p= \max\{ \rho_p\lambda_p , \rho_{p}  \| {B}_{p}^d \| +b\}$, $ \mu_{pq}^1 = \max\{\mathbb{A}_{pq}^d  +a\overline{ \Gamma}_2,\mathbb{B}_{pq}^d +b+a\overline{ \Gamma}_2 \}$, $\mathbb{A}_{pq}^d  = \max\limits_{\overline{t} \in [0,\tau_s]} \| \tilde{A}_{pq}^d (\overline{t})\|$, 
	$\mathbb{B}_{pq}^d  = \max_{\overline{t} \in [0,\tau_s]} \| \tilde{B}_{pq}^d (\overline{t})\|$,	$\mu_{pq}^2 = \max \{ \xi_{pq} \eta_{pq} , \xi_{pq}  \| \tilde{B}_{pq}^d (\tau_s)\| +b\}$, 	$\mu_{pq}^3 = \max \{ \eta_{pq}, \xi_{pq}  \| \tilde{B}_{pq}^d (\tau_s)\|+b \}$, and $N_{\max}$ is the maximal number of asynchronous interval,	then the switched system \eqref{equ_state} is stable under DoS attack.
\end{theorem}

Before proving Theorem \ref{thm1}, we first give a lemma to show the convergence of the quantizer error. Define the switching instant sequence as $\overline{\mathcal{T}}_s$ and the corresponding sampling number as $\mathcal{T}_s$.
Suppose that sampling instants $t_k,$ $t_{k+1},$ $\cdots,$ $t_{k+m-1}$ belong to the $n$-th DoS attack interval, i.e., $t_{k-1}< h_n\leq t_k< t_{k+m-1}<h_n+\tau_n \leq t_{k+m}$.
For convenience, let $\lfloor {h}_n \rfloor$ be the last sampling instant before the $n$-th DoS attack, i.e., $t_{k-1}$. From the definition of $\lceil\cdot \rceil$, one has $\lceil h_n+\tau_n\rceil = t_{k+m}$.
At instant $\lfloor {h}_i \rfloor$, $E_{\lfloor {h}_i \rfloor/\tau_s}^e$ is rewritten as $E_{\lfloor \hat{h}_i \rfloor }^e$ for simplicity. 

\color{black}
\begin{lemma}[Convergence of  encoder parameter]\label{lemme3}
	If switching signal and DoS attack constraints satisfy \eqref{equ_con1},
	then the encoder parameter $E_k^e$ satisfies
	\begin{align}\label{equ_eke}
		E_{k}^e \leq ~& ab^{\frac{t_k-t_0}{\tau_s}} E_{0}+ \sum _{i\in \mathcal{T}_s} ab^{\frac{t_k-t_i}{\tau_s}}  \overline{\Gamma}_2\| x_{i}^{e*}\| 
	\end{align}
	where $a = N^{\frac{\overline{\kappa}}{\tau_s}}$, $b = \left(\frac{\overline{\Gamma}}{\Gamma^{N_{\max}}  }\right)^{\frac{\tau_s}{\tau_d}} N^{\frac{1}{\overline{T}}}\frac{\Gamma}{N}$, $\Gamma = \max\limits_{p\in\mathcal{M}}\{\Gamma_{p}\}$, $\overline{\Gamma} =  \max\limits_{p,q\in\mathcal{M},p\neq q, m\in [0,N_{\max}-1]}(\Gamma_{pq}^3)^{m-1}\Gamma_{pq}^{1}$ and $\overline{\Gamma}_2 = \max \limits_{p,q\in \mathcal{M}}\Gamma_{pq}^{2}$.
\end{lemma}

\begin{proof}
	
	First, we divide the update law of quantization parameter into four scenarios during $[t_k,t_{k+m})$.

	{\textit{Scenario i:}} The controller mode is the same as the system mode within DoS attack interval.
	Define $p = \sigma(t_k)$, then the system mode is $p$ during $[t_k,t_{k+m})$.
	From \eqref{equ_Eke}, we have 
	\begin{equation}\label{equ_ekm1}
		E_{k+m}^e \leq \Gamma_{p} E_{k+m-1}^e\leq  \Gamma_{p}^m E_{k}^e.
	\end{equation}
	
	{\textit{Scenario ii}:} The controller mode differs from the system mode during DoS attack intervals.
	In such a scenario, the switching instant $t_s$ belongs to $[\lfloor h_n \rfloor, \lceil h_n+\tau_n \rceil)$. 
	Suppose that $\llcorner t_s \lrcorner = t_{k+\overline{m}} $, $\lceil h_n+\tau_n\rceil = t_{k+m}(\overline{m} \leq m)$, $q = \sigma(t_s^-)$ and $p = \sigma(t_s^+)$.	Hence, the asynchronous behavior occurs during $[t_s,t_{k+m})$, then  one gets
	\begin{equation}\label{equ_ekm2}
		\begin{aligned}
			E_{k+m}^e \leq ~& \Gamma_{pq}^3E_{k+m-1}^e\\
			\leq ~& (\Gamma_{pq}^3)^{m-\overline{m}-1} \left(\Gamma_{pq}^{1}\Gamma_{q}^{\overline{m}}E_{k}^e+ \Gamma_{pq}^{2} \| x_{k+\overline{m}}^{e*}\| \right).
		\end{aligned}
	\end{equation}

	{\textit{Scenario iii:}} The controller mode is the same as the system mode for DoS attack-free case.
	For interval $[t_k,t_{k+m})$, we have
	\begin{equation}\label{equ_ekm11}
		E_{k+m}^e \leq \frac{\Gamma_{p} E_{k+m-1}^e}{N} \leq \frac{ \Gamma_{p}^m E_{k}^e} {N^m}.
	\end{equation}
	
	{\textit{Scenario iv:}} The controller mode differs from the system mode for DoS attack-free case. 
	Suppose that $\llcorner t_s \lrcorner = t_{k+\overline{m}} $. Similar to the analysis in \textit{Scenario ii}, the quantization parameter at $t_{k+m}$ is 
	\begin{equation}
		E_{k+m}^e \leq \Big (\frac{\Gamma_p}{N} \Big )^{m-\overline{m}-1}\left(\Gamma_{pq}^{1}\Big (\frac{\Gamma_q}{N}\Big)^{\overline{m}}E_{k}^e + \Gamma_{pq}^{2} \| x_{k+\overline{m}}^{e*}\|\right).
	\end{equation}

	Next, we discuss all possible dynamics in $[t_0,t_k)$. 
	Define $\Gamma = \max\limits_{p\in\mathcal{M}}\{\Gamma_{p}\}$ and $\overline{\Gamma} =  \max\limits_{p,q\in\mathcal{M},p\neq q, m\in [0,N_{\max}-1]}(\Gamma_{pq}^3)^{m-1}\Gamma_{pq}^{1} $. Note that a switch occurs during DoS attack intervals, the asynchronous interval can be larger than one sampling interval but less than $N_{\max}$ periods. To ensure the convergence of quantization parameter, the worst case where all the switches occur during DoS attack intervals will be handled.

	From the definition of $E_k^e$ and \eqref{equ_ekm1}--\eqref{equ_ekm11}, we have 
		$	E_{k}^e \leq  \varPhi(0)E_{0}
		+ \sum _{i\in \mathcal{T}_s} \varPhi(i)   \Gamma_{pq}^{2}\| x_{i}^{e*}\|$
	where $\varPhi(i) = \overline{\Gamma}^{N_{\sigma}(t_k,t_i)}  \Gamma^{\overline{\Xi}(t_k,t_i)/\tau_s-G(t_k,t_i)} (\frac{\Gamma}{N})^{\overline{\Theta}(t_k,t_i)/\tau_s}$ with  $G(t_k,t_0)$ being the total asynchronous interval over $[t_0,t_k)$.
	By taking into account the average dwell time and the maximum asynchronous interval constraints, we have 
		$	{G(t_k,t_0)} \leq N_{\max} N_{\sigma}(t_k,t_0)$,
	then
	$	\varPhi(i) 
	\leq 	\Big(\frac{\overline{\Gamma}}{\Gamma^{N_{\max}}}\Big)^{N_{\sigma}(t_k,t_i)}  N^{\overline{\Xi}(t_k,t_i)/\tau_s} (\frac{\Gamma}{N})^{(t_k-t_i)/\tau_s} 
	= ~a b^{\frac{t_k-t_i}{\tau_s}} $
	where $a = N^{\frac{\overline{\kappa}}{\tau_s}}$ and $b = \left(\frac{\overline{\Gamma}}{\Gamma^{N_{\max}}  }\right)^{\frac{\tau_s}{\tau_d}} N^{\frac{1}{\overline{T}}}\frac{\Gamma}{N}$. 	Define $\overline{\Gamma}_2 = \max _{p,q\in \mathcal{M}}\Gamma_{pq}^{2}$, therefore \eqref{equ_eke} holds.
\end{proof}
From  \Cref{lemme3}, the average upper bound of quantization parameter $E_{k}^e$ is given.  It can be found that $E_k^e$ jumps at switching instants. 
Note that  $E_{\llcorner k_s\lrcorner }^e \geq E_{\llcorner k_{s+1}\lrcorner } ^e$ if the dwell time is sufficiently large.
In addition, due to $\| x_{i}^{e*}\| \leq \| x(t_i)\| + E_i^e $ and  (\ref{equ_eke}), we have 
\begin{equation}\label{equ_Eks1}
	E_{\llcorner t_s\lrcorner+1} ^{e-}\leq bE_{\llcorner t_s \lrcorner}^e + a\overline{\Gamma}_2( \|x(t_s)\| + E_{\llcorner t_s\lrcorner } ^e)
\end{equation}
and for  other instants, the quantization parameter admits
$	E_{k+1}^e \leq bE_{k}^e .$
Based on the  equivalent dynamics of $E_k^e$, we are now ready to prove Theorem \ref{thm1}.


\begin{proof}[Proof of Theorem~\ref{thm1}]
	Define the candidate Lyapunov function 
	$V(k) = 	\| 	{x}(t_k) \| + E_k^e$ where $E_k^e$ is defined in \eqref{equ_Eke}.
	First, we consider the case where the system switches at instant $t_s$ during DoS attack intervals $[h_n,h_n+\tau_n)$. Let $t_{k-1}= \llcorner t_s \lrcorner$. 		
	For the first asynchronous interval, it can be derived from \eqref{equ_closed_asy} that 
	\begin{align} \label{equ_asy1}
		\tilde{x}(t_{k})  
		= ~& \tilde{A}_{pq}^d (\overline{t}) 	\tilde{x}(t_{k-1})  + \tilde{B}_{pq}^d (\overline{t}) \tilde{ e}(t_{k-1})
	\end{align}
	where $ \tilde{A}_{pq}^d (\overline{t})  = \left[\begin{array}{c}
		A_{pq}^d(\overline{t})	\\
		\textbf{0}
	\end{array}\right]$ and $ \tilde{B}_{pq}^d (\overline{t})  = \left[\begin{array}{c}
		B_{pq}^d(\overline{t})	\\
		\textbf{0}
	\end{array}\right]$.
	The system dynamics for the remaining asynchronous interval is
	$	\tilde{x}(t_{k+n}) = ( \tilde{A}_{pq}^d (\tau_s) )^{{n}} \tilde{x}(t_k) + \sum_{i=0}^{{n}-1} ( \tilde{A}_{pq}^d (\tau_s)  )^{ n-1-i}    \tilde{B}_{pq}^d (\tau_s) \tilde{ e}(t_{k+i})$
	where $t_{k+n} \leq \lceil h_n+\tau_n\rceil$. Due to definition of infinity norm, we have $	\| 	{x}(t_{k+n}) \| = 		\| 	\tilde{x}(t_{k+n}) \|$. 
	Then the norm of state satisfies
	\begin{align*}
		&	\| 	{x}(t_{k+n}) \| \\
		\leq  ~& \|  ( \tilde{A}_{pq}^d (\tau_s) )^{{n}}\| \| \tilde{x}(t_k)\|
		+ \sum_{i=0}^{{n}-1} \| ( \tilde{A}_{pq}^d (\tau_s)  )^{ n-1-i} \|  \| \tilde{B}_{pq}^d (\tau_s)\|  E_{k+i}^e \\
		\leq ~& \xi_{pq} \eta_{pq}^n\| \tilde{x}(t_k)\|+ \sum_{i=0}^{{n}-1}  \xi_{pq} \eta_{pq}^{ n-1-i}  \| \tilde{B}_{pq}^d (\tau_s)\|  E_{k+i}^e .
	\end{align*}
	Moreover,  the above formula can be rewritten as 
	\begin{align}\label{equ_clo_asy_1}
		\| 	{x}(t_{k+1}) \|\leq  ~& \xi_{pq} \eta_{pq}\| 	{x}(t_{k}) \| + \xi_{pq}  \| \tilde{B}_{pq}^d (\tau_s)\|  E_{k}^e.
	\end{align}
	For other instants, it turns to be
	\begin{align}\label{equ_clo_asy_2}
		\| 	{x}(t_{k+i+1}) \|\leq  ~& \eta_{pq}\| 	{x}(t_{k+i}) \| + \xi_{pq}  \| \tilde{B}_{pq}^d (\tau_s)\|  E_{k+i}^e .
	\end{align}
	Based on the above analysis, it yields from \eqref{equ_Eks1} and  \eqref{equ_asy1} that 
	$	V(k+1)   
	= \|\mathbb{A}_{pq}^d \| \|{x}(t_k) \| +\|\mathbb{B}_{pq}^d\| E_k ^e+ bE_{k} ^e+ a\overline{\Gamma}_2 (\|x(t_k)\| + E_{k} ^e)
	\leq  \mu_{pq}^1V(k) $
	where  $\mathbb{A}_{pq}^d  = \max_{\overline{t} \in [0,\tau_s]} \tilde{A}_{pq}^d (\overline{t})$, $\mathbb{B}_{pq}^d  = \max_{\overline{t} \in [0,\tau_s]} \tilde{B}_{pq}^d (\overline{t})$, and $ \mu_{pq}^1 = \max\{\|\mathbb{A}_{pq}^d \| + a\overline{\Gamma}_2,\|\mathbb{B}_{pq}^d\| +b+ a\overline{\Gamma}_2 \}$.
	
	For the next asynchronous instant, we have 
	$$	V(k+1) 
	\leq \xi_{pq} \eta_{pq} \| {x}(t_{k}) \| + \xi_{pq}  \| \tilde{B}_{pq}^d (\tau_s)\| E_{k}^e +b E_k^e\leq  \mu_{pq}^2V(k) $$
	where $\mu_{pq}^2 = \max \{ \xi_{pq} \eta_{pq} , \xi_{pq}  \| \tilde{B}_{pq}^d (\tau_s)\| +b\}$.
	Similarly, for the other asynchronous instants, one has
	$$	V(k+1)  
	\leq  \eta_{pq} \| 	{x}(t_{k}) \| + \xi_{pq}  \| \tilde{B}_{pq}^d (\tau_s)\| E_{k}^e +b E_k^e
	\leq   \mu_{pq}^3V(k) $$
	where $\mu_{pq}^3 = \max \{ \eta_{pq}, \xi_{pq}  \| \tilde{B}_{pq}^d (\tau_s)\| +b\}$.
	
	The system dynamics at the first instant of synchronous stage is
	\begin{align}\label{equ_clo_sy_1}
		\| 	{x}(t_{k+1}) \|\leq  ~& \rho_{p} \lambda_{p}\| 	{x}(t_{k}) \| + \rho_{p}  \| {B}_{p}^d \|  E_{k} ^e
	\end{align}
	and for other instants, it becomes 
	\begin{align}\label{equ_clo_sy_2}
		\| 	{x}(t_{k+i+1}) \|\leq  ~& \lambda_{p} \| 	{x}(t_{k+i}) \| +\rho_{p}  \| {B}_{p}^d \| E_{k+i}^e .
	\end{align}
	For the first instant during synchronous stage, one gets 
	$V(k+1)  = \rho_p\lambda_p   \|{x}(t_k) \| +\rho_{p}  \| {B}_{p}^d \| E_{k} ^e+bE_k^e	\leq  \hat{\nu}_p V(k)$
	where $ \hat{\nu}_p= \max\{ \rho_p\lambda_p , \rho_{p}  \| {B}_{p}^d \| +b\}$.
	For other instants, we can see
	$	V(k+1)  =\lambda_p   \|{x}(t_k) \| +\rho_{p}  \| {B}_{p}^d \| E_{k} ^e+b E_k^e\leq  \nu_p V(k)$
	where $\nu_p = \max\{ \lambda_p , \rho_{p}  \| {B}_{p}^d \| +b\}$.
	Considering the value of Lyapunov function at switching instants, we have
	$V(\llcorner t_{s+1}\lrcorner) =  \nu_p^{\tau_d/\tau_s-N_{\max}-1}\hat{\nu} _p(\mu_{pq}^3)^{N_{\max}-2}\mu_{pq}^2\mu_{pq}^1V(\llcorner t_{s}\lrcorner) $.
	Based on \eqref{equ_con2}, it is obvious that $V(\llcorner t_{s+1}\lrcorner) \leq V(\llcorner t_{s}\lrcorner) $. Hence,  $\lim\limits_{t\rightarrow \infty} V(t) \rightarrow 0$, which implies the convergence of the state. 
\end{proof}
The analysis of \Cref{thm1} consists of two parts: the convergence of the quantization parameter $ E_k^e $ and the convergence of the system state. However, the quantization parameter $E_k^e $ jumps at switching instants, as a switching occurs between two consecutive sampling instants. To guarantee stability, a sufficient but somewhat conservative condition is derived based on a Lyapunov function, due to the coupling between the quantization parameters \( E_k^e \) and \( x_k^{e*} \). Notably, if a switch occurs at sampling instant, \( E_k^e \) is independent of \( x_k^{e*} \).
In what follows, a well-designed switching signal is proposed to ensure that the system switches at sampling instants. Consequently, the cases listed in \Cref{tab1} could be simplified, which means that \cref{case_3,case_4,case_5,case_6} can be omitted.
\begin{corollary}\label{coro1}
	For  the switched system \eqref{equ_state}  with  \textbf{\textit{Strategy 1}}, if the quantizer updates its parameters according to 
	
	{\scriptsize \vspace{-1em}	 	\begin{align*}
			E_{k+1}^e =& 
			\begin{cases}
				\begin{array}{ll}
					\frac{\Gamma_{\sigma(t_k)}}{N} {E_k^e} &\text{\cref{case_1}} \\
					\Gamma_{\sigma(t_k)} {E_k^e}  &\text{\cref{case_2}}\\
					\Gamma_{{\sigma}(t_k) \hat\sigma(t_k) }^3E_{k}^e&\text{\cref{case_7,case_8}}
				\end{array}
			\end{cases}\\
			x_{k+1}^{e*} = &
			\begin{cases}
				\begin{array}{ll}
					e^{\big(A_{ \sigma(t_k)} +B_{ \sigma(t_k)} K_{ \sigma(t_k)} \big)\tau_s} \hat{x}(t_k^+)&\text{\cref{case_1,case_2}} \\
					[\begin{matrix}
						\textbf{I} &\textbf{0}
					\end{matrix}
					] e^{ \mathcal{A}_{{\sigma}(t_k)\hat \sigma(t_k) } n \tau_s} \left[\begin{matrix}
						\textbf{I} \\
						\textbf{I}
					\end{matrix}
					\right]  \hat x( t_s)&\text{\cref{case_7,case_8}} 	
				\end{array}
			\end{cases}\\
			E_{k+1}^d = &
			\begin{cases}
				\frac{	\Gamma_{\hat{\sigma}(t_k)}}{N} E_k^d &\text{\cref{case_1}} \\
				\Gamma_{\hat \sigma(t_k)} {E_k^d}  &\text{\cref{case_2,case_7}}\\
				\Big(	\Gamma_{{\sigma}(t_k) \hat\sigma(t_k) }^3\Big)^n	E_{{k-n}}^d &\text{\cref{case_8}  } \\
			\end{cases}\\
			x_{k+1}^{d*}=  &
			\begin{cases}
				e^{\big(A_{\hat{\sigma}(t_k)}+B_{\hat{\sigma}(t_k)}K_{\hat{\sigma}(t_k)}\big)\tau_s} \hat{x}(t_k^+)&\text{\cref{case_1,case_2,case_7}} \\
				[\begin{matrix}
					\textbf{I} &\textbf{0}
				\end{matrix}
				] e^{ \mathcal{A}_{{\hat{\sigma}(t_k){\sigma}(t_k)}}n \tau_s}  \left[\begin{matrix}
					\textbf{I} \\
					\textbf{I}
				\end{matrix}
				\right] 	 \hat x( t_s)&\text{\cref{case_8}}  \\	
			\end{cases}\notag
	\end{align*} }
	the quantization level $N> \Gamma$ is odd,  DoS attack and dwell time satisfy
	\begin{equation}\label{equ_con1_1}
		\frac{N_{\max}\tau_s}{\tau_d}\log \frac{\hat{\Gamma}}{\Gamma}+ \left(\frac{1}{{T}}+\frac{\tau_s}{\tau_D}-1\right ) \log N+\log {\Gamma}\leq 0
	\end{equation}
	\begin{equation}\label{equ_con2_1}
		\tau_d \geq \max_{p,q \in \mathcal{M}}\left\{\frac{\log(\rho_p\xi_{pq})+ N_{\max} \log \frac{\eta_{pq}}{\lambda_p}}{-\log(\lambda_p)}
		\right\}\tau_s
	\end{equation}
	where $\Gamma_p  = \| e^{A_p \tau_s} \|$,  $\Gamma_{pq}^3 = \|e^{ \mathcal{A}_{pq}\tau_s}\| $,
	$\Gamma = \max_{p\in\mathcal{M}} \{ \Gamma_{p}\}$, $\hat{\Gamma} =\max_{p,q \in \mathcal{M} }\{\Gamma_{pq}^3 \}$,  $n = \frac{t_k-t_s}{\tau_s}$, $N_{\max}$ is the maximal number of asynchronous interval, and switching instant satisfies
	$	t_{s+1} = \left\{ k\tau_s>(t_s + \tau_d ) \right\},$
	then the switched system \eqref{equ_state} is stable under DoS attack.
\end{corollary}
	\begin{proof}
			Based on \Cref{lemma1,lemma2}, \eqref{equ_bound} and  \eqref{equ_error_bound} are satisfied. The quantization parameters are the same when there is no DoS attack. By repeating the proof for \Cref{lemme3}, it can be seen that $E_k^e \rightarrow 0 $ as time goes by. The convergence of  quantization parameters implies that the quantization error will also be convergent due to $\|e(t_k^+)\| \leq \frac{E_k^e}{N} $. When the state information  is successfully transmitted, the predicted state $\hat{x}(t_k^+) = e(t_k^+) + x(t_k) $ will approach the actual one. Once the error becomes sufficiently small, the system's performance will resemble that of a state-feedback system without a quantizer.
			The system will keep stable if \( A_{p} + B_{p}K_{p} \) is Hurwitz and the switching signal is appropriately designed. Next, we will show that \eqref{equ_con2_1} is a sufficient condition to ensure stability.
			
	
	From \eqref{equ_clo_asy_1}--\eqref{equ_clo_sy_2}, the system evolution during two consecutive switching instants is given by:
	$
	\| \tilde{x}( t_{s+1} )\| \leq \rho_p \lambda_p^{\tilde{n} - \hat{n}} \xi_{pq} \eta_{pq}^{\hat{n}} \| \tilde{x}( t_s ) \| + \sum_{i=0}^{\tilde{n} - \hat{n}-1} \rho_p \lambda_p^{\tilde{n} - \hat{n}-1-i} \| \tilde{B}_p^d\| E_{\lceil \hat{\overline{h}}_n\rceil+i}^e
	+ \sum_{i=0}^{\hat{n}-1} \rho_p \lambda_p^{\tilde{n} - \hat{n}} \xi_{pq} \eta_{pq}^{\hat{n}-1-i} \| \tilde{B}_{pq}^d(\tau_s) \| E_{ \hat{t}_s +i}^e
	$
	where $ \hat{n} $ represents the duration of asynchronous behavior and $ \tilde{n} - \hat{n} $ stands for that of synchronous behavior. Since the maximum asynchronous interval is $ N_{\max} $, we have $ \hat{n} \leq N_{\max} $. Given the definition of dwell time $ \tau_d $, it follows that $ \tilde{n} \geq \tau_d/\tau_s $.
	From \eqref{equ_con2_1}, we know that $ \rho_p \lambda_p^{\tilde{n} - \hat{n}} \xi_{pq} \eta_{pq}^{\hat{n}} \leq 1 $. Combining this with the fact that $ \lim_{k\rightarrow \infty} E_k^e \rightarrow 0 $ as well as the boundedness of  $ \sum\lim_{i=0}^{\hat{n}-1} \rho_p \lambda_p^{\tilde{n} - \hat{n}} \xi_{pq} \eta_{pq}^{\hat{n}-1-i} \| \tilde{B}_{pq}^d(\tau_s) \| $ and $ \sum_{i=0}^{\tilde{n} - \hat{n}-1} \rho_p \lambda_p^{\tilde{n} - \hat{n}-1-i} \| \tilde{B}_p^d \| $, we can get $ \| \tilde{x}( t_{s+1} ) \| \leq \| \tilde{x}( t_s ) \| $. Therefore, the system is stable as $ \lim_{t\rightarrow \infty} \| \tilde{x}( t ) \| \rightarrow 0 $.
	\end{proof}
	
	\color{blue}
	\begin{remark}
		The dwell time condition \eqref{equ_con2}  in \Cref{thm1} is strictly more restrictive than condition \eqref{equ_con2_1} in \Cref{coro1}. To see this, rewrite \eqref{equ_con2} as 
		$	\tau_d\geq \max_{p, q\in \mathcal{M}}   \bigg(\frac{ N_{\max}  \log \frac{\eta_{pq}}{\lambda_p}+\log( \rho_p \xi_{pq})+\log \frac{\mu_{pq}^1  \lambda_p }{\eta_{pq}}}{-\log \lambda_p} \bigg) \tau_s  $ with the definitions of $\mu_{pq}^1$, $\mu_{pq}^2$, $\mu_{pq}^3$, $\nu_p$ and $\hat \nu_p$ yielding  $\log \mu_{pq}^1 - \log \eta_{pq}+\log\lambda_p > 0$ because the effective increase rate $\log \mu_{pq}^1$ when switch occurs during two sampling instants  is larger than the increase rate $\log \frac{\eta_{pq}}{\lambda_p}$ when the system switches at some sampling instant.  Consequently, the lower bound on the dwell time implied by  \eqref{equ_con2} is larger than that required by \eqref{equ_con2_1}. Next, observer that $\Gamma_{pq}^1\geq \Gamma_{pq}^3$ implies $\overline{\Gamma} \geq \hat {\Gamma}^{N_{\max}}$. Moreover, the presence of $b$ defined in \eqref{equ_con1} means that DoS attack constraints can be relaxed in  \Cref{coro1} for the same dwell time. 
		Finally, \Cref{thm1} and \Cref{coro1} treat the worst-case scenario in which a switch occurs within the DoS intervals. Because an attacker cannot know the switching instants in advance, this scenario is unlikely in practice. In contrast to \Cref{thm1}, asynchronous behavior does not occur for attack-free intervals. Thus, \Cref{coro1} demonstrates that by designing an appropriate switching signal, the number of asynchronous behavior can be reduced.  This yields improved closed-loop performance and a faster convergence rate of the quantization range, since it will not jump at switching instants during the attack-free intervals.
	\end{remark}
	\color{black}

\subsection{Strategy 2}

In Subsection \ref{sec_quan1}, a quantizer with its center at the predicted state is designed. This quantizer calculates not only the quantization range but also the center, which requires higher computational resources. A quantizer centered at the origin, as studied in the existing literature, reduces computational burden and is more suitable for distributed sensors, as pointed out in \cite{output2}. However, this cannot be achieved by merely setting $x_k^{e*}$ to the origin. To this end, a new updating rule for the quantization parameters need be re-designed.
Similar to \Cref{lemma1,lemma2}, the update laws for the encoder and decoder are designed independently. In this strategy, a more detailed discussion is required to reduce conservatism. Specifically, \cref{case_1} is divided into two sub-cases: \cref{case_1}-\textbf{a} and \cref{case_1}-\textbf{b}.
\cref{case_1}-\textbf{a} (\cref{case_1}-\textbf{b}) represents the case where $ACK_{t_{k-1}}=0$ or $SY_{t_k}=1~ \& ~SY_{t_{k+1}} = 0$ ($ACK_{t_{k-1}}=1$). Here, an upper bound on DoS attack is given below.
\begin{assumption}[Maximum DoS attack interval]
	The duration for the $n$-th DoS attack  has an upper bound, i.e., $\tau_n \leq n_{\max}\tau_s$.
\end{assumption}
\begin{lemma}[update law of quantizer]\label{lemma4}
	For the switched system \eqref{equ_state}  with \textbf{\textit{Strategy 2}}, if the   update laws for quantizer obey
	
	{\scriptsize \vspace{-1em}	\begin{align}
			E_{k+1}^e =~& 
			\begin{cases}
				\begin{array}{ll}
					{\Lambda}_{\sigma(t_k)}^1  E_k^e &\text{\cref{case_1}-\textbf{a}} \\
					{\Lambda}_{\sigma(t_k)}^2  E_k^e &\text{\cref{case_1}-\textbf{b}~}\\
					{\Lambda}_{\sigma(t_k)}^3 E_k^e&\text{\cref{case_2}} \\
					{\Lambda}_{{\sigma}(t_k) \hat{\sigma}(t_k) }^4 E_{k}^e &\text{\cref{case_3}}\\
					{\Lambda}_{{\sigma}(t_k) \hat{\sigma}(t_k) }^5 E_k^e&\text{\cref{case_4,case_5,case_6}}\\
					{\Lambda}_{{\sigma}(t_k) \hat{\sigma}(t_k) }^6 E_{{k}}^e&\text{\cref{case_7,case_8}}\\
				\end{array}
			\end{cases}\label{equ_Eke_origin}\\
			E_{k+1}^d = ~&
			\begin{cases}
				\begin{array}{ll}
					{\Lambda}_{\sigma(t_k)}^1  E_k^d &\text{\cref{case_1}-\textbf{a}~}\\
					{\Lambda}_{\sigma(t_k)}^2  E_k^d&\text{\cref{case_1}-\textbf{b}~}\\
					{\Lambda}_{\sigma(t_k)}^3 {E_k^d}  &\text{\cref{case_2}} \\
					{\Lambda}_{{\sigma}(t_k) \hat{\sigma}(t_k) }^4 E_{k}^d &\text{\cref{case_3}}\\
					{\Lambda}_{\sigma(t_k)}^3 {E_k^d}  &\text{\cref{case_4,case_5,case_7}} \\
					{\Lambda}_{{\sigma}(t_k) \hat{\sigma}(t_k) }^5	\Big(	{\Lambda}_{{\sigma}(t_k) \hat{\sigma}(t_k) }^6\Big)^n  E_{{k-n}}^d	&\text{\cref{case_3,case_6,case_8}} \\
				\end{array}
			\end{cases}\label{equ_Ekd_origin}
	\end{align}}
	where   ${\Lambda}_p^1  = \rho_p \lambda_p + \frac{\rho_p\| B_p^d\| }{N} $, 
	${\Lambda}_p^2  =  \lambda_p + \frac{\rho_p\| B_p^d \| }{N} $,   
	$ {\Lambda}_p^3 = (\overline{\Psi}_p)^{\frac{1}{n_{\max}}}$, 
	$	\overline{\Psi}_p =	\max\limits_{\ell \in [0,n_{max}]} \Big \{\rho_{p} \lambda_p^\ell + \sum_{j=0}^{\ell -1} \rho_p \lambda_p^{\ell-j-1} \| B_p^d \| \frac{\| e^{{A}_p \tau_s} \| ^j }{N} \Big \}$, 
	$ {\Lambda}_{pq}^4 $ $= \max\limits_{\overline{t}\in [0,\tau_s)}\|e^{\mathcal{A}_{pq}   \overline{t}}  e^{\mathcal{A}_{p}(\tau_s - \overline{t})}\|$,
	$ {\Lambda}_{pq}^5 = \max \big \{ \frac{ 1- \lambda_p\underline{\Psi}_p\overline{\Psi}_p^{-1} }{\rho_p \| B_p^d\| },1 \big\}  {\Lambda}_{pq}^4  $, 
	$ {\Lambda}_{pq}^6 = 	 \| e^{\widetilde{\mathcal{A}}_{pq}\tau_s} \|$ and $n = \frac{\lceil t_{s} \rceil - \ulcorner t_s \urcorner }{\tau_s}$. 
	Then 
		$	\|x(t_k)\| \leq E_k^e ( k\geq 0)$
	and $E_k^d = E_k^e$ for $t_k \in \Theta(\infty,t_0)$. 
\end{lemma}
\begin{proof}
	Similar to the error norm in \Cref{lemma1},  we have 
		$	\|	e(t_k^+)\|\leq \frac{E_k^e}{N}$ for attack-free instants
	and
		$	\| e(t_k^+) \| = \| e(t_k^-) \| $ for attack instants.   

	\textbf{\cref{case_1}: $ACK_{t_k} = 1$ and $SY_{t_{k+1}}=1$.}
	
	From \eqref{equ_xtk1_ds}, we have 
	\begin{equation}\label{equ_xkl}
		x(t_{k+\ell}) = (A_p^d)^\ell x(t_k) + \sum_{j=0}^{\ell-1}  (A_p^d)^{\ell-j-1} B_p^d K_p e(t_{k+j}).
	\end{equation}
	It follows from  \eqref{equ_error_bound} and \eqref{equ_error} that 
	$	\| x(t_{k+\ell})\|  =\| (A_p^d)^\ell \| E_k^e+ \sum_{j=0}^{\ell-1} \| (A_p^d)^{\ell-j-1}\|  \| B_p^d \| \frac{E_{k+j}^e}{N} 
	\leq  \rho_{p} \lambda_p^\ell E_k^e+ \sum_{j=0}^{\ell-1} \rho_p \lambda_p^{\ell-j-1} \| B_p^d \| \frac{E_{k+j}^e}{N}. $
	From the first  and second formulas  in \eqref{equ_Eke_origin}, we have 
	$E_{k+\ell}^e = \rho_{p} \lambda_p ^{\ell} E_{k}^e+ \Delta \sum_{j=0}^{\ell-1}  \lambda_p^{\ell-j-1} E_{k+j}^e$
	where $\Delta = \frac{\rho_p \| B_p^d \| }{N}.$
	Therefore, we get 	$\| x(t_{k+\ell})\|  \leq E_{k+\ell}^e $.

	\textbf{\cref{case_2}:  $ACK_{t_k} = 0$ and $SY_{t_{k+1}}= 1$.}
	
	When a DoS attack occurs, the error $e(t)$ will not be updated at sampling instants. Based on \eqref{equ_etsol}, one has 
		$	\| e(t_{k+\ell})\|  \leq   \| e^{{A}_p \tau_s} \| ^\ell E_k^e$
	where $E_{k}^e$ is the quantization parameter at the sampling instant $t_{k}$ just before the DoS attack beginning instant, i.e.,  $ACK_{t_{k}} = 1$ and $ACK_{t_{k+1}} = 0$. 
	From \eqref{equ_xkl}, we know 
	\begin{align*} 
		\| x(t_{k+\ell})\|  =  ~& \| (A_p^d)^\ell \| E_k^e+ \sum_{j=0}^{\ell-1} \| (A_p^d)^{\ell-j-1}\|  \| B_p^d \| \frac{\| e^{{A}_p \tau_s} \| ^jE_k ^e}{N} \\
		\leq ~&  \rho_{p} \lambda_p^\ell E_k^e+ \sum_{j=0}^{\ell-1} \rho_p \lambda_p^{\ell-j-1} \| B_p^d \| \frac{\| e^{{A}_p \tau_s} \| ^jE_k^e }{N} \\
		=:~& \Psi_p(\ell) E_k^e.
	\end{align*}
	Combining it with the maximum DoS attack duration, the maximum increase of $E_k^e$ is 
	$E_{k+n_{\max}} ^e \leq \overline{\Psi}_p E_k^e$ with 
		$\overline{\Psi}_p =	\max_{\ell\in [0,n_{max}]} \big \{\rho_{p} \lambda_p^\ell + \sum_{j=0}^{\ell-1} \rho_p \lambda_p^{\ell-j-1} \| B_p^d \| \frac{\| e^{{A}_p \tau_s} \| ^j }{N} ,1 \big \}$.
	Since $\lambda_p <1$, $ \Psi_p(\ell) $ may decrease and then increase as $\ell$ increases, or it may increase monotonously. It can be obtained that 
		$	E_{k+\ell}^e \leq (\overline{\Psi}_p)^{\frac{\ell}{n_{\max}}}E_k^e$.
	Hence, $\| x(t_{k+\ell})\|  \leq E_{k+\ell}^e $.

	\textbf{\cref{case_3}:  $ACK_{t_k} = 1$, $ACK_{t_{k+1}} = 1$ and $SY_{t_{k+1}}= 0$.}
	
	Suppose that the system mode is $q$, i.e.,  $\sigma(t) = q$ and the predictor mode is $p$, $\hat{\sigma}(t) = p$. Define an auxiliary vector $z(k)$ as that in \Cref{lemma1}. From \eqref{equ_equ_asy_error_asy},  one gets
	\begin{equation*}
		\begin{aligned}
			\| z(t_{k+1}) \| \leq ~&\|e^{\mathcal{A}_{pq}  \overline{t}}  e^{\mathcal{A}_{p} (\tau_s - \overline{t})} \| \| z(t_k^+)\|  \\
			\leq~ & \max_{\overline{t}\in [0,\tau_s)}\|e^{\mathcal{A}_{pq}  \overline{t}}  e^{\mathcal{A}_{p} (\tau_s - \overline{t})} \| \| z(t_k^+)\| \\
			\leq ~ &{\Lambda}_{pq}^4   E_{k}^e  =: E_{k+1}^e 
		\end{aligned}
	\end{equation*}
	where ${\Lambda}_{pq}^4  = \max_{\overline{t}\in [0,\tau_s)}\|e^{\mathcal{A}_{pq} \overline{t} }  e^{\mathcal{A}_{p} (\tau_s - \overline{t})} \|$. The third inequality is derived according to  $\|x(t_k)\| \leq E_k^e$ and $\| \hat x(t_k^+)\| = \|c_k\| \leq E_k^e$.

	\textbf{\cref{case_4,case_5,case_6}: $SY_{t_{k}}= 1$, $SY_{t_{k+1}}= 0$ and $ACK_{t_k} \times ACK_{t_{k+1}} = 0$.}
	
	From \eqref{equ_ez}, the state norm admits
		$	\| 	e_z(t_{k+1}) \| \leq  \| e^{\tilde{\mathcal{A}}_{pq} \overline{t}} e^{\overline{A}_{qq} (\tau_s-\overline{t})} \|  \| e_z(t_k)\|$.
	Due to the definition of infinity norm, it is true that $\| e_z(t_k)\| \leq \max\{ \| x(t_k)\| , \| e(t_k)\| \}$ and $\| x(t_{k+1})\| \leq \| e_z(t_{k+1})\| $. Next we calculate the norm of $e(t_k)$.   
	
	Let $E_{k-\ell}^e$ be the quantization parameter at the sampling instant $t_{k-\ell}$ just before the DoS attack starting instant, i.e.,  $ACK_{t_{k-\ell}} = 1$ and $ACK_{t_{k-\ell+1}} = 0$. Thus, $\| e(t_k) \| \leq \frac{\| e^{A_p\tau_s} \| ^{\ell} }{N} E_{k-\ell}^e$.
	
	Using the definition of $\Psi_p(\ell)$, we have 
		$	\Psi_p(\ell +1) - \lambda_p \Psi_p(\ell) = \rho_p \| B_p^d\| \frac{\| e^{A_p \tau_s}\|^{\ell}}{N}$.
	Let the minimum of $\Psi_p(\ell)$ as $\underline{\Psi}_p:=	\min\limits_ {l\in [0,n_{max}]} \Big \{\rho_{p} \lambda_p^l + \sum_{j=0}^{l-1} \rho_p \lambda_p^{l-j-1} \| B_p^d \| \frac{\| e^{{A}_p \tau_s} \| ^j }{N} \Big \}.$
	It yields 
		$	\rho_p \| B_p^d\| \frac{\| e^{A_p \tau_s}\|^{\ell}}{N} \leq  \overline{\Psi}_p^{\frac{l}{n_{\max}}} -\lambda_p\underline{\Psi}_p$.
	Then the upper bound of error becomes
		$	\| e(t_k) \| 
		\leq {\frac{1}{\rho_p \| B_p^d\| }  E_{k}^e-\frac{1 }{\rho_p \| B_p^d\| }  \lambda_p\underline{\Psi}_p} \overline{\Psi}_p^{\frac{-\ell}{n_{\max}}}E_{k}^e.$
		Since $\overline{\Psi}_p \geq 1$ and $\ell \in [0, n_{\max}]$, we have $\overline{\Psi}_p^{-1} \leq  \overline{\Psi}_p^{\frac{-\ell}{n_{\max}}} \leq 1$. Hence, one has 
		$		\| e(t_k) \| 
		\leq {\frac{1}{\rho_p \| B_p^d\| }  E_{k}^e-\frac{1 }{\rho_p \| B_p^d\| }  \lambda_p\underline{\Psi}_p} \overline{\Psi}_p^{-1}E_{k}^e$.
		In terms of the fact  $\| x(t_k)\| \leq E_k^e$, we know 
		\begin{align}
			~&	\| x(t_{k+1}) \| \leq \| e_z(t_{k+1})\|\label{equ_xtk1} \\
			\leq~&  \| e^{\tilde{\mathcal{A}}_{pq} \overline{t}} e^{\overline{A}_{qq} (\tau_s-\overline{t})} \| \max \Big \{ \frac{ 1- \lambda_p\underline{\Psi}_p\overline{\Psi}_p^{-1} }{\rho_p \| B_p^d\| },1 \Big \} E_k^e \notag \\
			=:~& {\Lambda}_{pq }^5E_k^e  = E_{k+1}. \notag
		\end{align}
		This further deduces $\|x (t_k)\| \leq E_{k}^e$. 
		
		\textbf{\cref{case_7,case_8}: $SY_{t_{k}}= 0$, $SY_{t_{k+1}}= 0$ and $ACK_{t_k} = 0$.} 
		
		Similar to \cref{case_5}, the norm of the state is 
		\begin{equation}
			\| e_z(t_{k+1})\| \leq \| e^{\widetilde{\mathcal{A}}_{pq}\tau_s} \| \| e_z(t_k)\| \leq \| e^{\widetilde{\mathcal{A}}_{pq}\tau_s} \|E_{k}^e=: {\Lambda}_{pq }^6 E_k^e
		\end{equation}
		which can be derived according to $\| e_z(t_k)\| \leq E_k^e$ from  \eqref{equ_xtk1}.
		
		With the above analysis, it is clear that $\|x(t_k)\| \leq E_k^e ( k\geq 0)$.
		The analysis of uniformity between $E_k^e$ and $E_k^d$ can be performed in a similar way thus it is omitted for brevity.  
	\end{proof}
	
	In this strategy, we focus on the status of DoS attack at instant $t_{k-1}$. Such a treatment is motivated by the presence of an additional constant, $\rho_p$, that arises at the first instant after a DoS attack due to the inequality $\| (A_p^d)^\ell \| \leq \rho_p \lambda_p^\ell$. Additionally, when a system is suffering from a DoS attack, the error between the actual state and the predicted state becomes larger since the predictor cannot be updated. However, if the error is sufficiently small, the state may still be convergent. Specifically, the quantization parameter $E_k^e$ may decrease during a DoS attack, as noted in \cite{Rui3}. If the value at the instant just before the DoS attack is recorded, the update law for $E_k^e$ can incorporate $\Psi_p(\ell)$, which reduces conservatism at the cost of greater complexity.

	\begin{theorem}[Stability]\label{thm1_ori}
		For  the switched system \eqref{equ_state}  with \textbf{\textit{Strategy 2}}, if the quantizer updates its parameters according to \eqref{equ_Eke_origin} and \eqref{equ_Ekd_origin}, the  quantization level $N \geq \max\limits_{p \in \mathcal{M}} \left \{ \frac{\rho_p \| B_p^d\|}{1-\lambda_p}\right\}$ is odd,  DoS attack and the switching signal satisfy
		\begin{equation}\label{equ_con1_origin}
			{\footnotesize \begin{aligned}
					& \frac{1}{\tau_d}\Big ( \log \frac{\overline{\Lambda}^5}{\overline{\Lambda}^6}+ N_{\max} \log \frac{\overline{\Lambda}^6}{\overline{\Lambda}^3}
					\Big ) + \frac{1}{\tau_D} \log \frac{\overline{\Lambda}^1}{\overline{\Lambda}^2}\\
					&+ \big(\frac{1}{{T\tau_s}}+\frac{1}{\tau_D}\big)\log \frac{\overline{\Lambda}^3}{\overline{\Lambda}^2} + \frac{1}{\tau_s} \log \overline{\Lambda}^2 \leq 0
			\end{aligned}}
		\end{equation}
		\begin{equation}\label{equ_con2_origin}
			{\tiny 	\tau_d \geq \max_{p,q \in \mathcal{M}}\Big\{\frac{\log(\rho_p\xi_{pq} \mathbb{A}_{pq}^d)+ N_{\max} (\log \eta_{pq}-\log\lambda_p)}{-\log(\lambda_p)}
				\Big\}\tau_s}
		\end{equation}
		where $\overline{\Lambda}^{i} = \max_{p \in \mathcal{M}} \overline{\Lambda}^{i}_p (i\in \{1,2,3\})$, $\overline{\Lambda}^{i} = \max\limits_{p,q \in \mathcal{M}, p\neq q} \overline{\Lambda}^{i}_{pq} (i\in \{5,6\})$, 
		$\mathbb{A}_{pq}^d  = \max_{\overline{t} \in [0,\tau_s]} \| \tilde{A}_{pq}^d (\overline{t})\|$,  and $N_{\max}$ is the maximal number of asynchronous interval,	then the switched system \eqref{equ_state} is stable under DoS attack.
	\end{theorem}
	
	The proof of this theorem is similar to that of \Cref{thm1}. First, the convergence of the quantization parameter $E_k^e$ is guaranteed by \eqref{equ_con1_origin}, specifically addressing the worst-case scenario where a switch occurs between two consecutive sampling instants during DoS attack intervals. In this case, the asynchronous duration reaches its maximum value $N_{\max}$. The analysis of the first instant after DoS attack is inspired by \cite{output2}, although, unlike \cite{output2}, this situation also arises at switching instants without DoS attack.
	However, because of $\Lambda_{pq}^4 \leq \Lambda_{pq}^5$, a switch occurring during a DoS attack causes a larger increase of quantization parameter $E_k^e$, representing the worst case. To maintain stability, the number of \textit{zooming-out}, using the first equation in \eqref{equ_Eke_origin}, matches the number of DoS attack, $n(t_k,t_0)$. Repeating the analysis in \Cref{lemma1}, \eqref{equ_con1_origin} ensures that $\lim\limits_{k\rightarrow \infty} E_k^e \rightarrow 0$.
	The convergence of the state can then be obtained using a similar approach to \Cref{coro1}, with the main difference being the case where a switch occurs between two consecutive sampling instants. This issue can be addressed following the analysis in \Cref{coro1}. The detailed proof is omitted here.
	
	\begin{remark}
		In \textbf{\textit{Strategy 2}}, the quantization center is set to be the origin. Once the quantization parameter $E_k^e$ converges to zero, the system achieves stability, as indicated by $\| x(t_k)\| \leq E_{k}^e$. However, \eqref{equ_con1_origin} characterizes a coupling relationship among DoS attack constraints,  switching signal constraint, and quantization parameters. In contrast, \eqref{equ_con2_origin} is solely related to the system matrices. This sufficient condition is derived from the convergence of the quantization parameter $E_k^e$, which poses a stricter dwell time constraint.
		Here, we propose a method to calculate the switching signal and DoS attack constraints. First, by utilizing \eqref{equ_con2_origin}, we can determine the dwell time. Then, by applying the quantization level condition and taking into account  the physical mechanism, we can get the quantization level $N$. In this scenario, DoS attack constraints may become overly conservative or even infeasible. Therefore, increasing $N$ or enlarging the dwell time could help relax DoS attack constraints.
	\end{remark}

	\begin{remark}
		The first sufficient condition will boil down to that for non-switched systems, specifically, when $\tau_d \rightarrow \infty$. Compared to \cite{output1}, the condition in \eqref{equ_con1_origin} is similar. More specifically, the output feedback control with an observer was employed in \cite{output1}. After matrix augmentation, the closed-loop system under DoS attack is represented in zero-input form, which simplifies the analysis.
		However, in this paper, the control input is not equal to zero when a DoS attack occurs. The update law of quantization parameters during DoS attack poses a challenging task. An active controller could mitigate state divergence during DoS attack, as it may not lead to an increase of $E_k^e$.
		Furthermore, the asynchronous behavior in \cref{case_5,case_6} brings certain additional difficulty. The augmentation matrix is utilized to characterize the closed-loop system, but the evolution of the quantization error $e(t_k)$ is included when calculating state norm. In the absence of DoS attack, we have $\| e(t_k)\| \leq \frac{E_k^e}{N}$. However, this condition does not hold during DoS attack intervals. The value of $\|e(t_k)\|$ becomes dependent on the duration of DoS attack. To address this issue, an upper bound is computed that is solely related to the latest value of $E_k^e$.
	\end{remark}

\section{Passive quantized control strategy}\label{sec_passive}
In the aforementioned control strategies, the controller continuously generates control signal based on the predicted state. However, in certain applications, ZOH mechanism is utilized to generate the control signal during the intervals between sampling instants. Furthermore, if DoS attack occurs, the sensor signal becomes zero, resulting in a zero control signal, see \cite{Rui2}. This approach belongs to the passive control framework, as the defenders do not actively respond to attacks.
In the sequel, \textbf{\textit{Strategy 3}} is devised similar to \textbf{\textit{Strategy 2}}. Then based on the simplest \textbf{\textit{Strategy 3}}, two quantized control strategies without ACK signal are investigated in terms of time-triggered and event-triggered approaches. 

	\subsection{Closed-loop System with Passive Controller}
The closed-loop system for \textit{passive controller} becomes simpler than that for \textit{active controller}. 
Unlike closed-loop dynamics with \textit{active controller}, the asynchronous interval is shorter than one period during DoS attack intervals, as the control signal becomes zero. Moreover, the system dynamics between the attack and the attack-free scenarios is different.

$\bullet$ Synchronous stage for DoS attack-free case:
The state solution for DoS attack-free case turns to be
\begin{equation}\label{equ_close_1_zero}
	{x} (t_{k+1}) = \hat{A}_p^d x(t_k)  + \hat{B}_p^d u(t_k)
\end{equation}
where $\hat{A}_p^d = e^{A_p \tau_s}$, $\hat{B}_p^d = \int_0^{\tau_s} e^{A_ps}B_pds$ and $\sigma(t_k) = p$.
According to $u(t_k) = K_pc_k =K_p(x(t_k) - e(t_k))  $, we have 
$	x(t_{k+1})  = \hat{A}_p^{cl}x (t_k) + \hat B_p^dK_pe(t_k)$
where $\hat{A}_p^{cl} = \hat{A}_p^d + \hat B_p^dK_p$.

$\bullet$ Synchronous stage for DoS attack case:
For this case, the state is
$	x(t_{k+1}) =  \hat{A}_p^d x(t_k) .$

$\bullet$ Asynchronous stage for DoS attack-free case:
Suppose that $p= \sigma(t),~q= \hat{\sigma}(t)$ for all $t\in [t_s,t_{k+1})$.
The state solution for attack-free case obeys
$x(t_{k+1}) 
= \hat{A}_{pq}^d(\overline{t}) x(t_k) + \hat{B}_{pq}^d (\overline{t}) u(t_k) $
where  $ \overline{t} = t_{k+1}-t_s$, $\hat{A}_{pq}^d(\overline{t}) =e^{A_p \overline{t} +A_q( \tau_s-\overline{t} )} $ and $\hat{B}_{pq}^d (\overline{t})  = e^{A_p \overline{t}} \int_{0}^{\tau_s-\overline{t}} e^{A_q s } B_q ds+\int_{0}^{\overline{t}} e^{A_p s } B_p ds $. Then the closed-loop system can be rewritten as
\begin{equation}\label{equ_xtk_asy_zero_1}
	x(t_{k+1}) =  \hat{A}_{pq}^{cl}(\overline{t}) x(t_k) + \hat{B}_{pq}^d (\overline{t}) K_{q} e(t_k)
\end{equation}
where $\hat{A}_{pq}^{cl} (\overline{t}) =   \hat{A}_{pq}^{d}(\overline{t}) +  \hat{B}_{pq}^d (\overline{t}) K_{q}$.

$\bullet$ Asynchronous stage for DoS attack case:
The system dynamics  admits
\begin{equation}\label{equ_clo_asy_2_ori}
	x(t_{k+1}) = e^{A_p \overline{t}+A_q (\tau_s-\overline{t})} x(t_k) =\hat {A}_{pq}^d (\overline{t}) x(t_k)
\end{equation}

Similarly, there exist constants $0<\hat{\lambda}_p <1$, $\hat \eta_{p} ,~\hat \xi_{p}>0$ and $\hat{\rho}_p>0$ such that 	$ \| (\hat{A}_p^{cl})^k\| \leq \hat {\rho}_{p} \hat \lambda_p^k$ and $ \| (\hat{A}_p^d)^k\| \leq \hat \xi_{p}  \hat\eta_{p}^k$  for all mode.

\subsection{Strategy 3}
Similar to the analysis in Section \ref{sec_active}, we propose the update laws for quantization parameters to ensure that the state stays within the quantization range. Subsequently, we derive a sufficient condition related to DoS attack constraints and switching signal to guarantee the stability. Following the treatment of \cref{case_1}-\textbf{a} and \cref{case_1}-\textbf{b}, we perform \cref{case_2}-\textbf{a} and \cref{case_2}-\textbf{b}. 
Here, \cref{case_2}-\textbf{a} corresponds to the scenario where \(ACK_{t_{k-1}}=1\), while \cref{case_2}-\textbf{b} deals with the scenario where \(ACK_{t_{k-1}}=0\).

\begin{lemma}[update law of quantizer]\label{lemma5}
	For the switched system \eqref{equ_state}  with \textbf{\textit{Strategy 3}}, if the quantizer updates according to 
	{\scriptsize \begin{align}
			E_{k+1}^e = &
			\begin{cases}
				\begin{array}{ll}
					{\Upsilon}_{\sigma(t_k)}^1  E_k^e &\text{\cref{case_1}-\textbf{a}} \\
					{\Upsilon}_{\sigma(t_k)}^2  E_k^e &\text{\cref{case_1}-\textbf{b}~}\\
					\hat{\xi}_{\sigma(t_k)}  \hat \eta_{\sigma(t_k)}	E_{k}^e & \text{\cref{case_2}-\textbf{a}} \\
					\hat \eta_{\sigma(t_k)}	E_{k}^e & \text{\textbf{Cases \ref{case_2}}-\textbf{b}, \textbf{\ref{case_7}}, \textbf{\ref{case_8}}} \\
					\Upsilon _{{\sigma}(t_k) \hat{\sigma}(t_k) }^3E_{k}^e& \text{\caref{case_3} }\\
					\widehat {\Upsilon }_{{\sigma}(t_k) \hat{\sigma}(t_k) } \tilde{\xi}_{{\sigma}(t_k)} E_{k}^e& \text{\cref{case_4,case_5}}\\
					\Upsilon _{{\sigma}(t_k) \hat{\sigma}(t_k) }^4	E_{k}^e &\text{\cref{case_6}}\\
				\end{array}
			\end{cases}\label{equ_Eke_zero}\\
			E_{k+1}^d = &
			\begin{cases}
				\begin{array}{ll}
					{\Upsilon}_{\sigma(t_k)}^1  E_k^d &\text{\cref{case_1}-\textbf{a}} \\
					{\Upsilon}_{\sigma(t_k)}^2  E_k^d &\text{\cref{case_1}-\textbf{b}~}\\
					\hat {\xi}_{\sigma(t_k)}  \hat \eta_{\sigma(t_k)}	E_{k}^d & \text{\textbf{Cases \ref{case_2}}-\textbf{a} and  \textbf{\ref{case_4}} } \\
					\hat \eta_{\sigma(t_k)}	E_{k}^d & \text{\textbf{Cases \ref{case_2}}-\textbf{b}, \textbf{\ref{case_5}}, \textbf{\ref{case_7}}} \\
					\Upsilon _{{\sigma}(t_k) \hat{\sigma}(t_k)}^3 E_{k}^d &  \text{\cref{case_3} }\\
					\Upsilon _{{\sigma}(t_k) \hat{\sigma}(t_k) }^4 	E_{k}^d &\text{\cref{case_6}}\\
					\widehat {\Upsilon }_{{\sigma}(t_k) \hat{\sigma}(t_k) } \tilde{\xi}_{{\sigma}(t_k)}  \hat \eta_{\sigma(t_k)}^ {n} E_{{k-n}}^d &\text{\cref{case_8}}\\
				\end{array}
			\end{cases}\label{equ_Ekd_zero}
	\end{align}}
	where $\tilde{\xi}_{p}  = \max\{\hat{\xi}_{p} ,1\}$,
	${\Upsilon}_p^1  = \hat \rho_p \hat \lambda_p + \frac{\hat \rho_p\| 
		\hat {B}_p^dK_p\| }{N} $, ${\Upsilon}_p^2  =  \hat \lambda_p + \frac{\hat\rho_p\| \hat B_p^d K_p\| }{N} $
	and $\Upsilon _{pq}^3 = \max _{\overline{t} \in [0,\tau_s]} \left( \| \hat{A}_{pq}^d(\overline{t})+ \hat{B}_{pq}^d (\overline{t}) K_{q}\frac{  N-1}{N}  \| \right) $, $ {\Upsilon}_{pq}^4 = \max_{\overline{t} \in [0,\tau_s]} \| \tilde{A}_{pq}^d(\overline{t})\|$ and $	\widehat {\Upsilon }_{pq}   =\max \{	\Upsilon _{pq }^3, 	\Upsilon _{pq }^4\}$. Then one gets  
	$
	\|x(t_k)\| \leq E_k^e,~ k\geq 0
	$ and $E_k^d = E_k^e$ for $t_k \in \Xi(\infty,t_0)$. 
\end{lemma}

\begin{proof}
	In what follows, we also discuss the cases listed in Table \ref{tab1} one by one. 
	
		\textbf{\cref{case_1}: Sampling instant $t_k$  is attack free and the system is synchronous, i. e., $ACK_{t_k} = 1$ and $SY_{t_{k+1}}=1$.}
		
		Similar to the \cref{case_1} in  \Cref{lemma4}, we can find $\| x(t_{k+l})\|  \leq E_{k+l}^e $.

	\textbf{\cref{case_2}: Sampling instant $t_k$  is under attack and the system is synchronous, i. e., $ACK_{t_k} = 0$ and  $SY_{t_{k+1}} = 0$.}
	
	Suppose that the system mode is $p$, i.e.,  $\sigma(t) = p$,  then the state satisfies 
	$	\| {x} (t_{k+n}) \| \leq  ~   \| (\hat A_p^d )^nx(t_k) \| $.
	It is easy to obtain $\| x(t_{k+l})\|  \leq E_{k+l}^e $ based on $ \| (\hat{A}_p^d)^k\| \leq \hat \xi_{p}  \hat \eta_{p}^k$.
	
		\textbf{\cref{case_3,case_4}:  Sampling instant $t_k$  is  without attack and the system is asynchronous, i.e., $ACK_{t_k} = 1$ and $SY_{t_{k+1}}= 0$.}

From (\ref{equ_xtk_asy_zero_1}), we have 
\begin{align*}
	\| x(t_{k+1}) \| =~& \|  \hat{A}_{pq}^{cl}(\overline{t}) \| \| x(t_k)\|  + \| \hat{B}_{pq}^d (\overline{t}) K_{q} \| \| e(t_k)\| \\
	\leq ~& \|  \hat{A}_{pq}^{cl}(\overline{t}) \|E_k^e + \| \hat{B}_{pq}^d (\overline{t}) K_{q} \| \frac{N-1}{N} E_k^e\\
		 \leq ~&  \Upsilon _{pq}^3E_k^e 
\end{align*}
	where $\Upsilon _{pq}^3 = \max _{\overline{t} \in [0,\tau_s]} \left(\| \hat{A}_{pq}^d(\overline{t})\|+\| \hat{B}_{pq}^d (\overline{t}) K_{q}\|\frac{  N-1}{N}  \right) $. Due to $\widehat {\Upsilon }_{pq}   =\max \{	\Upsilon _{pq }^3, 	\Upsilon _{pq }^4\}$, $\| x(t_{k+1})\| \leq E_{k+1}^e$ holds for two cases.

	\textbf{\cref{case_5,case_6}: The system switches during this interval and a DoS attack occurs, i. e., $SY_{t_{k}}= 1$, $SY_{t_{k+1}}= 0$ and  $ACK_{t_k} = 0$.}
	
	From (\ref{equ_clo_asy_2_ori}), it is obvious that 
$		\| x(t_{k+1})\| \leq  \| \tilde {A}_{pq}^d(\overline{t})\|  \| x(t_k)\| \leq {\Upsilon}_{\sigma(t_k)}^4 E_k^e$
	where $ {\Upsilon}_{\sigma(t_k)}^4 = \max_{\overline{t} \in [0,\tau_s]} \| \tilde{A}_{pq}^d(\overline{t})\|$. Due to $\tilde{\xi}_{p} \geq 1$, we have  $\| x(t_{k+1})\| \leq E_{k+1}^e$ in the two cases. 
	
	\textbf{\cref{case_7,case_8}: The system does not switch during this interval and a DoS attack occurs, i. e., $SY_{t_{k}}= 0$, $SY_{t_{k+1}}= 0$ and  $ACK_{t_k} = 0$.}
	
The analysis is the same as that for \caref{case_2}. The jump at the  first instant of DoS attack can be efficiently tackled at the switching interval.  

The analysis on uniformity of quantization parameters can be directly obtained and it is omitted here. The proof is complete.
\end{proof}

\color{black}
\begin{theorem}[Stability]\label{thm1_zero}
	For  the switched system \eqref{equ_state}  using \textbf{Strategy 3}, if the quantizer updates its parameters according to \eqref{equ_Eke_zero} and \eqref{equ_Ekd_zero}, the  quantization level $N \geq \max\limits_{p \in \mathcal{M}} \left \{ \frac{\hat \rho_p \| \hat B_p^dK_p\|}{1-\hat \lambda_p}\right\}$ is odd,  DoS attack and switching signal satisfy
	\begin{equation}\label{equ_con1_zero}
		\begin{aligned}
			\frac{1}{\tau_d} \log \frac{\tilde{\Upsilon}}{\tilde{\Upsilon}^2} + \frac{1}{\tau_D} \log \frac{\tilde{\Upsilon}^1\overline{\xi}}{\tilde{\Upsilon}^2} + \big(\frac{1}{{T\tau_s}}+\frac{1}{\tau_D}\big)  \log \frac{\overline{\eta}}{\tilde{\Upsilon}^2}+ \frac{\log \tilde{\Upsilon}^2}{\tau_s}  \leq 0
		\end{aligned}
	\end{equation}
	where $ \tilde{\Upsilon} = \max\limits_{p, q\in \mathcal{M}} \Big\{ \widehat {\Upsilon }_{pq} \tilde{\xi}_{p},\frac{\Upsilon_{pq}^3 \Upsilon_p^1}{\Upsilon_{p}^2}\Big\} $, $\tilde{\Upsilon}^i = \max\limits_{p\in\mathcal{M}}\{{\Upsilon}^i_p\} ~(i\in \{1,2\})$, $\overline{\xi}  = \max_{p \in \mathcal{M}} \tilde{\xi}_{p}$, $\overline{\eta}  = \max_{p\in \mathcal{M}} \hat{\eta}_{p}$ and $N_{\max}$ is the maximal number of asynchronous interval,	then the switched system \eqref{equ_state} is stable under DoS attack.
\end{theorem}

\begin{proof}
	Since the effect of switching is just one sampling period for  attack and attack-free cases, the effect of each case should be discussed first. From the updating law of $E_k^e$, the increase of $E_k^e$ in \cref{case_4,case_5} is larger than \cref{case_7}.  We consider the switch occurring in  \cref{case_4,case_5} during DoS attack interval. Next, we compare \cref{case_3,case_5}. For \caref{case_3}, the switch will lead to \cref{case_1}-\textbf{a}.  Then the  increase of $E_k^e$ of switch for \caref{case_3} is $\frac{\Upsilon_{pq}^3 \Upsilon_p^1}{\Upsilon_{p}^2}$. Similarly, for \cref{case_5}, the  increase of $E_k^e$ is $\widehat {\Upsilon }_{pq} \tilde{\xi}_{p} $. Hence, the maximum  increase of $E_k^e$ caused by switching is $ \tilde{\Upsilon} = \max\limits_{p, q\in \mathcal{M}} \big\{ \widehat {\Upsilon }_{pq} \tilde{\xi}_{p},\frac{\Upsilon_{pq}^3 \Upsilon_p^1}{\Upsilon_{p}^2}\big\} $. 
	
	Similar to the proof of \Cref{lemma1}, the dynamics of $E_{k}^e$ obeys
	$	E_{k}^e \leq \tilde{\Upsilon} ^{N_{\sigma(t)}(t_k,t_0)} (\tilde{\Upsilon}^1)^{n(t_k,t_0)} (\overline{\xi} \overline{\eta})^{n(t_k,t_0)}  \overline{\eta}^{| \overline\Xi(t_k,t_0)| /\tau_s-n(t_k,t_0)} 
		 \cdot (\tilde{\Upsilon}^2)^{\frac{t_k-t_0}{\tau_s} - N_{\sigma(t)(t_k,t_0)}-n(t_k,t_0)-| \overline \Xi(t_k,t_0)| /\tau_s} 
		 = c d^{t_k-t_0}E_{0}$
	where $d = 	\left(\frac{\tilde{\Upsilon} }{\tilde{\Upsilon}^2}\right)^{\frac{1}{\tau_d}}\left(\frac{\tilde{\Upsilon}^1\overline{\xi} }{\tilde{\Upsilon}^2}\right)^{ \frac{1}{\tau_D}}  
	\left(\frac{\overline{\eta}}{\tilde{\Upsilon}^2}\right)^{\frac{1}{\tau_s}(\frac{1}{T}+\frac{\tau_s}{\tau_D})} (\tilde{\Upsilon}^2)^{\frac{1}{\tau_s} }$ and  $c = 	\left(\frac{\tilde{\Upsilon}^1\overline{\xi} }{\tilde{\Upsilon}^2}\right)^{n_0}\left(\frac{\overline{\eta}}{\tilde{\Upsilon}^2}\right)^{\overline{\kappa}}$. From (\ref{equ_con1_zero}), we have $d\leq 1$. Then, the quantization $E_k^e$ converges. Using the result in  \Cref{lemma5}, one has $\| x(t_k)\| \leq E_k^e$. Therefore $ \lim_{k\rightarrow \infty}\| x(t_k)\|  \rightarrow 0$.
\end{proof}

\color{black}
\begin{remark}
	Unlike the proof for \Cref{thm1_ori}, the dynamics of the state is not considered in the stability analysis here. It is evident that the dynamics with and without DoS attack are the same when switched systems are synchronous in \textbf{\textit{Strategy 2}}. Therefore, the sufficient conditions derived from the system state in \Cref{thm1_ori} are independent of DoS attack, which facilitates the calculation of the switching signal.
\end{remark}
\color{blue}
\begin{remark}
	As can be seen from \eqref{equ_con1}, \eqref{equ_con1_origin} and \eqref{equ_con1_zero}, DoS attacks directly influence the convergence behavior of the quantization range. More specific, the convergence rate during DoS attack-free periods and the divergence rate over DoS attack intervals are critical for closed-loop stability. On the one hand, this issue has close relationship with the quantization error during attack-free intervals, which depends on the quantization range and level. In contrast to \textbf{\textit{Strategies 2--3}}, the quantizer in \textbf{\textit{Strategy 1}} yields a smaller quantization error since the distance between a well-designed quantizer center and the actual state is shorter than that between the origin and the state. Moreover, the update law of quantizers reveal that the convergence of $E_k^e$ relies on $\frac{\Gamma}{N}$ in \textbf{\textit{Strategy 1}}, whereas it is governed by $ \lambda_p + \frac{\rho_p\|  B_p^d \| }{N}$  and $\hat \lambda_p + \frac{\hat\rho_p\| \hat B_p^d K_p\| }{N} $ in \textbf{\textit{Strategies 2--3}}, respectively. 
	Consequently,\textbf{\textit{ Strategy 1}} can relax DoS attack constraints via tuning $N$, while this flexibility is not explicitly transparent for \textbf{\textit{Strategies 2--3}}. 	However, this direct relationship does not hold for \textbf{\textit{Strategies 2--3}}. On the other hand,  \textbf{\textit{Strategy 1}} deploys a predictor-based controller that actively compensates for packet losses induced by DoS attacks. 
\end{remark}

\begin{remark}
	The quantization level $N$, the DoS attack constraints and the dwell time are tightly coupled. Raising $N$ reduces quantization errors, thereby accelerating the convergence of quantization parameters, permitting longer dwell time, and relaxing DoS attack bounds. 
	Under a fixed dwell time and attack intensity, a larger $N$ also yields a faster convergence rate. Conversely, while a larger dwell time enhances resilience to against severe DoS attacks, it lowers the switching frequency and may degrade steady-state performance.  To retain stability against high-frequency or protracted DoS attacks, both $N$ and dwell time must be enlarged. Nevertheless, longer attack intervals generally compromise overall performance.
\end{remark}
\subsection{ Strategy 4}
\color{black}
The above updating rules are relevant to the ACK signal, which is commonly implemented in applications such as TCP network. However, several networks that do not utilize the ACK technique, like UDP network, are also prevalent in engineering applications. In the absence of an ACK signal, the \textit{zooming-in} and \textit{zooming-out} schemes become inapplicable, which results in the state potentially falling outside the quantization range.
This issue has not been fully addressed yet. For instance, \cite{Liu2022} tackled it by setting the quantization level to an even number, ensuring that the quantized state cannot be zero. In their work, both sensor and control signals are set to zero if a DoS attack occurs, which serves a similar function as the ACK technique. Even in the absence of a DoS attack, the quantized signal does not vanish  when the state is zero.
\textcolor{blue}{When ACK signal is lacking, the status of DoS attack becomes unknown. Consequently, transmission failures become indistinguishable, leaving the encoder without the information needed to update the quantizer parameters and avert saturation.
	During such undetected attacks, any subsequent mode switch risks desynchronizing the decoder and controller from the true subsystem mode. Without a reliable indicator of a DoS attack cannot be detected, the duration of the asynchronous interval remains unknown. }
In this paper, we propose a novel quantizer with an odd quantization level for switched systems operating under DoS attacks and without ACK technique, which is inspired by the idea of time-triggered and event-triggered mechanisms.

As pointed out in \Cref{coro1}, the switching signal can be designed to reduce the asynchronous behavior. Then a more practical assumption on switching signal is proposed. 
\begin{assumption} [Switching instants constraint]
	The switch only occurs at the sampling instants. 
\end{assumption}

Note that the DoS attack constraints based on frequency and duration are not suitable for the setting since the modeling method of DoS attack in \textbf{\textit{Strategies 1--3}} has the average meaning.
It is fickle since the minimum attack sleeping duration could be less than the average value.
Once the status of DoS attack is unknown, the state may be outside the quantization range. 
A somewhat stronger assumption is proposed for DoS attack, which is also common \cite{DoS_state_switching_2024}. 
\begin{assumption}[Intermittent DoS attack\cite{DoS_state_switching_2024}]\label{ass_attacknum}
	The number of DoS attack sleeping periods is not less than $n_{\min}$ and the number of DoS attack periods is less than $n_{\max}$.
\end{assumption}
Set $n_0=\kappa = 0$ in \eqref{equ_dosF} and \eqref{equ_dosD}, we have $n(t_k,t_0)\leq \frac{t_k-t_0}{n_{\min}+n_{\max}}$ and $|\Xi(t_k,t_0) | \leq \frac{t_k-t_0}{n_{\min}+n_{\max}}n_{\max}$.
In what follows, two approaches are proposed to update system state, which are inspired by time-triggering and  event-triggering  approaches, respectively. 

\subsubsection{Time-triggering (TT) approach }\label{sectime}
Similar to the classification in Table \ref{tab1}, the update law of the quantization parameter without an ACK signal is divided into four scenarios. Before introducing these scenarios, we define ``virtue DoS attack" and ``virtue switch." A virtue DoS attack is periodic, characterized by intervals $n_{\min} + n_{\max}$. The intervals $[\mathcal{T}n, \mathcal{T}n + n_{\min} \tau_s)$ and $[\mathcal{T}n + n_{\min} \tau_s, (n + 1) \mathcal{T})$ (where $\mathcal{T} = (n_{\min} + n_{\max}) \tau_s$) denote the sleeping interval and the attacking interval of the virtue DoS attack, respectively. Similarly, we define the virtue switching signal as periodic with a period $\tau_d$.
\textbf{Scenario 1} refers to the first instant in virtue attack-free cases. \textbf{Scenario 2} pertains to the other instants in virtue attack-free cases. \textbf{Scenarios 3 and 4} represent the first and subsequent instants during the virtue attack interval.

\begin{theorem}\label{thm_time}
	For the switched system \eqref{equ_state}  with \textbf{\textit{Strategy 4}} and the quantizer updates parameters according to
	
	{\scriptsize \vspace{-0.5em}	\begin{equation}\label{equ_Eke_ack}	
			\begin{aligned}
				&	 E_{k+1} = 
				\begin{cases}
					\varphi_1 E_k &\textbf{Scenario ~1}\\
					\varphi_2  E_k  &\textbf{Scenario ~2}\\
					\varphi_3 E_k & \textbf{Scenario ~3}\\
					\varphi_4 E_k & \textbf{Scenario ~4}\\
				\end{cases}
			\end{aligned}
	\end{equation}}
	and $E_k^+ = \varphi_5 E_k^-$ at the virtue switching instants, 
	where  ${\varphi}_1=\max\limits_ {p \in \mathcal{M}}  \tilde \rho_p \hat \lambda_p + \frac{\hat \rho_p\| 
		\hat {B}_p^dK_p\| }{N} $, $\varphi_2 = \max\limits_ {p \in \mathcal{M}}   \hat \lambda_p + \frac{\hat \rho_p\| 
		\hat {B}_p^dK_p\| }{N}  $, $\varphi_3 = \max\limits_{p\in\mathcal{M}}\tilde{\xi}_{p} \hat{\eta}_p $, $\varphi_4 = \max\limits_{p\in\mathcal{M}}\hat{\eta}_p $, $\varphi_5 = \max\limits_{p\in\mathcal{M}}\{ \hat{\xi}_{p}, \big(\tilde \rho_p \hat \lambda_p + \frac{\hat \rho_p\| 
		\hat {B}_p^dK_p\| }{N} \big)/ \big(\hat \lambda_p + \frac{\hat \rho_p\| 
		\hat {B}_p^dK_p\| }{N}\big),1 \}$, $\tilde{\rho}_p= \max\limits_ {p \in \mathcal{M}} \{\hat{\rho}_p,1\}$, and $\tilde{\eta}_p= \max_ {p \in \mathcal{M}} \{\hat{\eta}_p,1\}$.
	And the DoS attack constraints satisfy
	\begin{equation}\label{equ_DoS}
		\log (\varphi_1 \varphi_3)+ (n_{\min}-1)\log \varphi_2 + (n_{\max}-1) \log \varphi_4 + \frac{1}{\tau_d} \log \varphi_5<0
	\end{equation}
	then 	$ 	\|x(t_k)\| \leq E_k,~ k\geq 0	$ and the switched system is stable.
\end{theorem}
\begin{proof}
	
	Define the quantization parameter for systems with ACK technique as $\hat{E}_k$.
	In terms of the update law in  \Cref{lemma5}, we have 
	
	{\scriptsize \vspace{-0.5em}	 \begin{equation}\label{equ_Eke_ACK2}
			\hat{E}_{k+1} = 	
			\begin{cases}
				\varphi_1 \hat E_k &\text{\cref{case_1}-\textbf{a}}\\
				\varphi_2  \hat E_k  & \text{\cref{case_1}-\textbf{b}}\\
				\varphi_3 \hat E_k & \text{\cref{case_2}-\textbf{a}}\\
				\varphi_4  \hat E_k  & \text{\cref{case_2}-\textbf{b}}\\
			\end{cases}
	\end{equation}}
	and $\hat{E}_k^+ = \varphi_5 \hat{E}_k^-$ at switching instants.
	From the definition of $\varphi_5$, the jump at switching instants would be bounded. With the help of  \Cref{lemma5}, we have $\| x(t_k) \| \leq \hat{E}_k$. 
	From the definition of virtue DoS attack and Assumption \ref{ass_attacknum}, the frequency of virtue DoS attack and switch is higher than the actual ones. Moreover, the duration of virtue DoS attack is longer than the actual one. Based on $\varphi_1\geq \varphi_2$ and $\varphi_5 \geq 1$, we have $E_k \geq \hat E_k$.
	According to (\ref{equ_DoS}), one has 
	$\lim_{k\rightarrow \infty} E_k \rightarrow 0$.
	Then using the fact $\|  x(t_k) \| \leq E_k$, we know $\lim_{k\rightarrow \infty} \| x(t_k)\| \rightarrow 0$.
	This completes the proof.
\end{proof}

\subsubsection{Event-triggered (ET) approach}
In the subsection \ref{sectime}, a conservative result is derived to ensure that the state remains within the quantization range. The worst-case scenario is designed for periodic DoS attack and switching. The \textit{zooming-out} process is not dependent on the real-time state. 
In what follows, we give a result inspired by event-triggered mechanism: once the state exceeds the quantization range, the \textit{zooming-out} will be triggered. To simplify the analysis, we consider the situation that the switching signal is available for the decoder. The setting where the switching signal can be available instantly by the decoder is adopted. Such a setting is reasonable. 
In some engineering applications, the switching signal is transmitted via a local network that is secure from attacks, or it is pre-designed for the controller.

We define the number of periods from the latest switching (resp. \textit{zooming-out}) instant to the current instant as \( n_s \) (resp. \( n_z \)). Notably, the \textit{zooming-out} process may occur during DoS attack intervals, and the decoder may not be aware of it. A variable \( Flag \) is provided to record the status of \textit{zooming-out}: if \textit{zooming-out} occurs, then \( Flag = 1 \). Furthermore, the value of \( n_z \) is transmitted simultaneously to ensure the uniformity of the quantizer.

\begin{theorem} \label{thm1_ET}
	For the switched system \eqref{equ_state}  with \textbf{\textit{Strategy 4}} and the quantizer updates according to
	{\scriptsize 	\begin{equation}\label{equ_Eke_ack_et}	
			\begin{aligned}
				&	 E_{k+1}^e = 
				\begin{cases}
					\varphi^1 E_k^e   &  k=0 \text{ and switching instant}\\
					\varphi^1 E_k ^e+ \phi {E}_{ k-n_{\max}} ^e &\text{if } \| x(t_{k+1})\| >  \varphi_2  E_k \\
					\varphi^2 E_k^e  &  \text{otherwise} 
				\end{cases}
			\end{aligned}
		\end{equation}
		\begin{equation}\label{equ_Ekd_ack_et}	
			\begin{aligned}
				&	 E_{k+1}^d = 
				\begin{cases}
					\varphi^1 E_k^d   &  k=0 \text{ and switching instant}\\
					\varPhi(n_z,k) &Flag =1 \text{ and } n_s \leq n_z \\
					\varphi^1 E_k ^d+ \phi {E}_{ k-n_{\max}} ^d &Flag =1 \text{ and } n_s > n_z \\
					\varphi^2 E_k^d  &  \text{otherwise} 
				\end{cases}
			\end{aligned}
	\end{equation}}
	where  ${\varphi}^1= \max\limits_{p\in \mathcal{M}}\tilde \rho_p\hat  \lambda_p + \frac{\hat \rho_p\| 	\hat {B}_p^dK_p\| }{N} $, 
	$\varphi^2 =\max\limits_{p\in \mathcal{M}}  \hat \lambda_p + \frac{\hat \rho_p\| 
		\hat {B}_p^dK_p\| }{N}  $, 	$\varPhi(n_z,k) = (\varphi^2 )^{n_z-1}(	\varphi^1)^2 E_{k-n_z} ^d+\varphi^1(\varphi^2 )^{n_z-1}\phi {E}_{ k-n_z-n_{\max}} ^d$,  $\phi = 	 \max \limits_{\substack{p,q\in \mathcal{M}, \ell\in [0,n_{\max}], \overline{k}\in [0,\ell-1] }} \big( \| \varpi^x_{pq}(\ell,\overline{k})\| + \frac{ \overline{ \varpi}^e_{pq}(i,\overline{k})} {N}  \big)$ with $	\varpi^x_{pq}(\ell,\overline{k}) = (A_p^{cl})^{\ell-\overline{k}}( A_q^{cl} )^{\overline{k}}-  (\hat A_p^{d})^{\ell-\overline{k}}( \hat A_q^{d} )^{\overline{k}}$ and $ \overline{ \varpi}^e_{pq}(i,\overline{k}) = \| \hat B_p^d K_p\|  + \sum\limits_{i=0}^{\overline{k}}\|(A_p^{cl})^{\ell-\overline{k}}( A_q^{cl} )^{\overline{k}-i}\hat B_q^dK_q\| 
	+  \sum_{i=\overline{k}+1}^{\ell-1}\|(A_p^{cl})^{\ell-i}B_p^d K_p\|$.
	If there exist a odd quantization level $N \geq \max\limits_{p \in \mathcal{M}} \big \{ \frac{\hat \rho_p \| \hat B_p^dK_p\|}{1-\hat \lambda_p}\big\}$, constants $n_{\min}$,  $n_{\max}$ and $\tau_d$ such that  $n_{\max} < n_{\min}$ and 
	\begin{equation}\label{equ_DoS2}
		\begin{aligned}
			~&(\varphi^1)^{\frac{\mathcal{T}}{\tau_d}}\left ( (\varphi^2)^{\mathcal{T}(\frac{1}{\tau_s}-\frac{1}{\tau_d})-1}\varphi^1+ (\varphi^2)^{n_{\min}-1-\frac{\mathcal{T}}{\tau_d}}\phi\right) \leq 1 \\
		\end{aligned}
	\end{equation} with $\mathcal{T}= (n_{\min}+n_{\max})\tau_s$, 
	then 	$ 	\|x(t_k)\| \leq E_k^e,~ k\geq 0	$ and the system is stable.
\end{theorem}
\begin{proof}
	The proof is divided into two parts. Part I focuses on that  the \textit{zooming-out} of quantization parameters can capture the state within one step, that is, the state is always in the quantization range.  For a successive DoS attack interval and no-attack interval, zooming out occurs at most once.
	In Part II, DoS attack and switching constraint are derived in the worst case. 
	
	Part I: Suppose there is no DoS attack at the  sampling instant $t_k$  and the $n$-th DoS attack occurs at the following $\chi$ instants, i.e., $t_k\leq h_n < t_{k+1} < h_n+\tau_n \leq t_{k+\chi+1}$.
	The system dynamics (\ref{equ_close_1_zero}) for attack-free case can be rewritten as
	$	\tilde{x}(t_{k+\ell}) = \prod_{i=0}^{\ell-1} A^{cl}(i) \tilde{x}(t_k) + \sum_{i=0}^{\ell-2}  \prod_{j=i+1}^{\ell-1} A(j) B(i) K(i) \tilde{e}(t_{k+i}) + B(\ell-1)K(\ell-1) \tilde{e}(t_{k+\ell-1}) , \forall \ell \leq \chi \leq n_{\max}$
	where $A^{cl}(j) = A^{cl}_{\sigma(t_{k+j})} $, $B(j) = \hat B_{\sigma(t_{k+j})} ^d$ and $K(j) = K_{\sigma(t_{k+j})} $.
	$\tilde{x}(t)$ is the auxiliary state for DoS attack free case with $\tilde{x}(t_k) = x(t_k)$ at DoS attack-free instant $t_k$. $\tilde{e}(t)$ is the corresponding auxiliary error. 
	For DoS attack case, 
	$	x(t_{k+\ell}) =\prod_{i=0}^{\ell-1} A^{d}(i)x(t_{k}) $
	where $A^{d}(j) =  \hat A^{d}_{\sigma(t_{k+j})} $.

	According to the dynamics for system with/without DoS attack, the maximum error caused by consecutive $\ell$  DoS attack is
	$	\mathcal{E} (\ell,t_k)
	= \prod_{i=0}^{\ell-1} A^{cl}(i) {x}(t_k)+ \sum_{i=0}^{\ell-2}  \prod_{j=i+1}^{\ell-1} A(j) B(i) K(i) \tilde e(t_{k+i})  + B(\ell-1)K(\ell-1) \tilde e(t_{k+\ell-1}) -\prod_{i=0}^{\ell-1} A^{d}(i)x(t_{k})
	=:\varpi^x(\ell){x}(t_k)+ \sum_{i=0}^{\ell-1} \varpi^e(i) \tilde e(t_{k+i})$.
	Then, the norm of error admits
	$	\| 	\mathcal{E} (\ell,t_k)\|  
	\leq  \| \varpi^x(\ell)\| \| x(t_k)\|  +\sum_{i=0}^{\ell-1}   \left\| \varpi^e(i) \right\| \| \tilde e(t_{k+i-1-j})\| 
	\leq   \| \varpi^x(\ell)\| E_k^e +\sum_{i=0}^{\ell-1}   \left\| \varpi^e(i) \right\|  \hat E^{\max}_{k,i} /N$
	where $\hat E^{\max}_{k,i} = \max_{j \in [k, k+i-1]} \{ \hat E_{j}\}$ and $\hat E_{j}$ is the error bound of $\tilde e(t_j)$ for DoS attack free case. Moreover, we have $\hat E_k \geq \hat E_i,~\forall i\in [k,k+i-1] $ since $\hat E_k$ is decreasing during this interval. Therefore, one has $\hat E^{\max}_{k,i} =   \hat E_k = E_k^e$ as there is without attack at $t_k $. 
	Then the norm of error becomes
	$	\| 	\mathcal{E} (\ell,t_k)\|   \leq \| \varpi^x(\ell)\| E_k^e+\sum_{i=0}^{\ell-1}   \left\| \varpi^e(i) \right\| E_k^e /N $.

	The inequality $\tau_d > n_{\min}$ is fulfilled since the dwell time is large enough to compensate for the effect of switch and DoS attack.  Hence, during a DoS attack interval, there are at most two modes. Then $\varpi^x(\ell) $ and $\varpi^e(\ell) $ satisfy
	$\|	\varpi^x_{pq}(\ell,\overline{k}) \|= \| (A_p^{cl})^{\ell-\overline{k}}( A_q^{cl} )^{\overline{k}}-  (\hat A_p^{d})^{\ell-\overline{k}}( \hat A_q^{d} )^{\overline{k}}\|$
	and
	$\sum_{i=0}^{\ell-1}   \left\| \varpi^e_{pq}(i,\overline{k}) \right\| \leq  \| \hat B_p^d K_p\|  + \sum\limits_{i=0}^{\overline{k}}\|(A_p^{cl})^{\ell-\overline{k}}( A_q^{cl} )^{\overline{k}-i}\hat B_q^dK_q\| 
	+  \sum\limits_{i=\overline{k}+1}^{\ell-1}\|(A_p^{cl})^{\ell-i}B_p^d K_p\|$.
	Define
	$\phi = 	 \max \limits_{\substack{p,q\in \mathcal{M}, \ell\in [0,n_{\max}], \overline{k}\in [0,\ell-1] }} \big( \| \varpi^x_{pq}(\ell,\overline{k})\| +\sum\limits_{i=0}^{\ell-1}  \frac{ \left\| \varpi^e_{pq}(i,\overline{k}) \right\|}{N}  \big).$
	The \textit{zooming-out process} occurs during DoS attack or after DoS attack due to the effect of attack. Here, we discuss three situations: 1) the \textit{zooming-out} instant  belongs to DoS attack interval; 2) the \textit{zooming-out} happens at the ending instant of DoS attack; 3) the \textit{zooming-out} occurs over the attack-free interval. 
	Suppose that the \textit{zooming-out} is triggered at instant $t_{k+1}$. Then we have 
	$E_{k+1} ^e= \varphi^1 E_{k}^e + \phi E_{k-n_{\max}}^e$.
	From 
	$	\| x(t_{k+1})\|  =  \| \mathcal{E}({i},t_{k-i})-A^{cl}_p  x(t_{k}) -\hat  B_p^d K_p \tilde e(t_k)  \|
	\leq  \phi E_{k-n_{\max}} ^e+ \varphi^1 E_{k}^e,$
	one has $\|x(t_{k+1})\| \leq E_{k+1}^e$.
	For situation 1), if the DoS attack occurs at the next instant and the switch does not occur during this interval, we have 
	$	\| x(t_{k+2})\|  
	= \| \mathcal{E}({i+1},t_{k-i+1})-(A^{cl}_p )^2 x(t_{k})-\sum_{j=0}^{1} (\hat{A}_{p}^d) ^j\hat{B}_p^d K_p \tilde{e}(t_{k+1-j}) \|  
	\leq  \phi E_{k-n_{\max}+1}^e +  \varphi^1 \varphi^2 E_{k}^e
	= \varphi^2  ( \phi E_{k-n_{\max}} ^e+ \varphi^1 E_{k}^e) = \varphi^2  E_{k+1}^e$.
	Repeating this process, it is obvious that for a DoS attack interval, the\textit{ zooming-out} is  triggered at most one time. In addition, the jump of switch is involved when computing $\phi$. If the switch occurs during DoS attack intervals, the process of analysis is similar.
	
	For situation 2), \textit{zooming-out} occurs at the end instant of DoS attack, i.e., $ACK_{t_{k+2}}=1$. The increase of $\varphi^1$ at the first instant (i.e., \textbf{\cref{case_1}-a}) is considered in\textit{ zooming-out} process. Then the \textit{zooming-out} process will not happen until the next DoS attack interval. 
	
	Situation 3) may occur since \textit{zooming-out} of \textbf{\cref{case_1}-a} does not happen. In such a situation, we have $\| x(t_k)\| \leq E_k^e$ and $\|x(t_{k+1})\| > \varphi_2 E_k^e$ when the \textit{zooming-out} is triggered at instant $t_{k+1}$. Then the \textit{zooming-out} can capture the state. Moreover, it will not be triggered until the next DoS attack interval, which can be found from \Cref{lemma5}. 
	
	Part II: The worst case is that each DoS attack interval triggers a \textit{zooming-out}, the length of DoS attack interval becomes the maximum, and the length of DoS attack sleeping interval is the minimum. 
	Next, by analyzing the increase of $E_k$ when a switch occurs in different intervals one by one,  the case where  switch occurs during DoS attack intervals and after a zooming-out instant will cause the maximum increase, that is, 
	$	E_{\hat{k}\mathcal{N} }^e \leq  \varphi_1^{\frac{t_{{k}}-t_0}{\tau_d}}\Big ( (\varphi^2)^{\mathcal{N} (1-\frac{\tau_s}{\tau_d})-1}\varphi^1+ (\varphi^2)^{n_{\min}-1- \frac{\mathcal{N}\tau_s}{\tau_d}}\phi\Big)^{\hat{k}}E_0$
	where $\hat{k}$ is the number of period $\mathcal{N} = n_{\min}+n_{\max}$. Therefore, $t_k-t_0 =\hat{k}\mathcal{N} \tau_s$.
	From (\ref{equ_DoS2}), we have $\lim_{k\rightarrow \infty} E_k^e \rightarrow 0$. Due to $\|x(t_k)\| \leq E_k^e$, the system is stable. 
	
	The uniformity of encoder and decoder can be derived directly. The special case is that the switch occurs during DoS attack intervals and just after \textit{zooming-out} instant. In such a case, the decoder parameter should be obtained based on the historical data, but it is also obvious. 
	The proof is completed.
\end{proof}
\color{black}
\begin{remark}
	Since the system operates without the ACK signal, the encoder cannot ascertain the status of DoS attack. However, the decoder could detect the occurrence of a DoS attack. Once there is a DoS attack, the decoder records the quantization parameters and the switching instants. The encoder transmits the \textit{zooming-out} instant \( n_z \) and a flag to the decoder. Upon receiving this information, the decoder calculates the corresponding value at the current instant.
\end{remark}
\color{blue}
\begin{remark}
	The stability analysis for ACK-free case is nontrivial. Because the DoS attacks interval is unknown to the quantizer, the system alternates unpredictably between attack and attack-free regimes, and characterizing the resulting dynamics is a key challenge for achieving stabilization. Specifically, time-triggered designs simplify proofs but sacrifice performance.  The event-triggered policies cannot pre-determine jumps in quantization parameters before operation. Investigating the zooming-out behavior during DoS attacks is therefore essential.
	Furthermore, the switching signal must be constructed accordingly. 
\end{remark}
\begin{remark}
	The proposed quantization update laws in \eqref{equ_Eke_ack}, \eqref{equ_Eke_ack_et} and \eqref{equ_Ekd_ack_et}  reveal distinct design philosophies between event-triggered (ET) and time-triggered (TT) strategies under DoS attack-induced asynchronous behavior. The TT strategy prioritizes determinism using fixed update period, i.e., $\mathcal{T} = (n_{\min}+n_{\max})\tau_s$, which ensures stability under periodic DoS attacks as guaranteed by \Cref{thm_time}. This periodicity aligns with worst-case scenarios where quantization ranges are proactively expanded at each interval, simplifying encoder-decoder coordination through \Cref{thm1_ET}. However, this rigidity may lead to sub-optimal performance during non-periodic DoS attacks, as quantization ranges are expanded irrespective of actual state divergence.  
	In contrast, the ET strategy dynamically adjusts quantization parameters based on real-time state feedback, activating zoom-out mechanisms only when $\|x(t_k)\|> \varphi_2 E_k^e$. This adaptability minimizes redundancy during stable periods while accelerating convergence whenever attacks occur. ET strategy achieves a tighter error bound compared to TT one, especially when DoS is sporadic or low intensity. 
\end{remark}
\color{black}
	
	\section{simulations}\label{sec_sim}
\section{simulations}\label{sec_sim}
In this section, two examples are given. They will show the effectiveness of \textbf{\textit{Strategies 1--3}} and \textbf{\textit{Strategies 3--4}}, respectively. 
\subsection{Example A}
\color{blue}
Helicopters are indispensable across a wide range of domains, from military and civil aviation to precision agriculture. Among their key capabilities, vertical take-off and landing (VTOL) most directly governs overall dynamic performance; therefore, maintaining stable VTOL behavior under real-world operating conditions remains a critical research focus.
In \cite{}, the the VTOL helicopter model (VTOLHM) is described by $dx(t)=[A_{\sigma (t)}x(t)+B_{\sigma(t)}u(t)]dt$ , where the state vector $x(t)$ comprises horizontal velocity, vertical velocity, pitch rate, and pitch angle; while $\sigma(t)$ is the switching signal. 
From Table \ref{para}, the vertical take-off and landing helicopter operates at different speeds. Each speed induces unique aerodynamic characteristics that are represented by a dedicated subsystem with its own parameter set
\begin{align*}
	&A_{\sigma(t)} = \left[	\begin{array}{cccc}
		-0.0366& 0.0271&0.0188&-0.4555\\
		0.0482&-1.0100&0.0024&-4.0208\\
		0.1002&a_{32}(\sigma(t)) &-0.7070 &a_{34}(\sigma(t))\\
		0&0&1&0
	\end{array}\right]	\\
	&	B_{\sigma(t)} = \left[	\begin{array}{cc}
		0.4422&   0.1761\\
		b_{21}(\sigma(t))    &-7.5922\\
		-5.5200 & 4.4900\\
		0 & 0
	\end{array}\right].
\end{align*}
\begin{table}
	\caption{DoS attack constraints in each strategy}\label{para}
	\centering
	\renewcommand\arraystretch{1}
	\begin{tabular}{cccc}
		\toprule[1pt]
		Airspeeds(knots) &$a_{32}(\sigma(t)) $&$a_{34}(\sigma(t)) $ &$b_{21}(\sigma(t)) $ \\
		\midrule[0.7pt]
		135	&  0.3681 &1.4200	& 3.5446 \\
		170 & 0.5047 & 2.5460& 5.1120 \\
		\bottomrule[1pt]
	\end{tabular}
\end{table}
The controller gains are
\begin{align*}
	K_1= &\left[	\begin{array}{cccc}
		0.0077& 0.0742&0.2070&0.2516\\
		-0.0123&0.0620&0.1266&0.1317
	\end{array}\right]	,~\\
	K_2=& \left[	\begin{array}{cccc}
		-0.0030&-0.0055&0.2631&0.4130\\
		-0.0471&0.1000&0.2202&-0.3283
	\end{array}\right].
\end{align*}
Other parameters are set as follows: $\lambda_{1} = 0.9602$, $\lambda_{2} = 0.9609$, $\eta_1 = 1.0216$, $\eta_2 = 1.0922$, $\rho_{1} = 3.8927$, $\rho_{2} = 3.8620$, $\xi_{12} = 1.1687$, $\xi_{21} = 0.5751$, 
$\hat \lambda_{1} = 0.9822$, $\hat\lambda_{2} = 0.9607$, $\hat\eta_1 = 1.1004$, $\hat\eta_2 = 1.1029$, $\hat\rho_{1} = 3.1096$, $\hat\rho_{2} = 3.8799$, $\hat\xi_{1} = 1.0930$, $\hat\xi_{2} = 1.0599$ with sampling period $\tau_s = 0.05$ and the maximum asynchronous interval $N_{\max} = 2$.

\textbf{\textit{Strategy 1:}}
Using  \Cref{thm1}, the quantization level is required to meet  $N > 1.1541$. Based on the system definition and the corresponding matrices, we can get $b = 0.2951$, $\nu_1 = 0.9602$, $\nu_2 = 0.9609$, $\hat{\nu}_1 = 3.7377$, $\hat{\nu}_2 = 2.8416$, $\mu_{12}^1 = 2.9589$, $\mu_{21}^1 = 2.7118$, $\mu_{12}^2 = 1.1940$, $\mu_{21}^2 = 0.6281$, $\mu_{12}^3 = 1.0216$, and $\mu_{21}^3 = 1.0922$. Then the dwell time constraint satisfies $\tau_d \geq 2.5948$. In this simulation, we choose $N = 3$, $\tau_d = 5.9$, $\tau_D = 1.2$, and $T = 1.5$.  It is verified that conditions \eqref{equ_con1} and \eqref{equ_con2} are fulfilled.

Fig. \Cref{fig:switching} depicts a feasible switching signal, as well as  the system mode, controller mode, and the DoS attack intervals. We can see that the asynchronous behavior occurs due to the DoS attack. The corresponding state trajectories of switched systems under DoS attack and quantization with initial condition $x_0 = [-0.5,-1,1,-0.5]^T$, are presented in Fig. \ref{fig:state}. It is clear that the system is stable. 
Furthermore, Fig. \ref{fig:ek} plots the evolution of quantizer parameters $E_k^e$ and $E_k^d$ for the encoder and decoder, respectively. The curves over the intervals $[12,13]$ and $[18,18.5]$ are enlarged to accentuate the discrepancies between the encoder and decoder parameters during the DoS attack. It can be observed the trajectories of $E_k^e$ and $E_k^d$ are perfectly identical during attack-free intervals, which further corroborates the validity of  \Cref{lemma2}.

Finally, we verify the result in \Cref{coro1}. From \eqref{equ_con2_1}, the dwell time constraint admits $\tau_d \geq 2.0429$. With $N = 3$, \eqref{equ_con1_1} simplifies to $-0.8910 + \left(\frac{1}{T} + \frac{\tau_s}{\tau_D} - 1\right) 1.0986 \leq 0$. A feasible solution is  $T = 1.1$ and $\tau_D \geq 0.0420$. Compared to \Cref{thm1}, \Cref{coro1} markedly relaxes the DoS attack frequency and switching signal restrictions, yielding a substantially less conservative design.

\begin{figure*}
	\begin{subfigure}[b]{0.32\linewidth}
		\centering
		\includegraphics[width=\linewidth]{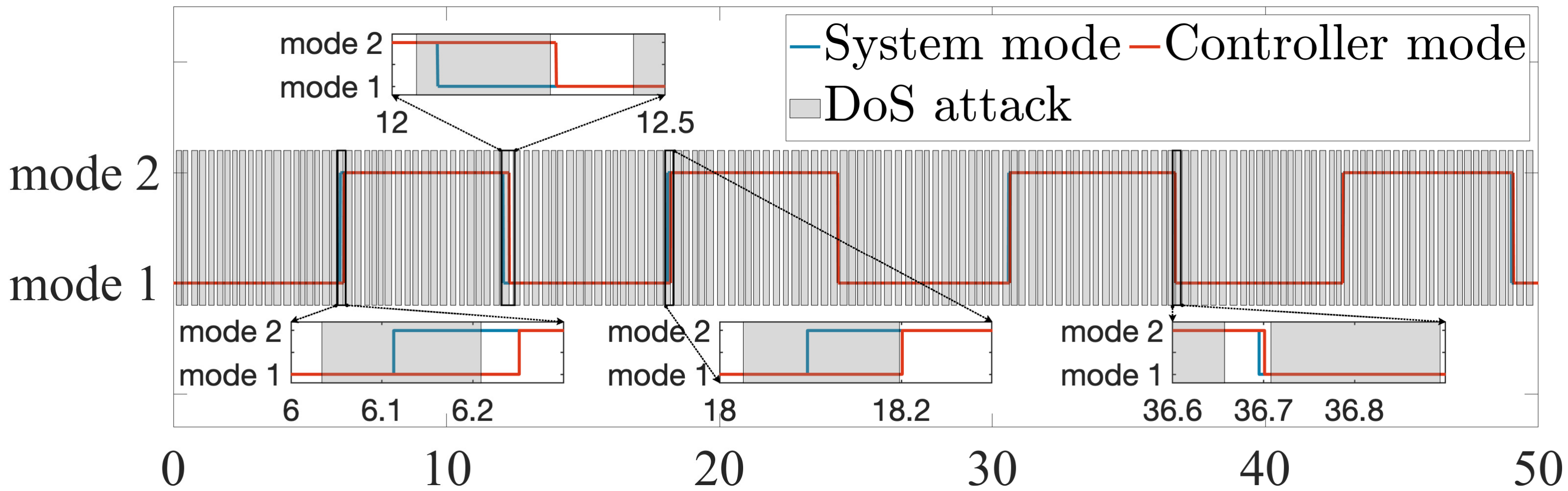}
		\caption{The switching signal and DoS attack} \label{fig:switching}
	\end{subfigure}
	\begin{subfigure}[b]{0.32\linewidth}
		\centering
		\includegraphics[width=\linewidth]{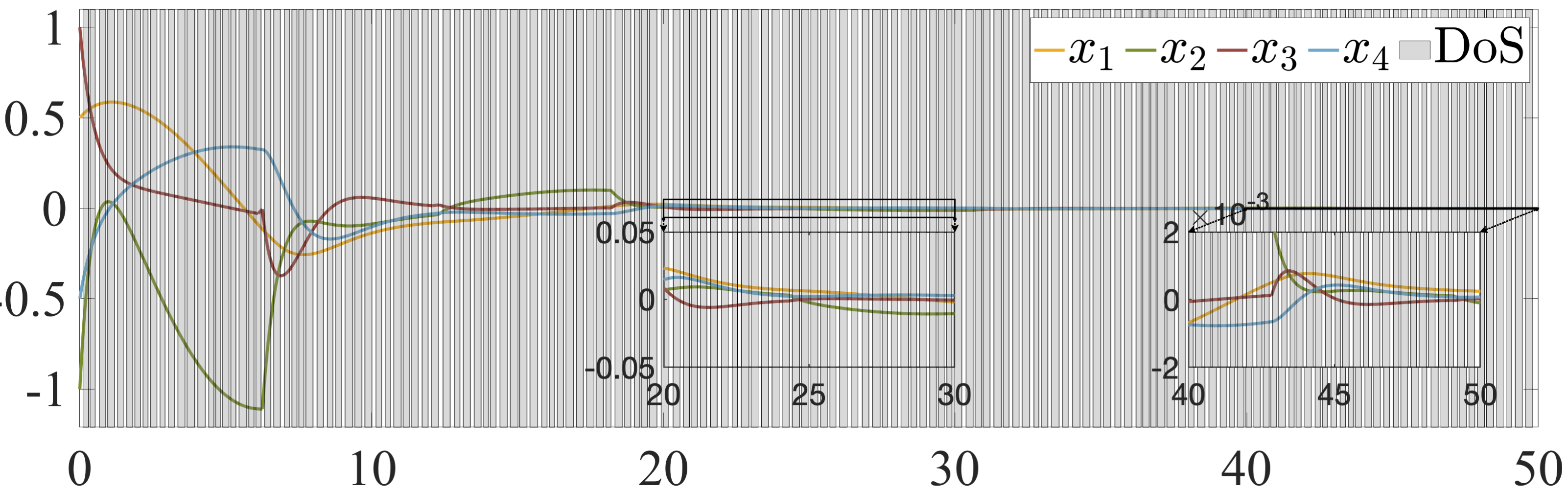}
		\caption{State trajectories} \label{fig:state}
	\end{subfigure}
	\begin{subfigure}[b]{0.32\linewidth}
		\centering
		\includegraphics[width=\linewidth]{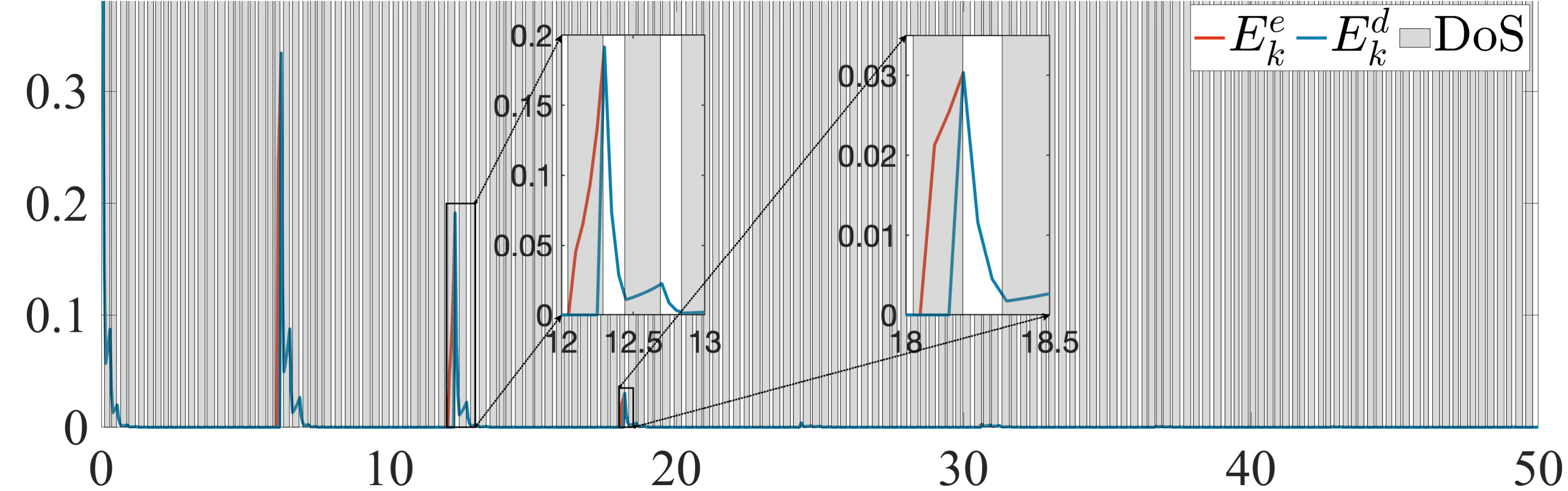}
		\caption{The values of $E_k^e$ and $E_k^d$} \label{fig:ek}
	\end{subfigure}
	\caption{\textbf{\textit{Strategy 1}} in Example A} 
\end{figure*}
\begin{figure*}
	\begin{subfigure}[b]{0.32\linewidth}
		\centering
		\includegraphics[width=\linewidth]{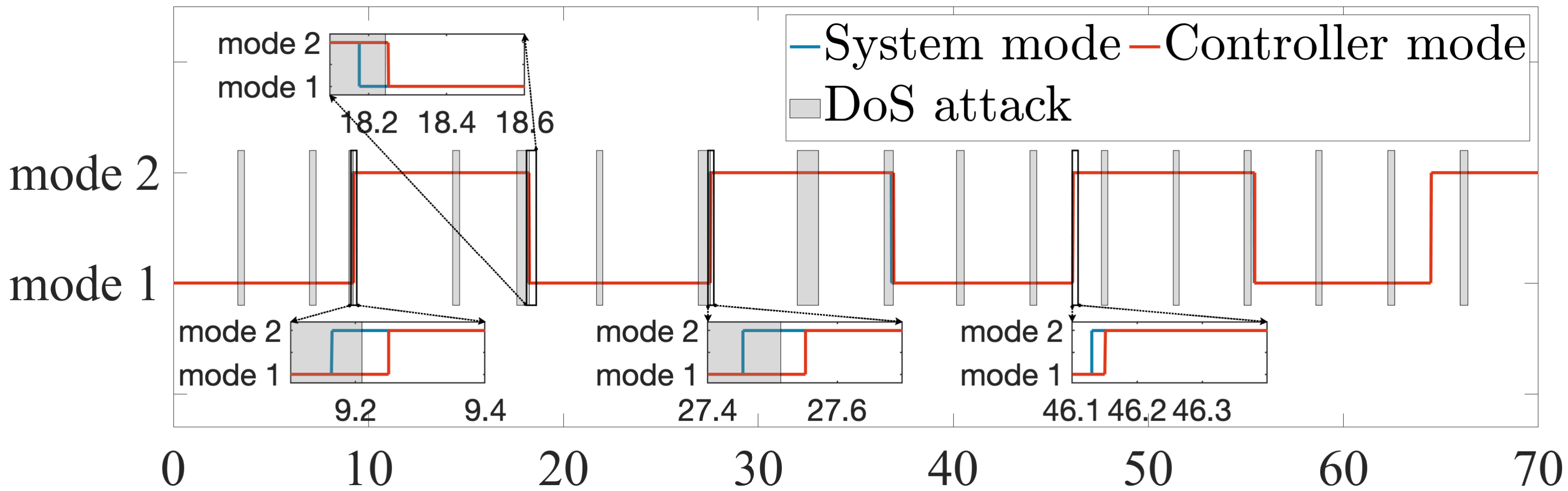}
		\caption{The switching signal and DoS attack} \label{fig:switching2}
	\end{subfigure}
	\begin{subfigure}[b]{0.32\linewidth}
		\centering
		\includegraphics[width=\linewidth]{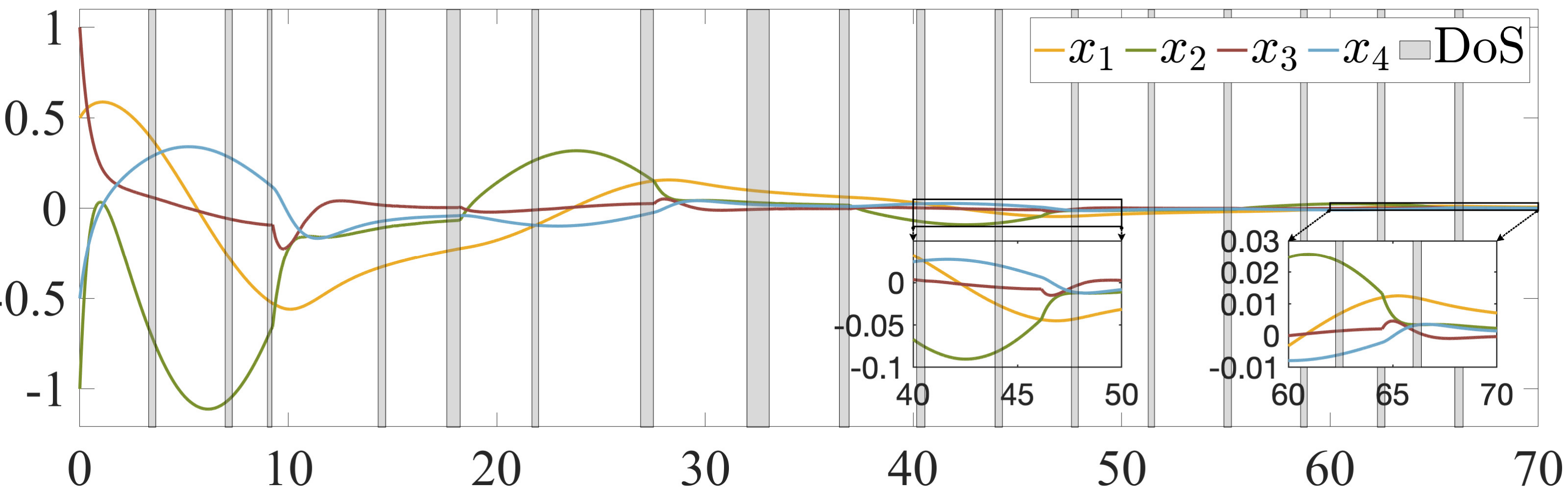}
		\caption{State trajectories} \label{fig:state2}
	\end{subfigure}
	\begin{subfigure}[b]{0.32\linewidth}
		\centering
		\includegraphics[width=\linewidth]{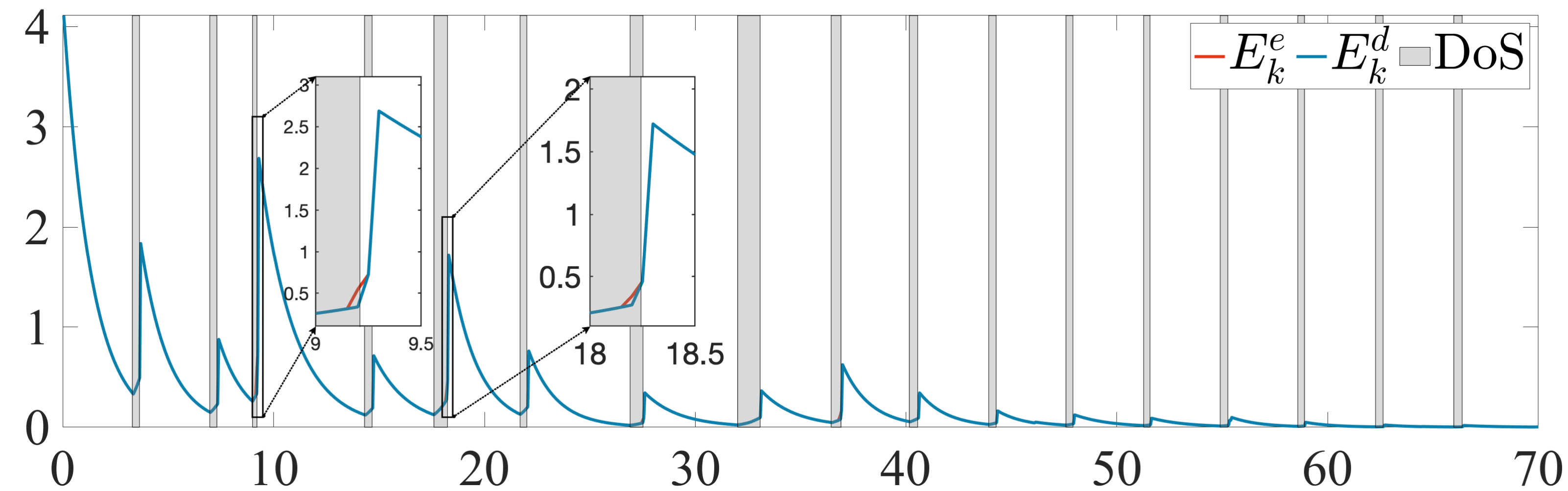}
		\caption{The values of $E_k^e$ and $E_k^d$} \label{fig:ek2}
	\end{subfigure}
	\caption{\textbf{\textit{Strategy 2}} in Example A} 	
\end{figure*}
\begin{figure*}
	\begin{subfigure}[b]{0.32\linewidth}
		\centering
		\includegraphics[width=\linewidth]{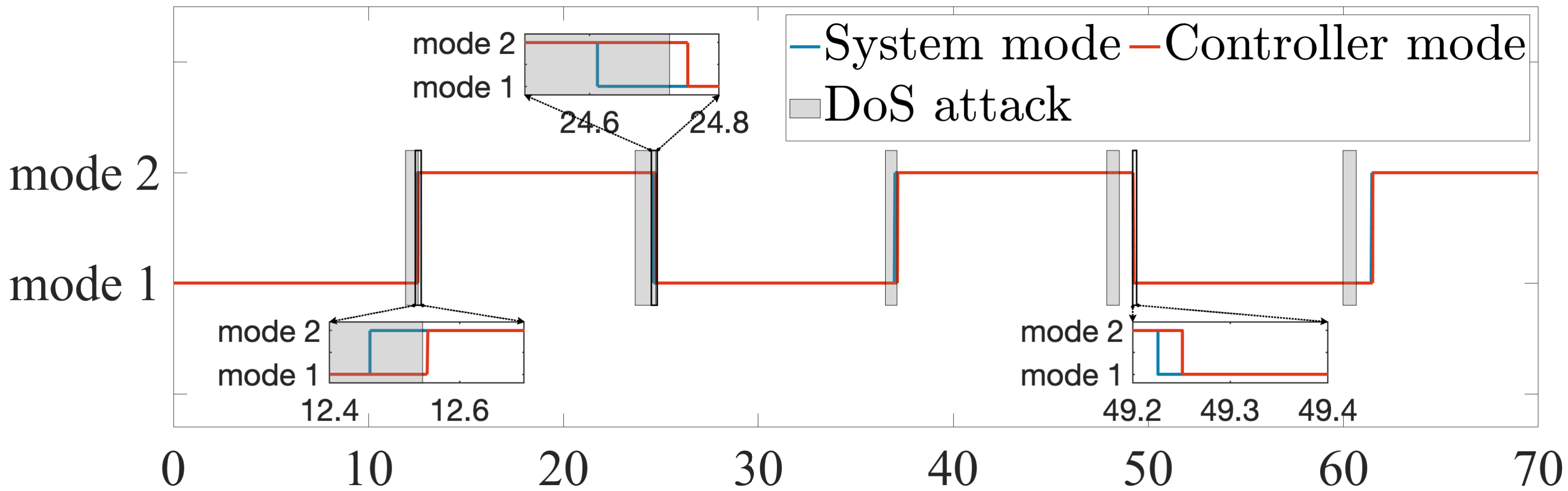}
		\caption{The switching signal and DoS attack} \label{fig:switching3}
	\end{subfigure}
	\begin{subfigure}[b]{0.32\linewidth}
		\centering
		\includegraphics[width=\linewidth]{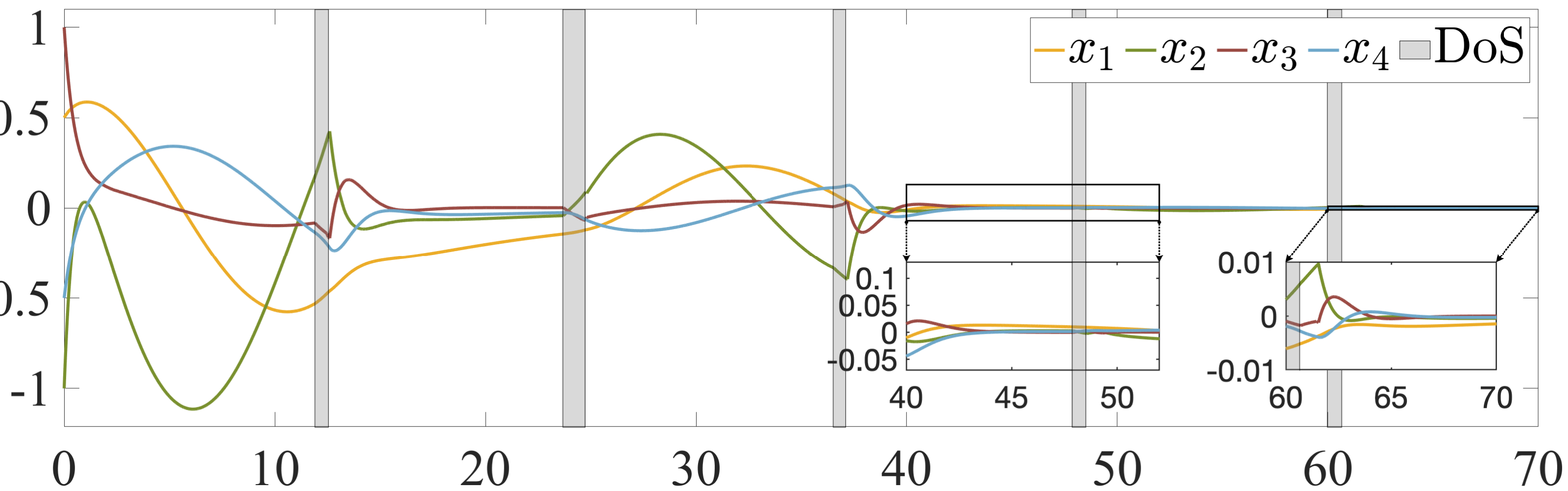}
		\caption{State trajectories} \label{fig:state3}
	\end{subfigure}
	\begin{subfigure}[b]{0.32\linewidth}
		\centering
		\includegraphics[width=\linewidth]{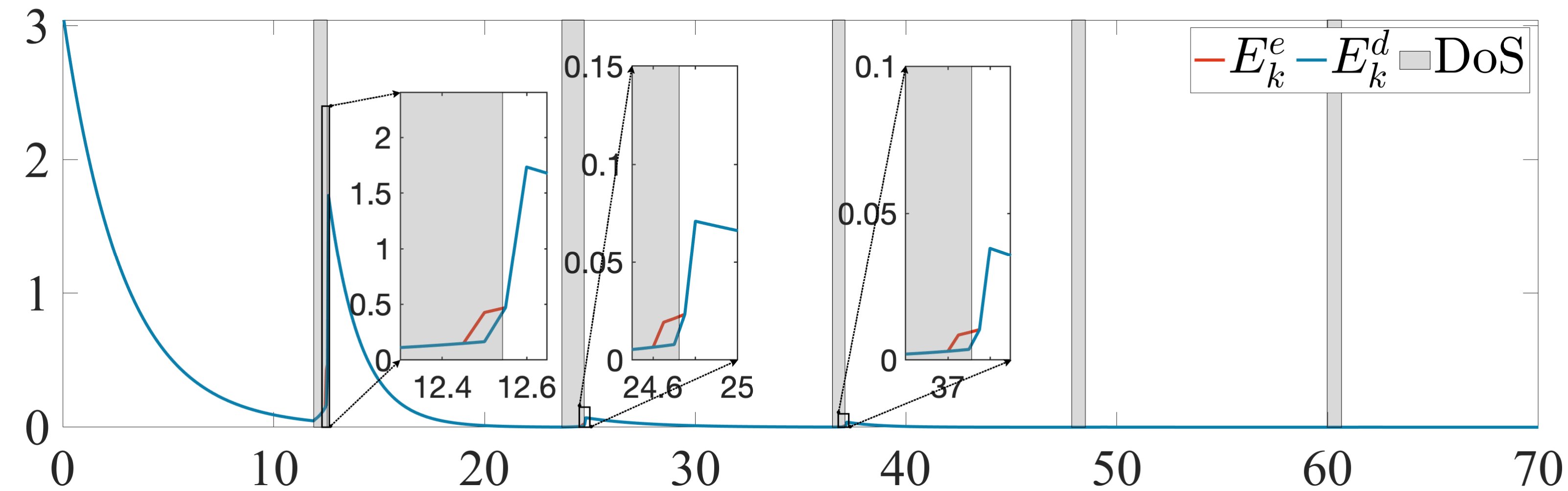}
		\caption{The values of $E_k^e$ and $E_k^d$} \label{fig:ek3}
	\end{subfigure}
	\caption{\textbf{\textit{Strategy 3}} in Example A} \label{fig:control_3}	
\end{figure*}

\textbf{\textit{Strategy 2:}}
From  \Cref{thm1_ori}, we obtain $N \geq 119.4795$. For faster convergence rate, we set  $N = 125$. Moreover, the dwell time constraint requirement derived from \eqref{equ_con2_origin} is $4.4913$. To relax DoS attack constraints, we choose $\tau_d = 9$, $\tau_D = 6$ and $T = 12$, which satisfies \eqref{equ_con1_origin}. Under this pattern, the maximum number of consecutive DoS attack periods is $n_{\max} \leq 18$, with a total duration less than 0.9s. According to the frequency and duration constraints, the average attack length is $\tau_D / T = 0.5$ (i.e., communication is interrupted for 10 sampling instants). Hence, the maximum number of DoS attack is not overly conservative.
As confirmed  in Fig. \ref{fig:state2}, the system remains stable and the quantization parameters stay unchanged during attack-free intervals.

\textbf{\textit{Strategy 3:}}
This case yields the quantization level satisfying $N \geq 29.2091$. We set $N = 175$. The parameters  $\tau_d = 12$, $\tau_D = 20$, $T = 20$ again satisfy the conditions of \Cref{thm1_zero}.
The switching signal, state trajectories, and quantization parameters shown in Fig. \ref{fig:control_3} verify that the system remains stable under \textbf{\textit{Strategy 3}}.

\begin{table}[h]
	\caption{\textcolor{blue}{Allowable DoS Attack and Switching Signal Parameters of Example A in the Revised Version}}
	\label{DoS}
	\centering
	\renewcommand\arraystretch{1}
	\color{blue}
	\begin{tabular}{cccc}
		\toprule[1pt]
		Parameters  & Strategy 1 & Strategy 2 & Strategy 3\\
		\midrule[0.7pt]
		$\tau_D$ &1.2  &6  & 20 \\
		$T$ &1.5  &  12& 20 \\
		$\tau_d$ &5.9 &9&12\\
		$N$ &3&125& 175\\
		\bottomrule[1pt]
	\end{tabular}
\end{table}

\color{black}

\textbf{\textit{Discussions}}: 
Example A illustrates the superiority of the active controller over \textbf{\textit{Strategy 3}}. \textcolor{blue}{More specific, the system with active control can tolerate more severe attacks, see \Cref{DoS} for details.}
The  key differences are summarized below.

1) Computation burden: In \textbf{\textit{Strategy 1}}, both the quantization parameters $E_k^e$ and $x_k^{e*}$ need to be calculated, whereas only $E_k^e$ is involved in \textbf{\textit{Strategies 2-3}}. This means that \textbf{\textit{Strategy 1}} incurs a noticeably higher computation burden.

2) Quantization level: The quantization level in \textbf{\textit{Strategy 1}} is significantly lower than other strategies. In \textbf{\textit{Strategy 1}}, the quantization level depends solely on $\|e^{A_p\tau_s}\|$, whereas in \textbf{\textit{Strategies 2-3}}, it is governed  by the convergence rate of $\|A_p^d\|$ or $\|\hat{A}_p^d\|$. For instance, the quantization level in \textbf{\textit{Strategy 2}} is $N \geq \max_{p \in \mathcal{M}} \left\{ \frac{\rho_p \|B_p^d\|}{1-\lambda_p}\right\}$. Since $\|e^{A_p\tau_s}\|$ and $\|B_p^d\|$ 
are of comparable magnitude while $\rho_p>1$ and $\lambda_p\approx 1$, the required quantization level in \textbf{\textit{Strategy 1}} is dramatically lower. Although there is no simple ordering between  \textbf{\textit{Strategies 2}} and \textbf{\textit{3}}, the active controller exhibits better defense capabilities against DoS attack compared to the passive controller. 
In other words, DoS attack causes a larger impact on systems using passive controllers, as seen from the update law of \cref{case_2} in each strategy. \textbf{\textit{Strategy 3}} deliberately uses a higher quantization level to accelerate the \textit{zooming-in} phase to compensate for the effects of DoS attack (see Figs. \ref{fig:ek2} and \ref{fig:ek3}).

3) Switching signal constraints: \textbf{\textit{Strategy 1}} imposes the least stringent dwell-time requirement. \Cref{coro1} further mitigates the influence of switching on quantization updates and attains the smallest dwell time constraint among all schemes.

4) DoS attack constraints:  \textcolor{blue}{\Cref{DoS} shows that \textbf{\textit{Strategy 1}} tolerates the highest frequency and the longest duration of DoS attack, as evidenced by the smallest parameters $\tau_D$ and $T$. 
}
The attack constraint in Corollary \ref{coro1} becomes more relaxed, with \textbf{\textit{Strategy 1}} being the least effective and \textbf{\textit{Strategy 3}} the most restrictive.

In summary, active control scheme significantly outperforms passive control ones in resisting DoS attacks, with  \textbf{\textit{Strategy 1}} exhibiting the largest resilience. This advantage stems from the active controller’s ability to generate an appropriate control action even during a DoS attack interval.

\begin{table*}[h]
	\caption{Comparison of Quantization Control Strategies}
	\label{tabqua2}
	\centering
	\renewcommand\arraystretch{1}
	\setlength\tabcolsep{2pt}
	\color{blue}
	\begin{threeparttable}
		\begin{tabular}{m{0.03\linewidth}<{\centering} m{0.31\linewidth}<{\centering} m{0.31\linewidth}<{\centering} m{0.31\linewidth}<{\centering}}
			\toprule[1pt]
			Strategy\tnote{*}	& Advantages & Disadvantages & Applicability \\
			\midrule[0.7pt]
			\textbf{\textit{S1}} & High  defense capability and low byte requirement (low quantization levels $N$) & High computational resource requirements for both quantizer and controller & Systems with high security requirements and limited communication resources \\
			\addlinespace
			\textbf{\textit{S2}} & Trade-off  quantizer computational resources and defense capability & High computational resource requirement for controller & Quantizer has limited computational resource, controller has sufficient resources \\
			\addlinespace
			\textbf{\textit{S3}} & Lower computational resource requirements & Poor defense capability & Low attack threats and limited computational resources \\
			\addlinespace
			\textbf{\textit{S4}} & No ACK signal required & Poor defense capability & Network where ACK is unavailable\\
			\bottomrule[1pt]
		\end{tabular}
		\begin{tablenotes}
			\footnotesize
			\item[*] \textbf{\textit{S1--S4}} correspond to \textit{\textbf{Strategies 1--4}}.  
		\end{tablenotes}
	\end{threeparttable}
\end{table*}

\subsection{Example B}

\begin{figure}[h]
	\centering
	\includegraphics[width=0.7\linewidth]{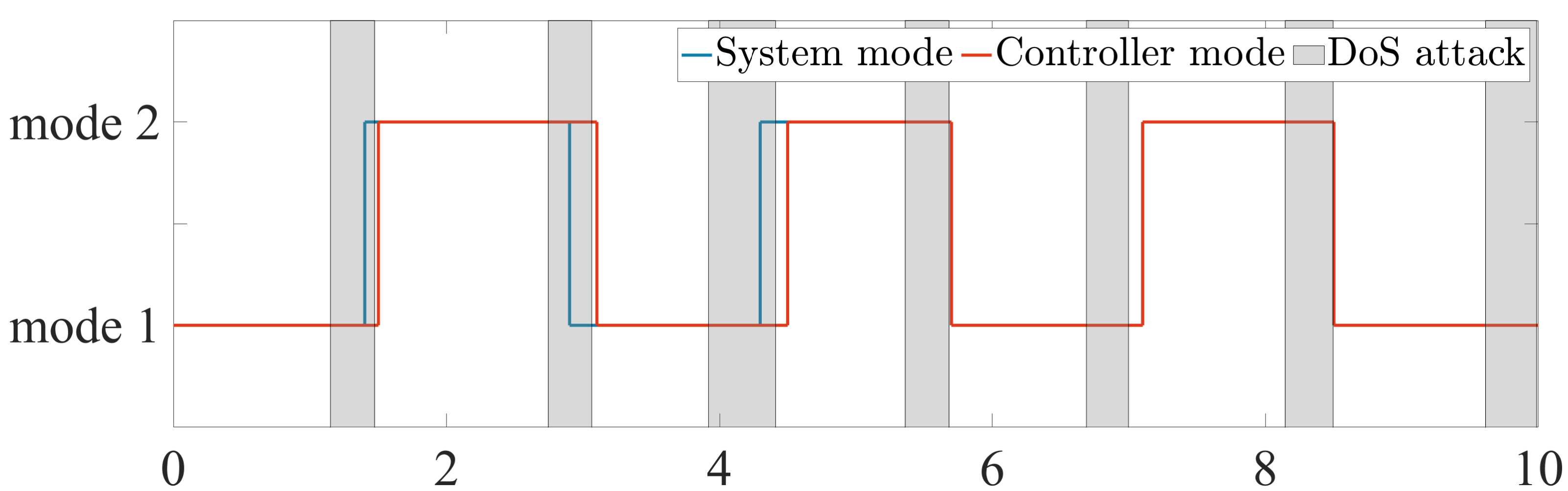}
	\caption{The switching signal and DoS attack in Example B}
	\label{fig:switching41}
\end{figure}

Consider a switched system consisting of two subsystems with parameters
\begin{align*}
	A_1 =& \left[	\begin{array}{cc}
		3.3968   & 1.4663\\
		3.0711 &  -3.2308
	\end{array}\right]	,~
	A_2= \left[	\begin{array}{cc}
		-0.2915  & -3.4221\\
		-0.2891  & -3.6517
	\end{array}\right],~\\
	B_1=& \left[	\begin{array}{cc}
		0.9207   &-1.2808\\
		1.7131  &  0.9749
	\end{array}\right],
	~	B_2= \left[	\begin{array}{cc}
		0.9207   & -1.2808\\
		1.7131  &  0.9749
	\end{array}\right].
\end{align*}
The controller gains are
	$	K_1= \left[	\begin{array}{cc}
		-3.6011  & -2.5837\\
		4.1549  & -1.1392
	\end{array}\right]	,~
	K_2= \left[	\begin{array}{cccc}
		-2.0251 &  -0.9775\\
		3.7468   &-2.4542
	\end{array}\right].$
Let $\hat{\lambda}_{1} = 0.375$, $\hat{\lambda}_{2} = 0.35$, $\hat{\eta}_{1} = 1.0314$, $\hat{\eta}_{2} = 1.4979$, $\hat{\rho}_{1} = 1.0986$, $\hat{\rho}_{2} = 1.1158$, $\hat{\xi}_{1} = 1.4421$, and $\hat{\xi}_{2} = 0.1$. The quantization level satisfies $N \geq \max_{p \in \mathcal{M}} \left\{ \frac{\hat{\rho}_p \|\hat{B}_p^dK_p\|}{1-\hat{\lambda}_p} \right\} = 2.1153$. Hence, we set $N = 3$. 	Choose the dwell time constraint be $2$ seconds, and the minimum DoS attack sleeping duration period is $n_{\min} = 10$ and the maximum consecutive DoS attack period is $n_{\max} = 4$. It can be verified that the parameters satisfy the conditions in  \Cref{thm1_zero,thm_time,thm1_ET}. 

The switching signal is plotted in Fig. \ref{fig:switching41}.  Figs. \ref{fig:control3B}-\ref{fig:control42B} show the state and quantization parameters for  \textbf{\textit{Strategy 3}} and two methods in  \textbf{\textit{Strategy 4}}. 
It is found that among these three control strategies, the one with ACK  signal has the fastest convergence speed. This is reasonable since  the quantization parameters given in  \textbf{\textit{Strategy 4}} are more conservative in order to compensate the absence of  ACK signal. 
In addition, compared with the updating methods of quantization parameters based on event-triggered and time-triggered mechanisms, we can see  that the state of event-triggered method has chatting in the early stage. But the event-triggered method converges faster than the time-triggered counterpart. Since the event-triggered method has close relationship with the real-time state of the system,  DoS attack against aperiodic distribution will perform better, especially in the case where the actual duration of attack-free is significantly higher than the constraint.
\begin{figure*}
	\setlength{\abovecaptionskip}{0pt} 
	\setlength{\belowcaptionskip}{0pt} 
	
	\begin{minipage}{0.32\linewidth}
		\begin{subfigure}[b]{\linewidth}
			\setlength{\abovecaptionskip}{0pt} 
			\setlength{\belowcaptionskip}{0pt} 
			\centering
			\includegraphics[width=\linewidth]{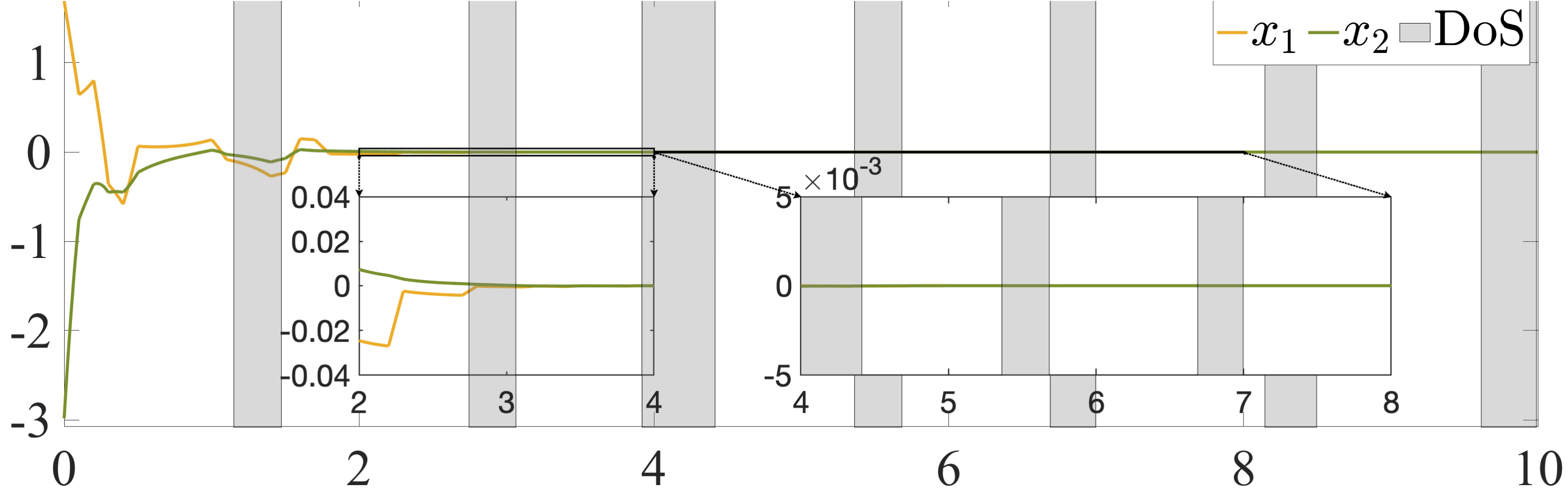}
			\caption{State trajectories} 
			\label{fig:state23}
		\end{subfigure}
		\begin{subfigure}[b]{\linewidth}
			\setlength{\abovecaptionskip}{0pt} 
			\setlength{\belowcaptionskip}{0pt} 
			\centering
			\includegraphics[width=\linewidth]{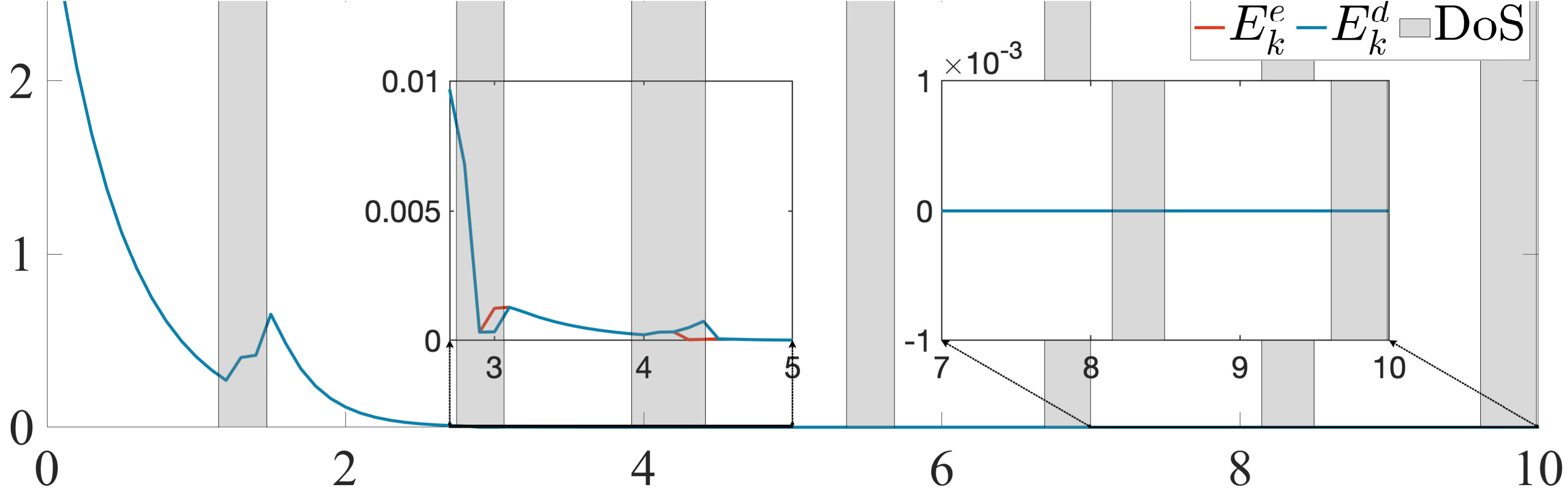}
			\caption{The values of $E_k^e$ and $E_k^d$} 
			\label{fig:ek23}
		\end{subfigure}
		\caption{\textbf{\textit{Strategy 3}} in Example B}
		\label{fig:control3B}
	\end{minipage}
	\begin{minipage}{0.32\linewidth}
		\begin{subfigure}[b]{\linewidth}
			\setlength{\abovecaptionskip}{0pt} 
			\setlength{\belowcaptionskip}{0pt} 
			\centering
			\includegraphics[width=\linewidth]{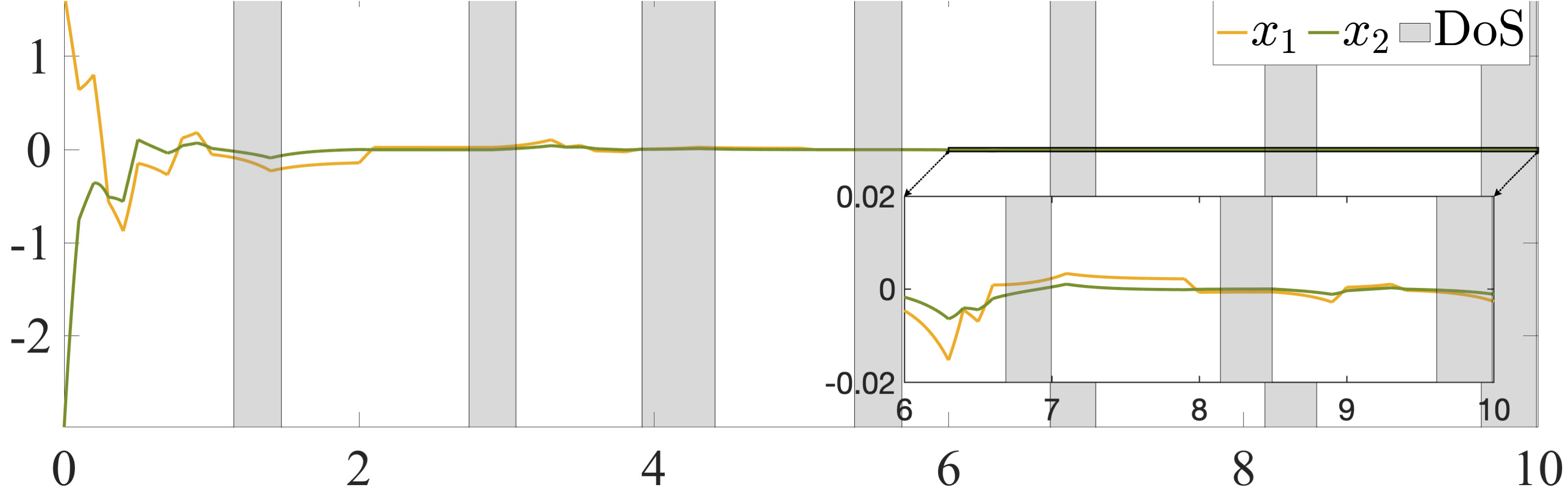}
			\caption{State trajectories} 
			\label{fig:state241}
		\end{subfigure}
		\hfill
		\begin{subfigure}[b]{\linewidth}
			\setlength{\abovecaptionskip}{0pt} 
			\setlength{\belowcaptionskip}{0pt} 
			\centering
			\includegraphics[width=\linewidth]{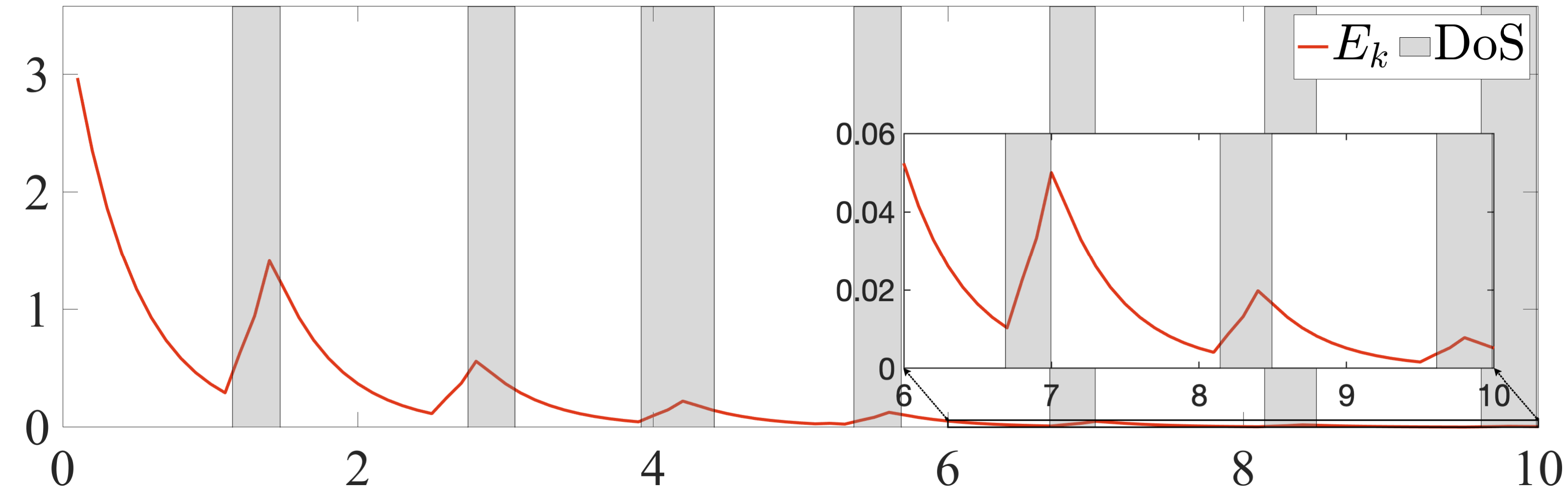}
			\caption{Quantization parameter $E_k$} 
			\label{fig:ek241}
		\end{subfigure}
		
		\caption{\textbf{TT strategy} for Example B}
		\label{fig:tt_strategy_B}
	\end{minipage}
	\begin{minipage}{0.32\linewidth}
		\begin{subfigure}[b]{\linewidth}
			\setlength{\abovecaptionskip}{0pt} 
			\setlength{\belowcaptionskip}{0pt} 
			\centering
			\includegraphics[width=\linewidth]{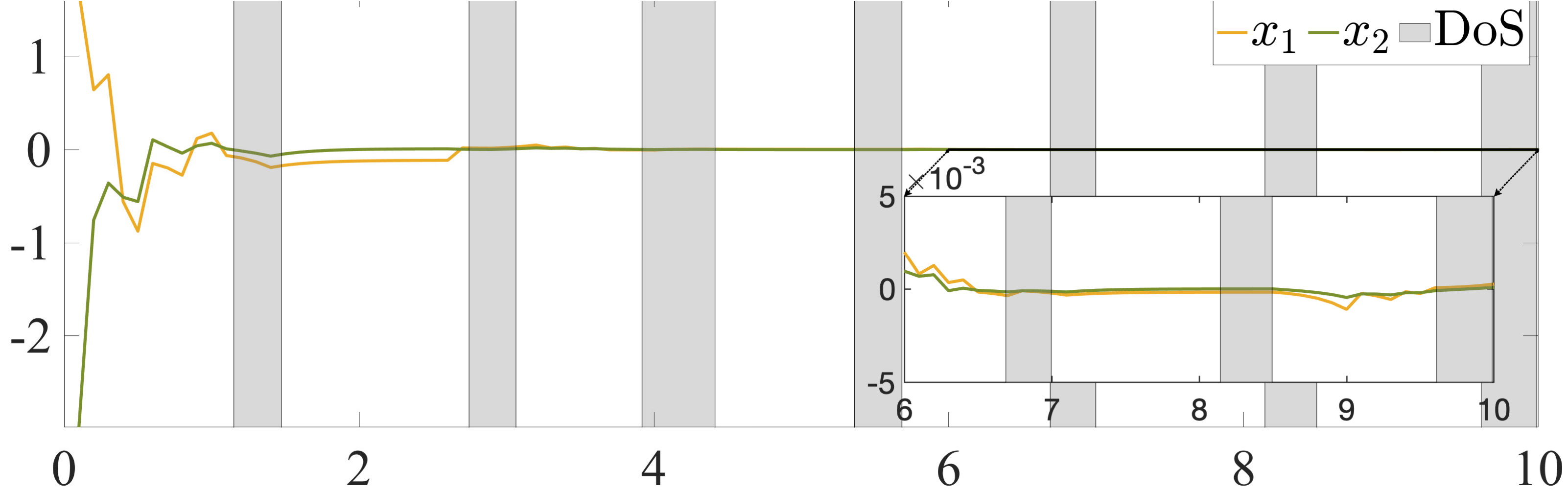}
			\caption{State trajectories} 
			\label{fig:state242}
		\end{subfigure}
		\hfill
		\begin{subfigure}[b]{\linewidth}
			\setlength{\abovecaptionskip}{0pt} 
			\setlength{\belowcaptionskip}{0pt} 
			\centering
			\includegraphics[width=\linewidth]{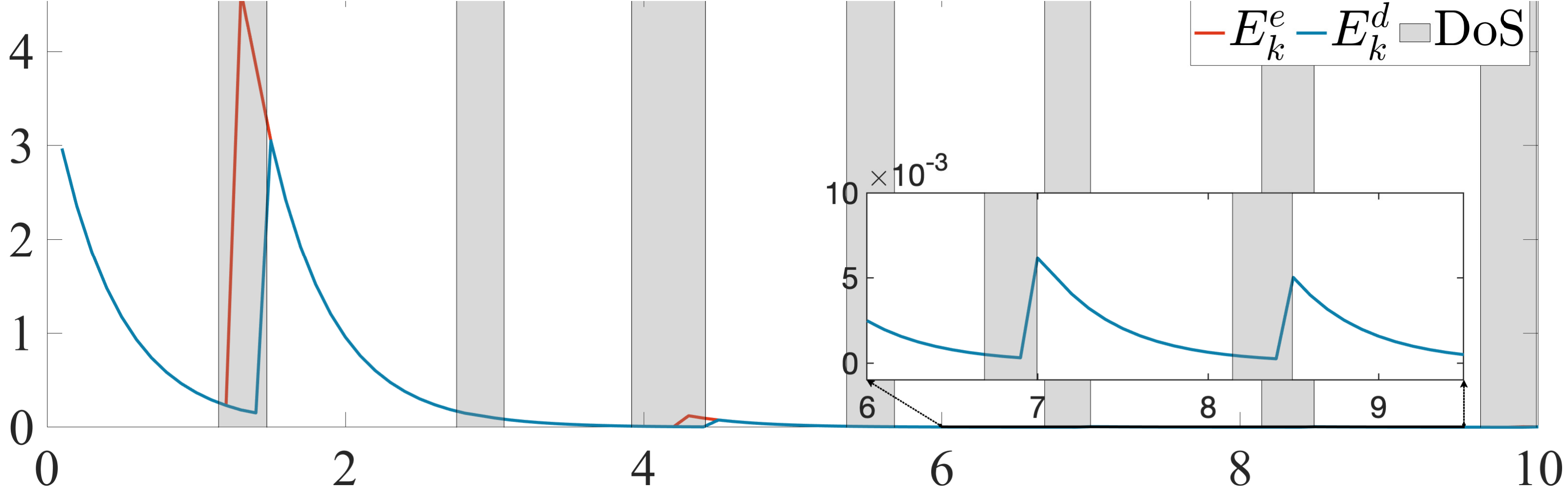}
			\caption{The values of $E_k^e$ and $E_k^d$} 
			\label{fig:ek242}
		\end{subfigure}
		
		\caption{\textbf{ET strategy} for Example B}
		\label{fig:control42B}
	\end{minipage}
\end{figure*}


	\subsection{Example C}
Consider a switched system consisting of two subsystems with parameters
{\scriptsize 	\begin{align*}
		&A_1 = \left[	\begin{array}{cccc}
			0.56 & 0.35 &  -0.30&   -0.60\\
			-0.10&   -0.64&  -0.64  & -0.73\\
			0.04   &-0.25 &  -0.42  &  0.57\\
			0.30  &  0.57  & -0.58 &  -0.50
		\end{array}\right]	\\
		&	A_2= \left[	\begin{array}{cccc}
			-0.50   &-0.64   &-0.74 &   0.30\\
			0.560 &   0.35   &-0.30  & -1.26\\
			-0.60  &  0.57  & -0.80&   -0.60\\
			0.40&  -0.25   & 0.42 &  -1.57
		\end{array}\right]\\
		&	B_1= \left[	\begin{array}{cc}
			1.91&   1.58\\
			0.69   &-0.62\\
			-1.05  & -1.82\\
			-0.20  &  0.39
		\end{array}\right]
		~	B_2= \left[	\begin{array}{cc}
			1.27  & -1.28\\
			1.32    &0.50\\
			-1.13 &  1.21\\
			-1.85  & 1.17\\
		\end{array}\right].
\end{align*}}
The controller gains are
{\scriptsize 	\begin{align*}
		K_1= &\left[	\begin{array}{cccc}
			-1.4037 &  -0.2758  & -0.1832   & 0.0781\\
			-0.8063 &   0.1959  &  0.3296  & -0.3401
		\end{array}\right]	,~\\
		K_2=& \left[	\begin{array}{cccc}
			-0.2191 &   -0.3600   & 0.2030   & 0.4401 \\
			0.2208  & -0.8859 &   0.1664  &  0.8850
		\end{array}\right].
\end{align*}}
\color{blue}
Other parameters are set as follows: $\lambda_{1} = 0.9482$, $\lambda_{2} = 0.9684$, $\eta_1 = 1.1879$, $\eta_2 = 1.2069$, $\rho_{1} = 1.0233$, $\rho_{2} = 1.0781$, $\xi_{12} = 1.8257$, $\xi_{21} = 1.4856$, 
$\hat \lambda_{1} = 0.9533$, $\hat\lambda_{2} = 0.9543$, $\hat\eta_1 = 1.0991$, $\hat\eta_2 = 1.1041$, $\hat\rho_{1} = 1.9587$, $\hat\rho_{2} = 1.9391$, $\hat\xi_{1} = 1.1490$, $\hat\xi_{2} = 1.2481$, and the sampling period $\tau_s = 0.1$. The maximum asynchronous interval is  $N_{\max} = 2$.

\color{black}
\textbf{\textit{Strategy 1:}}
Using  \Cref{thm1}, the quantization level is determined to satisfy $N > 1.2393$. Thus, we choose $N = 3$, $\tau_d = 5.9$, $\tau_D = 1.2$, and $T = 1.5$. Based on the system definition and the corresponding matrices, we can get $b = 0.3309$, $\nu_1 = 0.9482$, $\nu_2 = 0.9684$, $\hat{\nu}_1 = 0.9703$, $\hat{\nu}_2 = 1.0440$, $\mu_{12}^1 = 2.9589$, $\mu_{21}^1 = 2.1687$, $\mu_{12}^2 = 1.7930$, $\mu_{21}^2 = 1.7930$, $\mu_{12}^3 = 1.1879$, and $\mu_{21}^3 = 1.2069$. It is verified that conditions \eqref{equ_con1} and \eqref{equ_con2} are fulfilled.

Fig. \ref{fig:switchingC} illustrates a feasible switching signal, including the system mode, controller mode, and the DoS attack intervals. We can see that the asynchronous behavior occurs due to the DoS attack. The corresponding state trajectories of switched systems under DoS attack and quantization, with the initial condition $x_0 = [-2, 2, 1, -1]^T$, are presented in Fig. \ref{fig:stateC}. It is clear that the system is stable. 
Furthermore, Fig. \ref{fig:ekC} depicts the evolution of quantizer parameters $E_k^e$ and $E_k^d$ for the encoder and decoder, respectively. The interval $[6.8, 7.8]$ is magnified to highlight the differences between the encoder and decoder parameters during the DoS attack. It can be observed that $E_k^e$ and $E_k^d$ are the same during attack-free intervals, confirming the effectiveness of  \Cref{lemma2}.

Finally, we verify the result in \Cref{coro1}. From \eqref{equ_con2_1}, the dwell time constraint is given by $\tau_d \geq 2.8382$. When $N = 3$, the condition \eqref{equ_con1_1} simplifies to $-0.7849 + \left(\frac{1}{T} + \frac{\tau_s}{\tau_D} - 1\right) 1.0986 \leq 0$. A feasible solution is  $T = 1.1$ and $\tau_D \leq 0.1242$. Compared to \Cref{thm1}, the conservatism in the DoS attack and switching signal constraints has been significantly reduced in \Cref{coro1}.

\begin{figure*}
	\begin{subfigure}[b]{0.32\linewidth}
		\centering
		\includegraphics[width=\linewidth]{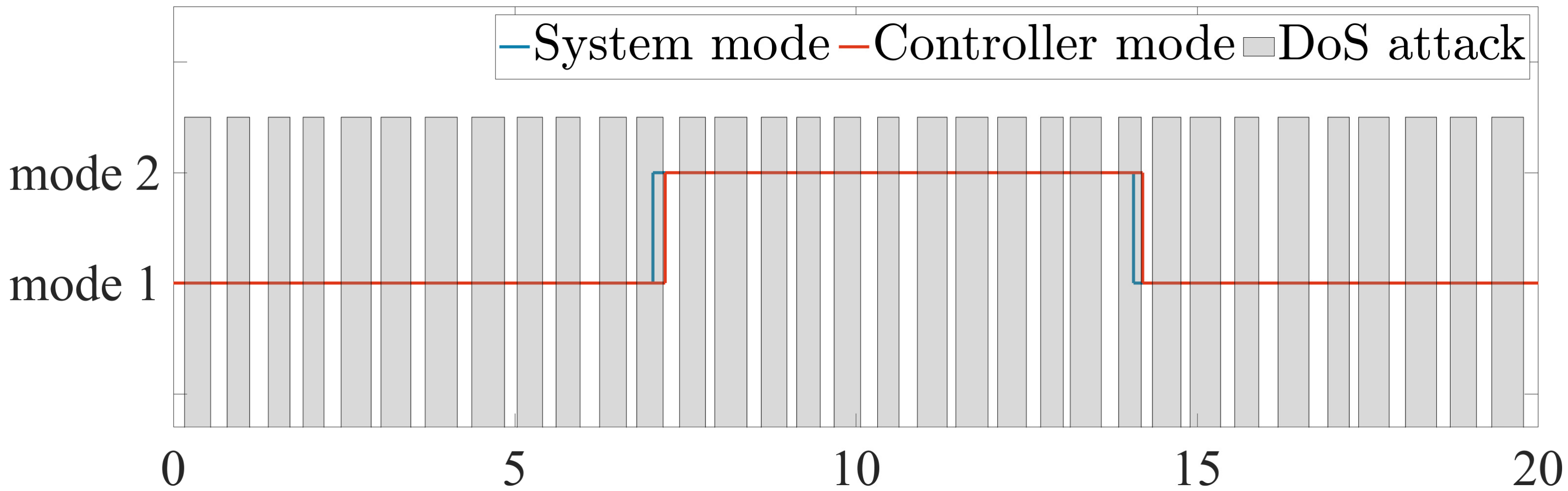}
		\caption{The switching signal and DoS attack} 
		\label{fig:switchingC}
	\end{subfigure}
	\hfill
	\begin{subfigure}[b]{0.32\linewidth}
		\centering
		\includegraphics[width=\linewidth]{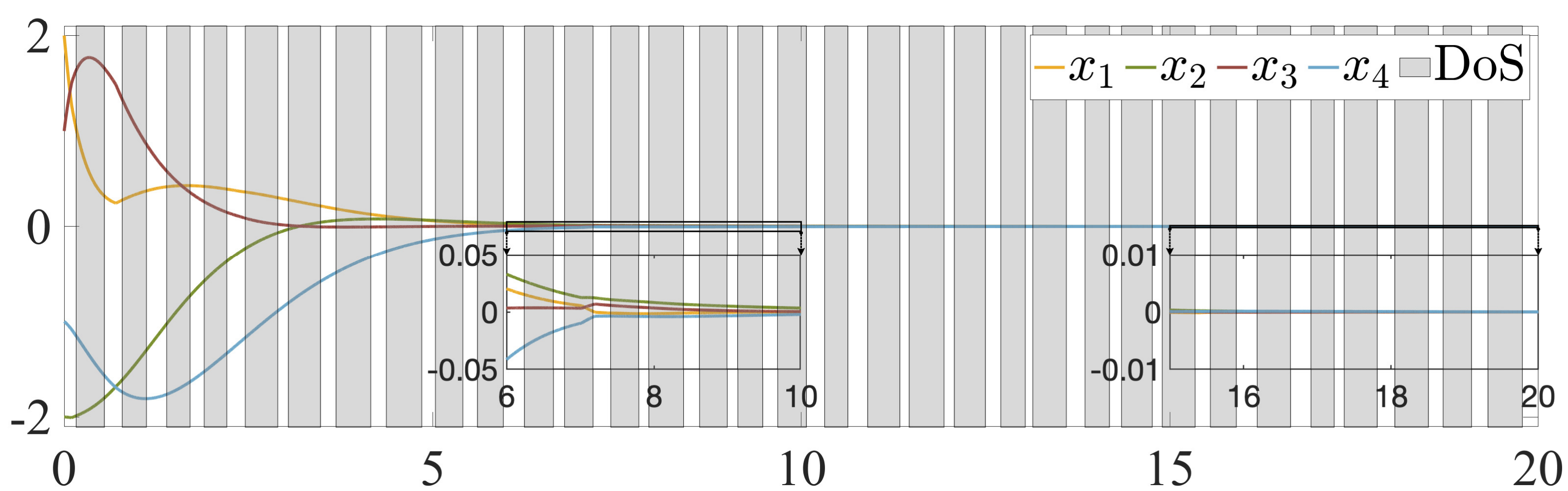}
		\caption{State trajectories} 
		\label{fig:stateC}
	\end{subfigure}
		\hfill
	\begin{subfigure}[b]{0.32\linewidth}
		\centering
		\includegraphics[width=\linewidth]{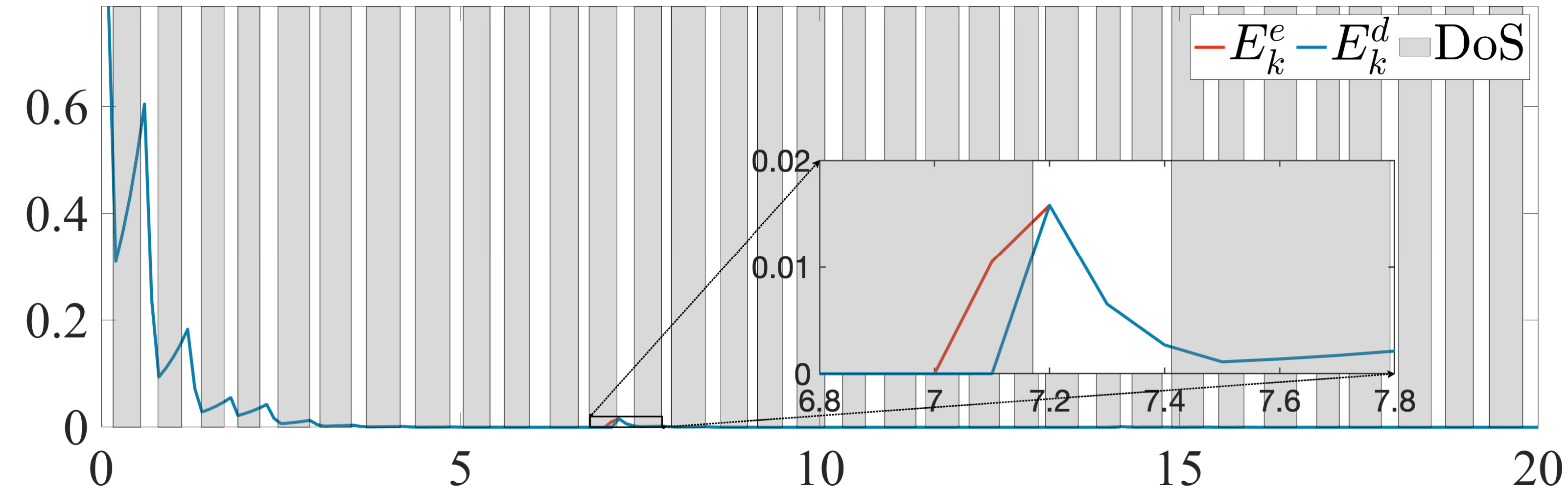}
		\caption{The values of $E_k^e$ and $E_k^d$} 
		\label{fig:ekC}
	\end{subfigure}
	\caption{Simulation results for \textbf{\textit{Strategy 1}} in Example C}	
	\label{fig:strategy1C}
\end{figure*}

\textbf{\textit{Strategy 2:}}
From  \Cref{thm1_ori}, we obtain $N \geq 39.5517$. For a faster convergence rate, we set  $N = 105$. Moreover, the dwell time constraint computed from \eqref{equ_con2_origin} is $4.5123$. Let the dwell time constraint be $4.6$, condition \eqref{equ_con1_origin} then  turns to be $\frac{0.1071}{\tau_D} + \frac{0.3208}{T} \leq -0.1600$, which is infeasible. Therefore, the dwell time constraint needs to be relaxed. By setting $\tau_d = 9$, we get $\frac{0.1071}{\tau_D} + \frac{0.3208}{T} \leq 0.0627$. 
\textcolor{blue}{
	Next, by defining $\tau_D = 5$ and $T = 10$, condition \eqref{equ_con1_origin} holds. Additionally, the maximum number of consecutive DoS attack periods is no more than 18, i.e., $n_{\max} \leq 18$, with a total duration less than 1.8 seconds. According to the frequency and duration constraints, the average length of a DoS attack is $\tau_D / T = 0.5$, which is less than 1.8 seconds. Hence, the maximum number of DoS attack is not too conservative.}
As shown in Fig. \ref{fig:stateC2}, the system remains stable. In addition, for the first subsystem, the DoS attack does not increase the quantization parameter, as illustrated in Fig. \ref{fig:ekC2}. Moreover, the quantization parameters are identical during attack-free intervals.

\textbf{\textit{Strategy 3:}}
Similarly, the quantization level admits $N \geq 20.6266$. We set $N = 155$ and other parameters  $\tau_d = 20$, $\tau_D = 20$, and $T = 20$. Under these settings, the conditions in \Cref{thm1_zero} are satisfied.
The switching signal, state trajectories, and quantization parameters are shown in Fig. \ref{fig:control_C3}. It is seen that the system under \textbf{\textit{Strategy 3}} is stable.

\begin{figure*}
	\hfill
	\begin{subfigure}[b]{0.32\linewidth}
		\centering
		\includegraphics[width=\linewidth]{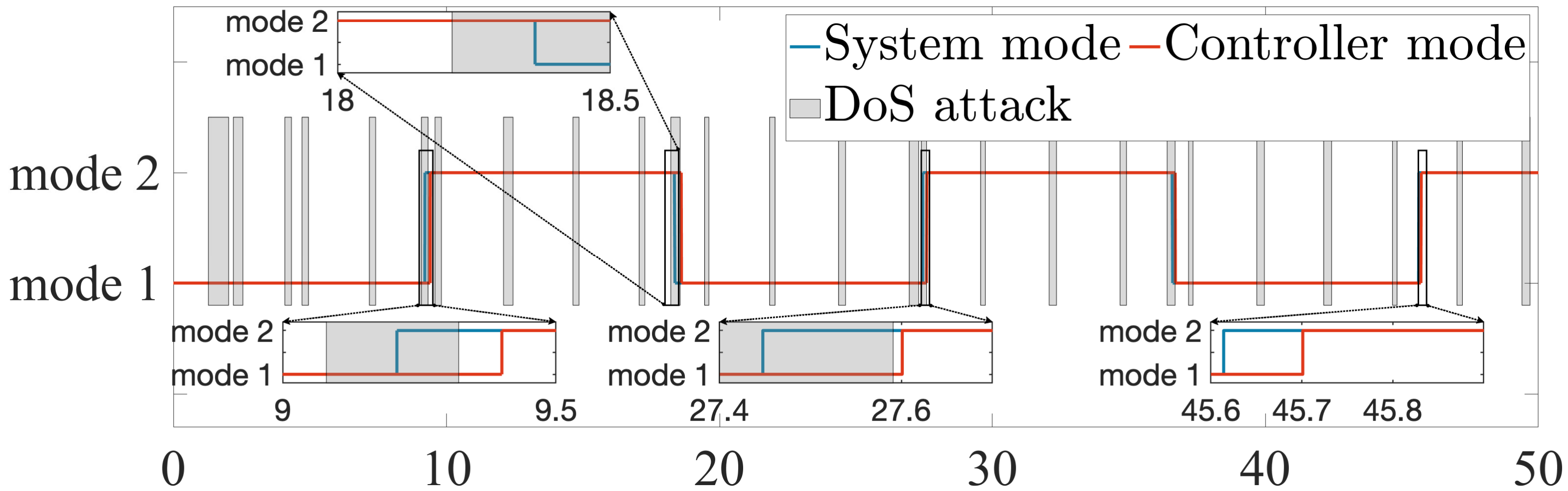}
		\caption{The switching signal and DoS attack} 
		\label{fig:switchingC2}
	\end{subfigure}
		\hfill
	\begin{subfigure}[b]{0.32\linewidth}
		\centering
		\includegraphics[width=\linewidth]{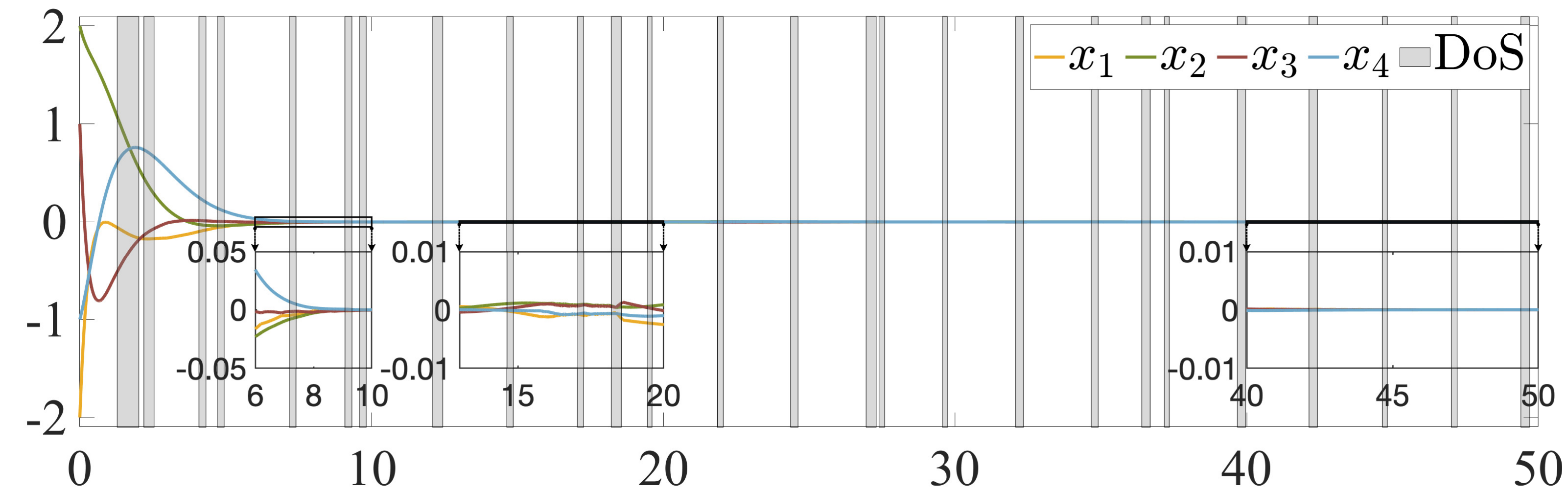}
		\caption{State trajectories} 
		\label{fig:stateC2}
	\end{subfigure}	\hfill
	\begin{subfigure}[b]{0.32\linewidth}
		\centering
		\includegraphics[width=\linewidth]{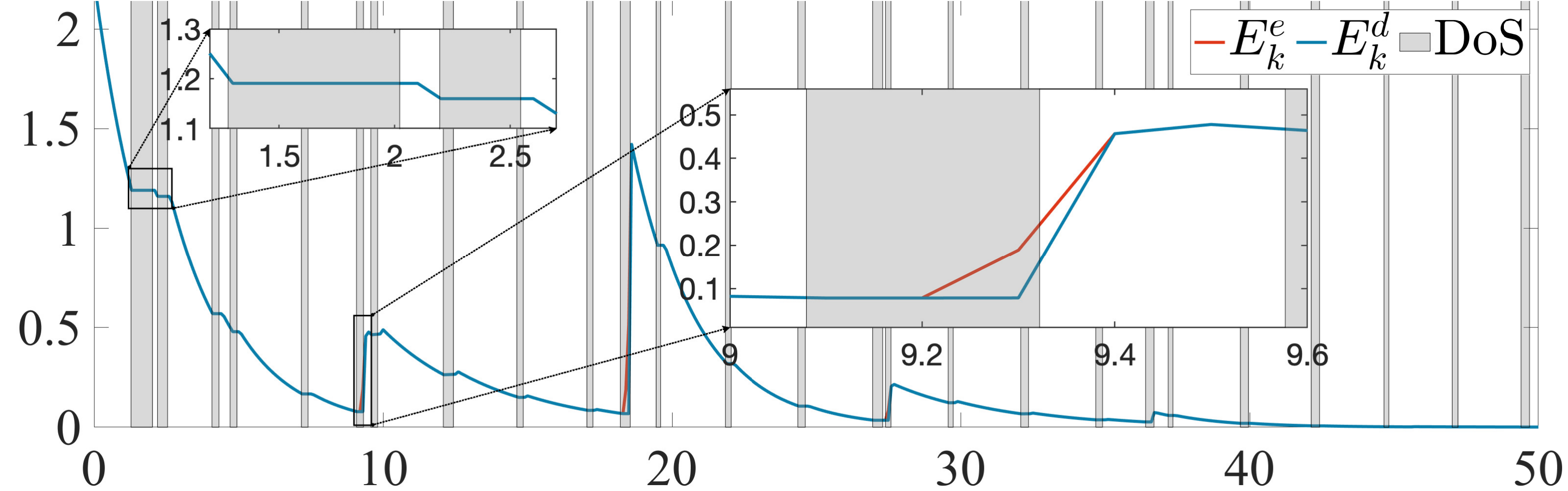}
		\caption{The values of $E_k^e$ and $E_k^d$} 
		\label{fig:ekC2}
	\end{subfigure}
	\hfill
	\caption{Simulation results for \textbf{\textit{Strategy 2}} in Example C}	
	\label{fig:control_C2} 
\end{figure*}
\begin{figure*}
		\hfill
	\begin{subfigure}[b]{0.32\linewidth}
		\centering
		\includegraphics[width=\linewidth]{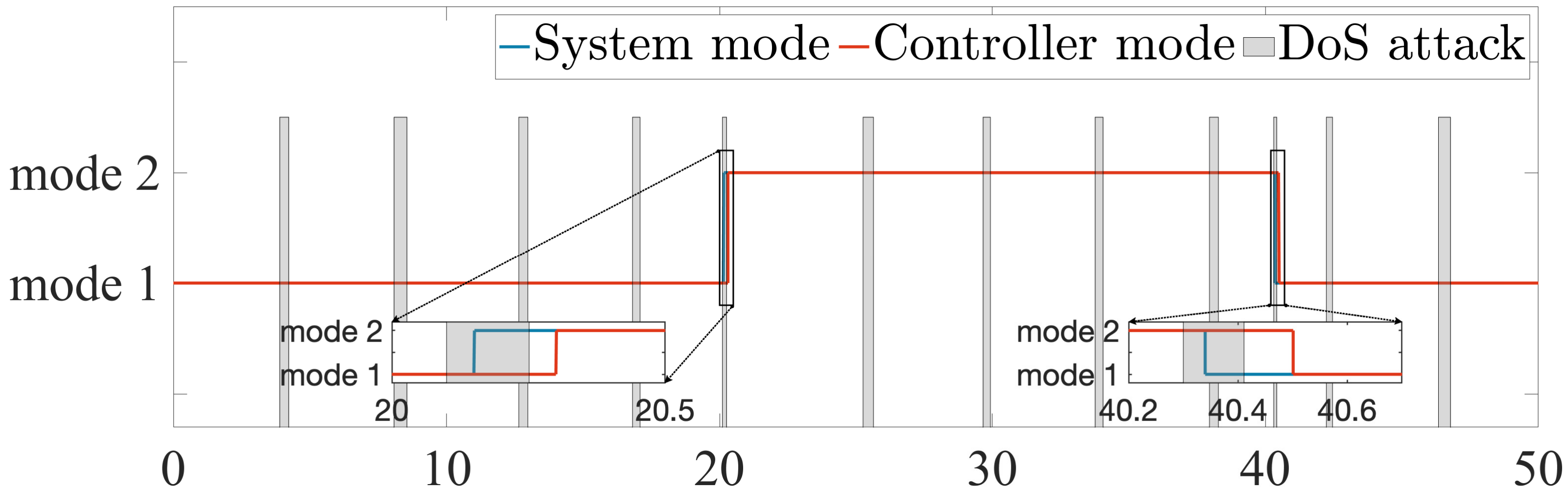}
		\caption{The switching signal and DoS attack} 
		\label{fig:switchingC3}
	\end{subfigure}	\hfill
	\begin{subfigure}[b]{0.32\linewidth}
		\centering
		\includegraphics[width=\linewidth]{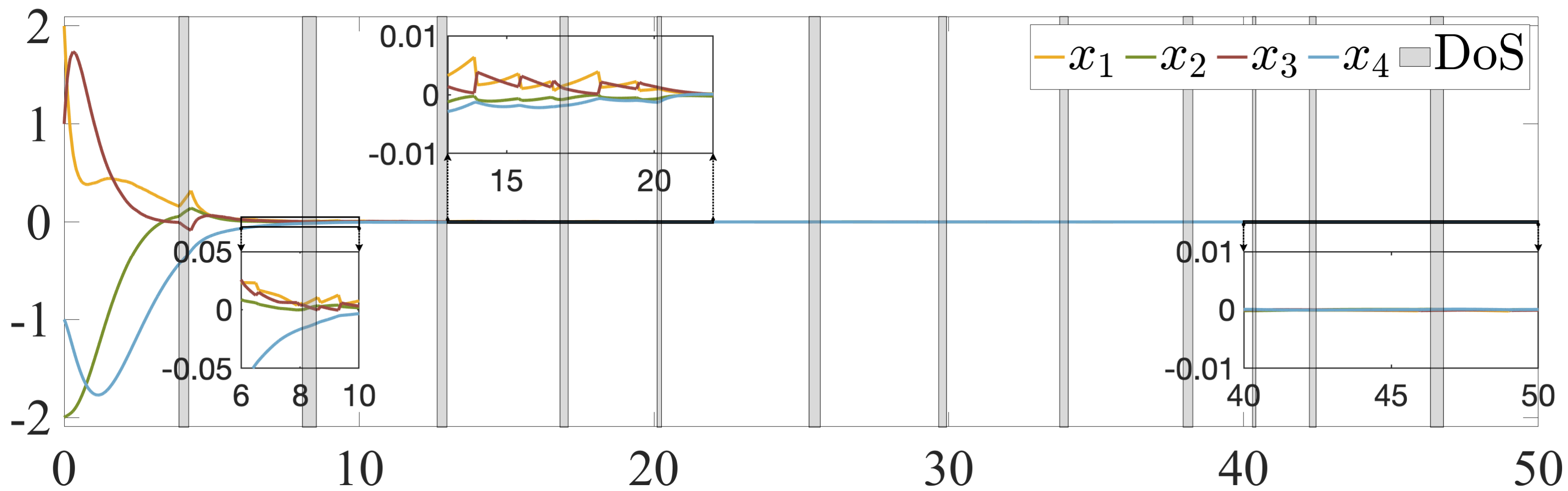}
		\caption{State trajectories} 
		\label{fig:stateC3}
	\end{subfigure}	\hfill
	\begin{subfigure}[b]{0.32\linewidth}
		\centering
		\includegraphics[width=\linewidth]{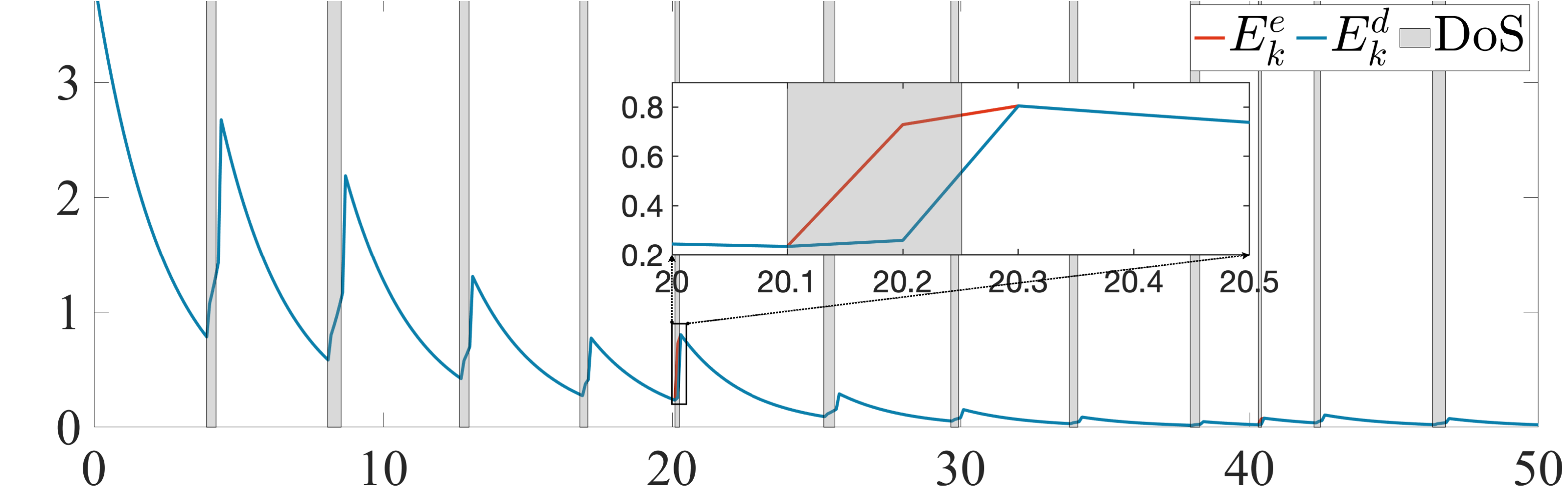}
		\caption{The values of $E_k^e$ and $E_k^d$} 
		\label{fig:ekC3}
	\end{subfigure}	\hfill
	\caption{Simulation results for \textbf{\textit{Strategy 3}} in Example C} 
	\label{fig:control_C3}
\end{figure*}

	\section{Conclusion} \label{sec_con}
The resilient control problem for switched systems under DoS attack has been addressed. Two active quantized control strategies are proposed to enhance the system's defense capabilities. During the DoS attack interval, an appropriate control signal has been designed with predicted state information to compensate for the signal loss during the attack interval. One of the strategies offers better performance but requires more computational resources, while the others exhibit a trade-off between system performance and computational burden. 
	Additionally, a well-designed switching signal is developed to enhance the system's tolerance to DoS attack. Furthermore, two control strategies without ACK signals has been provided in terms of the  event-triggered and time-triggered approaches. The updating law is designed to ensure that the state remains within the quantization range and to maintain synchronization between the decoder and encoder, even under asynchronous conditions caused by DoS attack. 
	Finally, several criteria  for ensuring system stability have been derived, which are related to the switching signal and the constraints imposed by the DoS attack.

\textcolor{blue}{
	Future research will extend the quantized control design for switched systems to multi-agent systems with switching topologies. This extension will leverage existing stochastic multi-agent system frameworks with infinite topologies \cite{R1-2,R1-3,R1-4} and prior studies on switching within finite topological sets \cite{R1-5}. Additionally, 
	we can extend the proposed framework to nonlinear switched systems subject to noisy/lossy networks, and embed real-time, learning-based DoS detection to achieve higher resilience without ACK feedback.
}

\color{black}

\appendix
\color{blue}
\subsection{Parameters chosen}
 The parameters $\tau_s$ and $N_{\max}$ are chosen to satisfy practical hardware and bandwidth constraints.  The remaining parameters $\rho_p$, $\lambda_p$, $\eta_{pq}$, and $\xi_{pq}$, are computed from inequalities that capture the switched system dynamics. For every admissible transition $p \to q$, we require
$\Bigg\| \Bigg[ \begin{matrix} 
	A_{pq}^d(\tau_s) \\ 
	\mathbf{0}
\end{matrix}\Bigg]^k \Bigg\| \leq \xi_{pq} \eta_{pq}^k $ \text{and} $\quad 
\| (A_p^d)^k \| \leq \rho_p \lambda_p^k,$
where $A_p^d = e^{\overline{A}_{pp} \tau_s}$ and $A_{pq}^d(\overline{t}) = \big[\begin{matrix}
	\mathbf{I} & \mathbf{0} \end{matrix}\big] e^{\widetilde{\mathcal{A}}_{pq}(\overline{t})} e^{\overline{\mathcal{A}}_p (\tau_s - \overline{t})}$.  In practice, we minimize $\lambda_p$ and $\eta_{pq}$ while keeping $\rho_p$ and $\xi_{pq} $ close to $1$ as feasible.
\textbf{\textit{Strategy 1}} employs a systematic parameter selection workflow. 
First, $\Gamma$ is derived from system matrix $A$ and sampling period $\tau_s$. The quantization level $N$ is then chosen such that $N > \Gamma$ to guarantee bounded quantization errors. Next, the update laws for the quantization parameters are calculated by deriving $\Gamma_p$, $\Gamma_{pq}^1$, $\Gamma_{pq}^2$, and $\Gamma_{pq}^3$ for all modes. These quantities yield the global parameters $\Gamma$, $\overline{\Gamma}$, $\overline{\Gamma}_2$, and $b$. Stability-related scalars $\nu_p$, $\hat{\nu}_p$, $\mu_{pq}^1$, $\mu_{pq}^2$, and $\mu_{pq}^3$ are subsequently evaluated from these results. A reference dwell time is obtained via (30), and the DoS attack constraints are constructed via (29). 
Whenever the resulting constraints are too restrictive, the dwell time is lengthened or $N$ is increased, and the workflow is re-iterated until a provably stable configuration is reached.

\begin{table}
	\caption{Definitions in \textbf{\textit{Strategies 1-2}}}
	\label{table12}
	\renewcommand\arraystretch{1.5}
	\setlength{\tabcolsep}{3pt}
	\begin{tabular}{p{240pt}}
		\hline
$\overline{A}_{\sigma(t)\hat \sigma(t)} = A_{\sigma(t)}+B_{\sigma(t)}K_{\hat \sigma(t)}$ \\
$A_p^d  = e^{\overline{A}_{pp} \tau_s}$,  $B_p^d = \int_{0}^{\tau_s} e^{\overline{A}_{pp} (\tau_s-s)}B_pK_pe^{{A}_p s}ds$\\
 $\Pi_{pq}^1 =  \overline{A}_{pq} - \overline{A}_{qq}$,  $\Pi_{pq}^2 = B_pK_q- \overline{A}_{qq}$ \\
 $e_z(t) = \left[ \begin{array}{c}	x(t) \\	e(t)\end{array}\right]$, $	\tilde{x}(t)  =  \left[\begin{array}{c}
 	{x}(t)	\\
 	\textbf{0}
 \end{array}\right]$, $	\tilde{e}(t)  =  \left[\begin{array}{c}
 \textbf{0} \\ {e}(t) 
 \end{array}\right]$\\
 $\widetilde{\mathcal{A}}_{pq} = \left[\begin{array}{cc}
 	A_p + B_p K_q	& B_p K_q  \\
 	\Pi_{pq}^1 &  \Pi_{pq}^2
 \end{array}
 \right]$, 
 $ \overline{\mathcal{A}}_{p} =  \left[\begin{array}{cc}
 	A_p	& B_pK_p  \\
 	\textbf{0}	&  A_p 
 \end{array}\right]$\\
 $ \mathcal{A}_{pq} =  \left[\begin{array}{cc}
 	A_p & B_pK_q \\
 	\textbf{0}& A_q+B_qK_q
 \end{array}
 \right]$, $ \hat{\mathcal{A}}_{q} =\left[\begin{array}{cc}
 	\overline{A}_{qq}&\textbf{0}  \\
 	\textbf{	0}& \overline{A}_{qq}
 \end{array}
 \right] $ \\
$z(t) = \left[\begin{array}{c}
	x(t) \\
	\hat x(t)
\end{array}
\right]$, $\hat z(  t_k)=\left[\begin{array}{c}
 	\hat x(  t_k^+ ) \\
 	\hat x(  t_k^+ )
 \end{array}
 \right]$ \\
 $A_{pq}^d (\overline{t})  = \left[\begin{array}{cc}
 	\textbf{I} &\textbf{0}
 \end{array}\right] e^{\widetilde{\mathcal{A}}_{pq}(\overline{t})}e^{ \overline{\mathcal{A}}_{p} (\tau_s -  \overline{t})} $, $B_{pq}^d (\overline{t})  =  A_{pq}^d (\overline{t})  $\\
 $ \| ({A}_p^d)^k\| \leq \rho_{p} \lambda_p^k$, $\rho_p>0$, $0<\lambda_p <1$ \\
 $  \left \| \left[\begin{array}{c}
 	{A}_{pq}^d(\tau_s)	\\
 	\textbf{0}
 \end{array}\right]^k\right\| \leq \xi_{pq} \eta_{pq}^k$,~$\eta_{pq} ,~\xi_{pq}>0$\\
		\hline
	\end{tabular}
	\label{tab122}
\end{table}

\begin{table}
	\caption{Parameters used in \textbf{\textit{Strategy 1}}}
	\label{table5}
	\renewcommand\arraystretch{1.5}
	\setlength{\tabcolsep}{3pt}
	\begin{tabular}{p{240pt}}
		\hline
	$\Gamma_p  = \| e^{A_p \tau_s} \|$\\
$\Gamma_{pq}^{1} = \hat \Gamma_{pq}^{1}+\Gamma_{pq}^{2} \frac{N}{N-1}$\\
	$\hat \Gamma_{pq}^{1}  = \max_{\overline{t}\in [0,\tau_s)} \left \|e^{\mathcal{A}_{pq} \overline{t}}  e^{\mathcal{A}_{qq} ( \tau_s -\overline{t})} \right\| $\\
    $\Gamma_{pq}^{2}  =  \max_{\overline{t}\in [0,\tau_s)} \left \|e^{\mathcal{A}_{pq} \overline{t}}  e^{\mathcal{A}_{qq} ( \tau_s -\overline{t})}-e^{\hat{\mathcal{A}}_{q} \tau_s } \right \| $\\
	$\Gamma_{pq}^3 = \|e^{ \mathcal{A}_{pq}\tau_s}\| $,~$	\Gamma_{pq}^4(n) =( \Gamma_{pq}^3)^n \Gamma_{pq}^1 $,~	$	\Gamma_{pq}^5(n) = (\Gamma_{pq}^3)^n \Gamma_{pq}^2 $\\
	$n = \frac{\lceil t_{s} \rceil - \ulcorner t_s \urcorner }{\tau_s}$\\
	$\Gamma = \max_{p\in\mathcal{M}}  \Gamma_{p}$, $\overline{\Gamma}_2 = \max \limits_{p,q\in \mathcal{M}}\Gamma_{pq}^{2}$\\
	 $\overline{\Gamma}=  \max\limits_{p,q\in\mathcal{M},p\neq q, m\in [0,N_{\max}-1]}(\Gamma_{pq}^3)^{m-1}\Gamma_{pq}^{1}$\\
	  	  $\nu_p = \max\{ \lambda_p , \rho_{p}  \| {B}_{p}^d \| +b \}$, 
$ \hat{\nu}_p= \max\{ \rho_p\lambda_p , \rho_{p}  \| {B}_{p}^d \| +b\}$\\
 $ \mu_{pq}^1 = \max\{\mathbb{A}_{pq}^d  +a\overline{ \Gamma}_2,\mathbb{B}_{pq}^d +b+a\overline{ \Gamma}_2 \}$\\
  $\mathbb{A}_{pq}^d  = \max_{\overline{t} \in [0,\tau_s]} \| \tilde{A}_{pq}^d (\overline{t})\|$,~$\mathbb{B}_{pq}^d  = \max_{\overline{t} \in [0,\tau_s]} \| \tilde{B}_{pq}^d (\overline{t})\|$\\
  	$\mu_{pq}^2 = \max \{ \xi_{pq} \eta_{pq} , \xi_{pq}  \| \tilde{B}_{pq}^d (\tau_s)\| +b\}$\\
  	$\mu_{pq}^3 = \max \{ \eta_{pq}, \xi_{pq}  \| \tilde{B}_{pq}^d (\tau_s)\|+b \}$\\
  	$a = N^{\frac{\overline{\kappa}}{\tau_s}}$, $b = \left(\frac{\overline{\Gamma}}{\Gamma^{N_{\max}}  }\right)^{\frac{\tau_s}{\tau_d}} N^{\frac{1}{\overline{T}}}\frac{\Gamma}{N}$\\ 
		\hline
	\end{tabular}
	\label{tab5}
\end{table}

\begin{table}
	\caption{Parameters used in \textbf{\textit{Strategy 2}}}
	\label{tab6}
	\renewcommand\arraystretch{1.5}
	\setlength{\tabcolsep}{3pt}
	\begin{tabular}{p{240pt}}
		\hline
	${\Lambda}_p^1  = \rho_p \lambda_p + \frac{\rho_p\| B_p^d\| }{N} $\\
	${\Lambda}_p^2  =  \lambda_p + \frac{\rho_p\| B_p^d \| }{N} $\\
	$ {\Lambda}_p^3 = \overline{\Psi}_p^{\frac{1}{n_{\max}}}$\\ 
	$	\overline{\Psi}_p =	\max\limits_{\ell \in [0,n_{max}]} \Big \{\rho_{p} \lambda_p^\ell + \sum_{j=0}^{\ell -1} \rho_p \lambda_p^{\ell-j-1} \| B_p^d \| \frac{\| e^{{A}_p \tau_s} \| ^j }{N} \Big \}$\\
	$ {\Lambda}_{pq}^4 $ $= \max_{\overline{t}\in [0,\tau_s)}\|e^{\mathcal{A}_{pq}   \overline{t}}  e^{\mathcal{A}_{p}(\tau_s - \overline{t})}\|$\\
	$ {\Lambda}_{pq}^5 = \max \Big \{ \frac{ 1- \lambda_p\underline{\Psi}_p\overline{\Psi}_p^{-1} }{\rho_p \| B_p^d\| },1 \Big\}  {\Lambda}_{pq}^4  $\\
	$ {\Lambda}_{pq}^6 = 	 \| e^{\widetilde{\mathcal{A}}_{pq}\tau_s} \|$\\
	$\overline{\Lambda}^{i} = \max_{p \in \mathcal{M}} \overline{\Lambda}^{i}_p (i\in \{1,2,3\})$\\
	 $\overline{\Lambda}^{i} = \max_{p,q \in \mathcal{M}, p\neq q} \overline{\Lambda}^{i}_{pq} (i\in \{5,6\})$\\
	$\mathbb{A}_{pq}^d  = \max_{\overline{t} \in [0,\tau_s]} \| \tilde{A}_{pq}^d (\overline{t})\|$\\
		\hline
	\end{tabular}
\end{table}

\begin{table}
	\caption{Definitions in \textbf{\textit{Strategies 3-4}}}
	\label{table121}
	\renewcommand\arraystretch{1.5}
	\setlength{\tabcolsep}{3pt}
	\begin{tabular}{p{240pt}}
		\hline
	$\hat{A}_p^d = e^{A_p\tau_s}$, $\hat{B}_p^d = \int_0^{\tau_s} e^{A_ps}B_pds$, $\hat{A}_p^{cl} = \hat{A}_p^d + \hat B_p^dK_p$\\
	 $ \overline{t} = t_{k+1}-t_s$\\
	 $\hat{A}_{pq}^d(\overline{t}) =e^{A_p \overline{t} +A_q( \tau_s-\overline{t} )} $\\
	  $\hat{B}_{pq}^d (\overline{t})  = e^{A_p \overline{t}} \int_{0}^{\tau_s-\overline{t}} e^{A_q s } B_q ds+\int_{0}^{\overline{t}} e^{A_p s } B_p ds $\\
	  $\hat{A}_{pq}^{cl} (\overline{t}) =   \hat{A}_{pq}^{d}(\overline{t}) +  \hat{B}_{pq}^d (\overline{t}) K_{q}$\\
	  $ \| (\hat{A}_p^{cl})^k\| \leq \hat {\rho}_{p} \hat \lambda_p^k$, $0<\hat{\lambda}_p <1$, $\hat \rho_{p} >0$\\
	  $ \| (\hat{A}_p^d)^k\| \leq \hat \xi_{p}  \hat\eta_{p}^k$,  $\hat{\eta}_p>0$, $\hat \xi_{p}>0$  \\
		\hline
	\end{tabular}
\end{table}

\begin{table}
	\caption{Definitions in \textbf{\textit{Strategy 3}}}
	\label{table122}
	\renewcommand\arraystretch{1.5}
	\setlength{\tabcolsep}{3pt}
	\begin{tabular}{p{240pt}}
		\hline
		$\tilde{\xi}_{p}  = \max\{\hat{\xi}_{p} ,1\}$\\
	${\Upsilon}_p^1  = \hat \rho_p \hat \lambda_p + \frac{\hat \rho_p\| \hat {B}_p^dK_p\| }{N} $\\
	${\Upsilon}_p^2  =  \hat \lambda_p + \frac{\hat\rho_p\| \hat B_p^d K_p\| }{N} $\\
	$\Upsilon _{pq}^3 = \max _{\overline{t} \in [0,\tau_s]} \left( \| \hat{A}_{pq}^d(\overline{t})+ \hat{B}_{pq}^d (\overline{t}) K_{q}\frac{  N-1}{N}  \| \right) $\\
	 $ {\Upsilon}_{pq}^4 = \max_{\overline{t} \in [0,\tau_s]} \| \tilde{A}_{pq}^d(\overline{t})\|$\\
	  $	\widehat {\Upsilon }_{pq}   =\max \{	\Upsilon _{pq }^3, 	\Upsilon _{pq }^4\}$\\
	  $ \tilde{\Upsilon} = \max_{p, q\in \mathcal{M}} \Big\{ \widehat {\Upsilon }_{pq} \tilde{\xi}_{p},\frac{\Upsilon_{pq}^3 \Upsilon_p^1}{\Upsilon_{p}^2}\Big\} $\\
	   $\tilde{\Upsilon}^i = \max_{p\in\mathcal{M}}\{{\Upsilon}^i_p\} ~(i\in \{1,2\})$\\
	   $\overline{\xi}  = \max_{p \in \mathcal{M}} \tilde{\xi}_{p}$, $\overline{\eta}  = \max_{p\in \mathcal{M}} \hat{\eta}_{p}$\\
		\hline
	\end{tabular}
\end{table}

\begin{table}
	\caption{Definitions in \textbf{\textit{Strategy 4}}}
	\label{table124}
	\renewcommand\arraystretch{1.5}
	\setlength{\tabcolsep}{3pt}
	\begin{tabular}{p{\linewidth}}
		\hline
${\varphi}_1=\max_ {p \in \mathcal{M}}  \tilde \rho_p \hat \lambda_p + \frac{\hat \rho_p\| \hat {B}_p^dK_p\| }{N} $\\
$\varphi_2 = \max_ {p \in \mathcal{M}}   \hat \lambda_p + \frac{\hat \rho_p\| \hat {B}_p^dK_p\| }{N}  $\\
$\varphi_3 = \max_{p\in\mathcal{M}}\tilde{\xi}_{p} \hat{\eta}_p $\\
$\varphi_4 = \max_{p\in\mathcal{M}}\hat{\eta}_p $\\
$\varphi_5 = \max_{p\in\mathcal{M}}\{ \hat{\xi}_{p}, \big(\tilde \rho_p \hat \lambda_p + \frac{\hat \rho_p\| \hat {B}_p^dK_p\| }{N} \big)/ \big(\hat \lambda_p + \frac{\hat \rho_p\| \hat {B}_p^dK_p\| }{N}\big),1 \}$\\
 $\tilde{\rho}_p= \max_ {p \in \mathcal{M}} \{\hat{\rho}_p,1\}$\\
 $\tilde{\eta}_p= \max_ {p \in \mathcal{M}} \{\hat{\eta}_p,1\}$\\
${\varphi}^1= \max\limits_{p\in \mathcal{M}}\tilde \rho_p\hat  \lambda_p + \frac{\hat \rho_p\| 	\hat {B}_p^dK_p\| }{N} $\\
$\varphi^2 =\max\limits_{p\in \mathcal{M}}  \hat \lambda_p + \frac{\hat \rho_p\| \hat {B}_p^dK_p\| }{N}  $\\
$\varPhi(n_z,k) =\varphi^1(\varphi^2 )^{n_z-1}\phi {E}_{ k-n_z-n_{\max}} ^d+(\varphi^2 )^{n_z-1}(	\varphi^1)^2 E_{k-n_z} ^d$\\
 $\phi = 	 \max \limits_{\substack{p,q\in \mathcal{M}, \ell\in [0,n_{\max}], \overline{k}\in [0,\ell-1] }} \big( \| \varpi^x_{pq}(\ell,\overline{k})\| +\sum\limits_{i=0}^{\ell-1} \frac{  \left\| \varpi^e_{pq}(i,\overline{k}) \right\| } {N}  \big)$\\ $	\varpi^x_{pq}(\ell,\overline{k}) = (A_p^{cl})^{\ell-\overline{k}}( A_q^{cl} )^{\overline{k}}-  (\hat A_p^{d})^{\ell-\overline{k}}( \hat A_q^{d} )^{\overline{k}}$ \\ $  \overline{\varpi}^e_{pq}(i,\overline{k})  = \| \hat B_p^d K_p\|  + \sum\limits_{i=0}^{\overline{k}}\|(A_p^{cl})^{\ell-\overline{k}}( A_q^{cl} )^{\overline{k}-i}\hat B_q^dK_q\| +  \sum_{i=\overline{k}+1}^{\ell-1}\|(A_p^{cl})^{\ell-i}B_p^d K_p\|$\\
		\hline
	\end{tabular}
\end{table}


\begin{thebibliography}{10}

\bibitem{DT}
D.~Liberzon, {\em Switching in Systems and Control}.
\newblock Birkhäuser Boston, 2003.

\bibitem{App1}
L.~Menhour, A.~Charara, and D.~Lechner, ``Switched {LQR/$H_\infty$} steering
  vehicle control to detect critical driving situations,'' {\em Control
  Engineering Practice}, vol.~24, pp.~1--14, 2014.

\bibitem{App2}
S.~Chen, L.~Jiang, W.~Yao, and Q.~Wu, ``Application of switched system theory
  in power system stability,'' in {\em 2014 49th International Universities
  Power Engineering Conference (UPEC)}, pp.~1--6, 2014.

\bibitem{APP4}
Y.~Shi and X.-M. Sun, ``Bumpless transfer control for switched linear systems
  and its appliaction to aero-engines,'' {\em IEEE Transactions on Circuits and
  Systeems I: Regular Papers}, vol.~68, no.~5, pp.~2171--2182, 2021.

\bibitem{NSS}
L.~Li, M.~Ju, D.~Ma, and T.~Li, ``Event-triggered multisource bumpless transfer
  control for networked switched systems with almost output regulation against
  switching deception attacks,'' {\em IEEE Transactions on Systems, Man, and
  Cybernetics: Systems}, vol.~53, no.~7, pp.~3956--3965, 2023.

\bibitem{Q_S_rui}
R.~Zhao, Z.~Zuo, and Y.~Wang, ``Event-triggered control for networked switched
  systems with quantization,'' {\em IEEE Transactions on Systems, Man, and
  Cybernetics: Systems}, vol.~52, no.~10, pp.~6120--6128, 2022.

\bibitem{quan_2014}
D.~Liberzon, ``Finite data-rate feedback stabilization of switched and hybrid
  linear systems,'' {\em Automatica}, vol.~50, no.~2, pp.~409--420, 2014.

\bibitem{quan_2018}
G.~Yang and D.~Liberzon, ``Feedback stabilization of switched linear systems
  with unknown disturbances under data-rate constraints,'' {\em IEEE
  Transactions on Automatic Control}, vol.~63, no.~7, pp.~2107--2122, 2018.

\bibitem{R2-1}
W.~He, F.~Qian, Q.-L. Han, and G.~Chen, ``Almost sure stability of nonlinear
  systems under random and impulsive sequential attacks,'' {\em IEEE
  Transactions on Automatic Control}, vol.~65, no.~9, pp.~3879--3886, 2020.

\bibitem{R2-2}
W.~He and Z.~Mo, ``Secure event-triggered consensus control of linear
  multiagent systems subject to sequential scaling attacks,'' {\em IEEE
  Transactions on Cybernetics}, vol.~52, no.~10, pp.~10314--10327, 2022.

\bibitem{etDoS-Ass}
C.~De~Persis and P.~Tesi, ``Input-to-state stabilizing control under
  denial-of-service,'' {\em IEEE Transactions on Automatic Control}, vol.~60,
  no.~11, pp.~2930--2944, 2015.

\bibitem{9903320}
W.~Liu, J.~Sun, G.~Wang, F.~Bullo, and J.~Chen, ``Data-driven resilient
  predictive control under denial-of-service,'' {\em IEEE Transactions on
  Automatic Control}, vol.~68, no.~8, pp.~4722--4737, 2023.

\bibitem{Shi2022}
M.~Shi, S.~Feng, and H.~Ishii, ``Quantized state feedback stabilization of
  nonlinear systems under denial-of-service,'' {\em Automatica}, vol.~139,
  p.~110180, 2022.

\bibitem{output2}
M.~Wakaiki, A.~Cetinkaya, and H.~Ishii, ``Stabilization of networked control
  systems under {DoS} attacks and output quantization,'' {\em IEEE Transactions
  on Automatic Control}, vol.~65, no.~8, pp.~3560--3575, 2020.

\bibitem{Liu2022}
W.~Liu, J.~Sun, G.~Wang, F.~Bullo, and J.~Chen, ``Resilient control under
  quantization and denial-of-service: {Codesigning} a deadbeat controller and
  transimission protocol,'' {\em IEEE Transactions on Automatic Control},
  vol.~67, no.~8, pp.~3879--3891, 2022.

\bibitem{Quan_DoS}
S.~Feng, A.~Cetinkaya, H.~Ishii, P.~Tesi, and D.~Persis, Claudio, ``Networked
  control under {DoS} attacks: {Tradeoffs} between resilience and data rate,''
  {\em IEEE Transactions on Automatic Control}, vol.~66, no.~1, pp.~460--467,
  2021.

\bibitem{Rui2}
R.~Zhao, Z.~Zuo, and Y.~Wang, ``Event-triggered control for switched systems
  with denial-of-service attack,'' {\em IEEE Transactions on Automatic
  Control}, vol.~67, no.~8, pp.~4077--4090, 2022.

\bibitem{Fu2022}
J.~Fu, Y.~Qi, N.~Xing, and Y.~Li, ``A new switching law for event-triggered
  switched systems under {DoS} attacks,'' {\em Automatica}, vol.~142,
  p.~110373, 2022.

\bibitem{Wang2022}
Y.-W. Wang, Z.-H. Zeng, X.-K. Liu, and Z.-W. Liu, ``Input-to-state stability of
  switched linear systems with unstabilizable modes under {DoS} attacks,'' {\em
  Automatica}, vol.~146, p.~110607, 2022.

\bibitem{Rui1}
R.~Zhao, Z.~Zuo, Y.~Wang, and W.~Zhang, ``Active control strategy for switched
  systems against asynchrnous {DoS} attack,'' {\em Automatica}, vol.~148,
  p.~110765, 2023.

\bibitem{Rui3}
R.~Zhao, Z.~Zuo, Y.~Wang, and W.~Zhang, ``Active control strategy for disturbed
  switched systems under asynchronous {DoS} attacks,'' {\em IEEE Control
  Systems Letters}, vol.~6, pp.~2701--2706, 2022.

\bibitem{NAHS}
W.-H. Wang, Y.-W. Wang, X.-K. Liu, and Z.-W. Liu, ``Quantized control for
  networked switched systems under denial-of-service attacks via a barrier
  event-triggered mechanism,'' {\em Nonlinear Analysis: Hybrid Systems},
  vol.~49, p.~101343, 2023.

\bibitem{quant2}
W.-H. Wang, Y.-W. Wang, X.-K. Liu, and Z.-W. Liu, ``Quantization-dependent
  dynamic event-triggered control for networked switched systems under
  denial-of-service attacks,'' {\em IEEE Transactions on Systems, Man, and
  Cybernetics: Systems}, vol.~54, pp.~3822--3833, June 2024.

\bibitem{IJRNC}
J.~Yan, Y.~Xia, X.~Wang, and X.~Feng, ``Quantized stabilization of switched
  systems with partly unstabilizable subsystems and denial-of-service
  attacks,'' {\em International Journal of Robust and Nonlinear}, vol.~32,
  no.~8, pp.~4574--4593, 2022.

\bibitem{output1}
M.~Wakaiki, T.~Zanma, and K.-Z. Liu, ``Observer-based stabilization of systems
  with quantized inputs and outputs,'' {\em IEEE Transactions on Automatic
  Control}, vol.~64, no.~7, pp.~2929--2936, 2019.

\bibitem{R1-1}
X.~Liu and T.~Chen, ``Cluster synchronization in directed networks via
  intermittent pinning control,'' {\em IEEE Transactions on Neural Networks},
  vol.~22, no.~7, pp.~1009--1020, 2011.

\bibitem{DoS_state_switching_2024}
D.~Li, Z.~Zuo, Y.~Wang, and M.~Fu, ``Stability of networked switched systems in
  the presence of denial‐of‐service attacks,'' {\em International Journal
  of Robust and Nonlinear Control}, vol.~33, no.~2, pp.~1267--1283, 2023.

\bibitem{R1-2}
B.~Liu and T.~Chen, ``Consensus in networks of multiagents with cooperation and
  competition via stochastically switching topologies,'' {\em IEEE Transactions
  on Neural Networks}, vol.~19, no.~11, pp.~1967--1973, 2008.

\bibitem{R1-3}
B.~Liu, W.~Lu, and T.~Chen, ``Consensus in networks of multiagents with
  switching topologies modeled as adapted stochastic processes,'' {\em SIAM
  Journal on Control and Optimization}, vol.~49, no.~1, pp.~227--253, 2011.

\bibitem{R1-4}
B.~Liu, W.~Lu, and T.~Chen, ``Global almost sure self-synchronization of
  hopfield neural networks with randomly switching connections,'' {\em Neural
  Networks}, vol.~24, no.~3, pp.~305--310, 2011.

\bibitem{R1-5}
G.~Wen, W.~Yu, G.~Hu, J.~Cao, and X.~Yu, ``Pinning synchronization of directed
  networks with switching topologies: A multiple lyapunov functions approach,''
  {\em IEEE Transactions on Neural Networks and Learning Systems}, vol.~26,
  no.~12, pp.~3239--3250, 2015.

\end{thebibliography}

\begin{IEEEbiography}{Rui Zhao}
	received the B.Eng. degree in automation from Tianjin University, Tianjin, China, in 2018, and the Ph.D. degree in control science and engineering from Tianjin University in 2025. From January 2024 to December 2024, he worked as a visiting Ph.D. student in the Department of Mechanical Engineering, University of Victoria, Victoria BC, Canada. 	
	She is currently a Postdoctoral Fellow with the Department of Electrical Engineering, City University of Hong Kong, Hong Kong SAR, China. 
	Her research interests include event-triggered control, switched systems,  network security and model predictive control.
\end{IEEEbiography}

\begin{IEEEbiography}{Zhiqiang~Zuo} (M'04--SM'18) received the M.S. degree in
	control theory and control engineering in 2001 from
	Yanshan University and the Ph.D. degree in control
	theory in 2004 from Peking University, China. In
	2004, he joined the School of Electrical and Information Engineering, Tianjin University, where he is a
	full professor. From 2008 to 2010, he was a Research
	Fellow in the Department of Mathematics, City
	University of Hong Kong. From 2013 to 2014, he
	was a visiting scholar at the University of California,
	Riverside. His research interests include nonlinear
	control, robust control, multi-agent systems and cyber-physical systems.
	
	Dr. Zuo is an associate editor of the Journal of the Franklin Institute (Elsevier).
\end{IEEEbiography}
\begin{IEEEbiography}{Yijing~Wang} received her M.S. degree in control
	theory and control engineering from Yanshan University
	and the Ph.D. degree in control theory
	from Peking University, China, in 2000 and 2004,
	respectively. In
	2004, she joined the School of Electrical and Information Engineering, Tianjin University, where she is a
	full professor.
	Her research interests are analysis and control of
	switched/hybrid systems, and robust control.
\end{IEEEbiography}
\begin{IEEEbiography}{Wentao~Zhang}
	(Member, IEEE) received Ph.D. in Control Science and Engineering from Tianjin University, China, in 2020.
	From Sept. 2020 to Mar. 2022, he was a research assistant/associate with the School of Electrical and Information Engineering, Tianjin University. Since Mar. 2022, he has been a Research Fellow with the School of Electrical and Electronic Engineering, Nanyang Technological University, Singapore. His current research interests include distributed control, networked control systems, control with limited information and their dial-a-ride application. He was a recipient of the Best Paper Award Finalist in IEEE 8th Annual International Conference on CYBER Technology in Automation, Control, and Intelligent Systems, 2018.
\end{IEEEbiography}
\begin{IEEEbiography}{Yang~Shi}
	(SM’09–F’17) received the Ph.D. degree in electrical and computer engineering from the University of Alberta, Edmonton, AB, Canada, in 2005. From 2005 to 2009, he was an Assistant Professor and Associate Professor in the Department of Mechanical Engineering, University of Saskatchewan, Saskatoon, SK, Canada. In 2009, he joined the University of Victoria, and now he is a Professor in the Department of Mechanical Engineering, University of Victoria, Victoria, BC, Canada. His current research interests include networked and distributed systems, model predictive control (MPC), cyber-physical systems (CPS), robotics and mechatronics, control of autonomous systems (AUV and UAV), and energy system applications.
	
	Dr. Shi received the University of Saskatchewan Student Union Teaching Excellence Award in 2007, and the Faculty of Engineering Teaching Excellence Award in 2012 at the University of Victoria (UVic). He is the recipient of the JSPS Invitation Fellowship (short-term) in 2013, the UVic Craigdarroch Silver Medal for Excellence in Research in 2015, the 2017 IEEE Transactions on Fuzzy Systems Outstanding Paper Award, the Humboldt Research Fellowship for Experienced Researchers in 2018. Currently he serves as the Chair of IEEE IES Technical Committee on Industrial Cyber-Physical Systems, and Co-Editor-in-Chief for IEEE Transactions on Industrial Electronics. He also serves as Associate Editor for Automatica, IEEE Transactions on Control Systems Technology, etc. He is General Chair of the 2019 International Symposium on Industrial Electronics (ISIE) and the 2021 International Conference on Industrial Cyber-Physical Systems (ICPS).
	
	He is a Fellow of IEEE, ASME, CSME, and Engineering Institute of Canada (EIC), and a registered Professional Engineer in British Columbia, Canada.
\end{IEEEbiography}
\end{document}